%% file: main.tex
\documentclass{article}
\usepackage{graphicx}
\usepackage{epsfig,subfigure}
\usepackage{amssymb}
\usepackage{amsthm}
\usepackage{amsfonts}
\usepackage{amsmath}
\usepackage{hyperref}
\usepackage{multicol,multirow}
\usepackage{algorithm,algcompatible}
\usepackage[margin=1in]{geometry}
\usepackage{comment}
\usepackage{enumitem}
\usepackage{float}
\usepackage{booktabs}
\newtheorem{theorem}{Theorem}
\newtheorem{lemma}[theorem]{Lemma}
\newtheorem{corollary}[theorem]{Corollary}
\newtheorem{proposition}[theorem]{Proposition}
\newtheorem{definition}{Definition}

\newtheorem{example}{Example}
\newtheorem{assumption}{Assumption}

\usepackage{graphicx}
\usepackage{placeins}
\usepackage{natbib}

\usepackage{array}
\newcolumntype{C}{>{\centering\arraybackslash}m{0.12\columnwidth}}
\usepackage{pifont}
\newcommand{\cmark}{\ding{51}}%
\newcommand{\xmark}{\ding{55}}%

\newcommand{\Proj}{\mathcal{P}}
\newcommand{\C}{\mathcal{C}}

\newcommand{\Y}{\mathbf{Y}}

\def\bbE{\mathbb{E}}
\def\bbP{\mathbb{P}}
\def\bbR{\mathbb{R}}
\def\cE{\mathcal{E}}
\def\cL{\mathcal{L}}
\def\cN{\mathcal{N}}

\def\cP{\mathcal{P}}
\def\cI{\mathcal{I}}
\def\cS{\mathcal{S}}

\def\var{\mathrm{Var}}
\def\cov{\mathrm{Cov}}
\def\tr{\mathrm{tr}}
\def\pto{\overset{p}{\to}}

\DeclareMathOperator*{\argmax}{argmax} 
\DeclareMathOperator*{\argmin}{argmin} 
\usepackage{xcolor}

\usepackage{tikz}
\usetikzlibrary{positioning}
\usetikzlibrary{arrows.meta,calc,decorations.markings,math,arrows.meta}

\usepackage{authblk}

\title{Causality-oriented robustness: \\exploiting general noise interventions} 
\author[$\star$]{Xinwei Shen}
\author[$\star$]{Peter B\"uhlmann}
\author[$\dag$]{Armeen Taeb}
\affil[$\star$]{Seminar for Statistics, ETH Z\"urich}
\affil[$\dag$]{Department of Statistics, University of Washington}
\date{}

\begin{document}

\maketitle

\begin{abstract}
\input{parts/abstract.tex}

\medbreak
\noindent Keywords: distribution shifts, robust prediction, interventional data, structural causal models, invariance
\end{abstract}

\input{parts/body}

\section*{Acknowledgments}
X.\ Shen's research was supported by the ETH AI Center. 
P.\ B\"uhlmann received funding from the European Research Council (ERC) under the European Union's Horizon 2020 research and innovation programme (grant agreement No. 786461). Armeen Taeb is supported by NSF DMS-2413074 and by the Royalty Research Fund at the University of
Washington. 

\nocite{*} 
\bibliography{ref.bib}
\bibliographystyle{apalike}

\newpage
\input{parts/appendix.tex}

\end{document}

%% file: parts/abstract.tex
Since distribution shifts are common in real-world applications, there is a pressing need to develop prediction models that are robust against such shifts. Existing frameworks, such as empirical risk minimization or distributionally robust optimization, either lack generalizability for unseen distributions or rely on postulated distance measures. Alternatively, causality offers a data-driven and structural perspective to robust predictions. However, the assumptions necessary for causal inference can be overly stringent, and the robustness offered by such causal models often lacks flexibility. In this paper, we focus on causality-oriented robustness and propose Distributional Robustness via Invariant Gradients (DRIG), a method that exploits general noise interventions in training data for robust predictions against unseen interventions, and naturally interpolates between in-distribution prediction and causality. In a linear setting, we prove that DRIG yields predictions that are robust among a data-dependent class of distribution shifts. Furthermore, we show that our framework includes anchor regression as a special case, and that it yields prediction models that protect against more diverse perturbations. We establish finite-sample results and extend our approach to semi-supervised domain adaptation to further improve prediction performance. Finally, we empirically validate our methods on synthetic simulations and on single-cell and intensive health care datasets.

%% file: parts/body.tex
\section{Introduction}
Statistical and machine learning models are often deployed on test data distributed differently from the training data. Such scenarios pose a major challenge for traditional learning methods that typically assume the test distribution is sufficiently close to the training distribution. For example, while empirical risk minimization (ERM) achieves minimal prediction error when the test and training data are identically distributed, the performance of this widely used prediction paradigm deteriorates significantly when the test distribution differs substantially from the training distribution \citep{Geirhos2020ShortcutLI,Sagawa2021ExtendingTW}. 

Distributional robustness \citep{ben1998robust,ben2006analysis,sinha2017certifiable,meinshausen2018causality} is an appealing framework for assessing how prediction models perform under distributional shifts. As the precise manner in which the test and training distributions differ is typically unknown, distributional robustness aims to identify a predictive model that performs favorably over a class of plausible test distributions. Formally, suppose $X$ is a set of covariates or predictors and $Y$ is a response or target variable of interest. Let $\theta$ be the parameter of a prediction model from $X$ to $Y$. Then, distributional robustness is formulated as the following minimax optimization problem
\begin{equation}\label{eq:dro}
	\min_\theta\sup_{P\in\cP}\bbE_{P}[\ell(X,Y;\theta)].
\end{equation}
Here, $\ell$ is a given loss function and $\cP$ is a class containing plausible test distributions. 

The choice of the set of distributions $\cP$ is central to the distributional robustness framework \eqref{eq:dro}. A common perspective taken by the literature in distributionally robust optimization (DRO) is to define $\cP$ based on a pre-specified distance measure, e.g., $\cP= \{P:D(P,P_0)\le\rho\}$, where $P_0$ is the training distribution, $D(\cdot,\cdot)$ is, e.g., the $f$-divergence, and $\rho$ is the parameter that controls the strength of potential distribution shifts in the test data relative to the training data \citep{sinha2017certifiable,duchi2021learning}. DRO thus learns a prediction model that is robust against distributional shifts in a pre-specified ``ball" of radius $\rho$ around the training distribution. However, protecting against all distributions in a ball ignores structural information about the distributional shifts and can yield overly conservative predictions, especially in high dimensions. As an illustration, consider Figure~\ref{fig:directional_shift}(left), where the shifts from the training to the test distributions are in a certain ``direction". To achieve robustness with respect to test environment $3$, DRO would require a large radius $\rho$ (as environment $0$ and $3$ are far apart) and thus protects against many more distributions than necessary. As we elaborate throughout the paper, a causal perspective provides an approach to attain robustness against a distribution class $\cP$ driven from the heterogeneity in the observed data and exploits structural relations among the training and test distributions. 

\begin{figure}
	\centering
	\begin{tabular}{cc}
\includegraphics[width=0.24\textwidth]{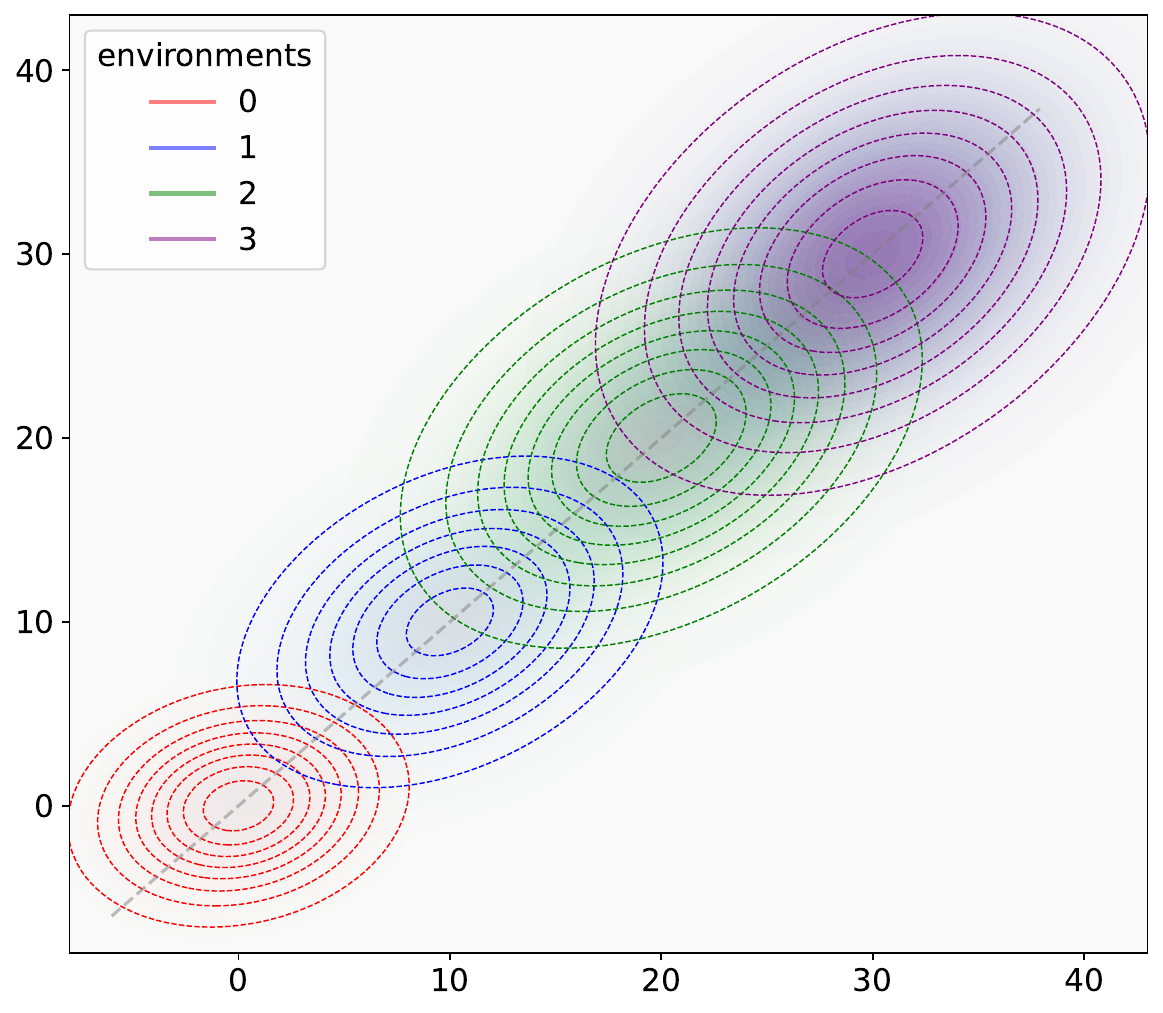} & 
\includegraphics[page=1, clip, trim=14cm 14cm 14cm 14cm, width=0.7\textwidth]{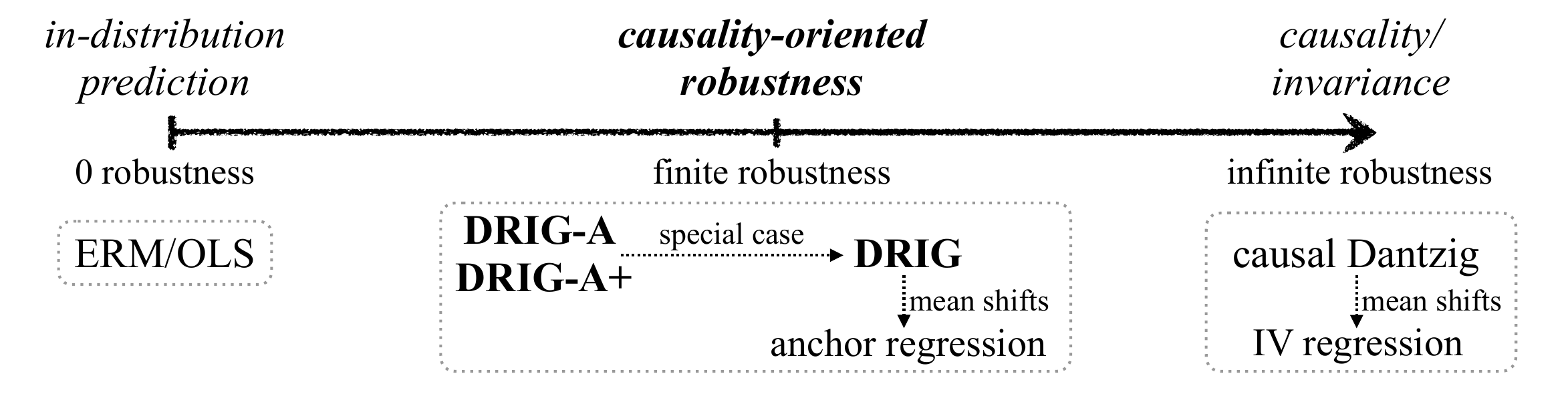}
\end{tabular}
	\caption{{\small{\textbf{{(left)}}: An example of structural shifts: environment 0 represents training environment and environments 1-3 represent possible test environments, where the shift between the training and test distributions is in a particular ``direction" (here, the support of each distribution is the same); \textbf{(right)}: Causality-oriented robustness: a trade-off between in-distribution prediction and causality using our method DRIG that exploits general additive interventions in the data. DRIG encompasses anchor regression as a special case with mean shifts only. Our extended proposals DRIG-A and DRIG-A+ provides a more flexible robustness framework.}}}\label{fig:directional_shift}
\end{figure}

In many real-world data, the distribution of variables $(X,Y)$ can be effectively described by a causal mechanism \citep{spirtes2000causation,pearl2009causality}. The virtue of causal modeling is that distributional shifts (and consequently the distribution class $\cP$) could be naturally formalized as interventions or perturbations to the observed or latent variables. This perspective, known as \emph{causality-oriented robustness}~\citep{buhlmann2020invariance,meinshausen2018causality,rothenhausler2021anchor}, enables us to model distribution shifts in a more structured and data-dependent manner than those considered in DRO. In such a framework, a natural prediction model to consider 
is one involving merely the causal parents of $Y$, known as a causal prediction model. Indeed, the causal prediction model performs equally well under any interventions on the covariates~\citep{haavelmo1943statistical,buhlmann2020invariance}, thus providing certain robustness guarantees even when the interventions or shifts are arbitrarily strong. 

Nevertheless, identifying the causal parents and estimating the causal effects are often ambitious tasks that rely on relatively strong assumptions about the data distribution. 
For example, instrumental variable (IV) regression~\citep{bowden1990instrumental,angrist1996identification,imbens2015causal} is a popular approach to estimate causal effects in the presence of latent confounding. IV regression relies on the assumption that the instrumental variables are independent of the latent confounders and do not directly affect the response variable, known as the valid IV condition. When the instrumental variables are categorical, for example when they encode the different interventional environments, the valid IV condition requires that the interventions happen only on the covariates and the number of environments must exceed the number of covariates. 
However, in a wide range of real-world prediction scenarios, such identifiability conditions are rarely fulfilled. This inspires the pursuit of an alternative solution that relies on weaker assumptions and yet remains effective for producing robust predictions, which is the essence of causality-oriented robustness. In particular, causality-oriented robustness does not require the full knowledge of the underlying causal mechanism, but directly aims for robust prediction by leveraging insights from causality. 

Even when the underlying causal structure can be identified from data, the resulting prediction model may not be desirable in terms of robust prediction. In particular, the causal prediction model protects against arbitrarily strong interventions, and is thus 
a conservative approach with subpar predictive performance on moderately perturbed data. 

Our goal is to use a causal framework to learn distributionally robust prediction models against a \emph{finite and learned uncertainty set} without knowledge of the underlying causal structure. We leverage heterogeneous training data from multiple environments with \emph{general noise interventions} to learn sets that are much more adaptive than standard DRO methods, being larger in some directions and smaller in other directions.





\subsection{Our contributions}
\label{sec:our_contributions}
We propose in Section~\ref{sec:method} our method \emph{distributional robustness via invariant gradients (DRIG)}, a regularized ERM formulation, where the regularization term is inspired by a gradient invariance condition across the environments. We show that DRIG is convex under certain natural settings, and that anchor regression \citep{rothenhausler2021anchor} is a special case of DRIG. Finite sample guarantees are also established. In Section~\ref{sec:robust}, we present robustness guarantees of DRIG under a linear structural causal model. We show that DRIG's prediction models achieve finite robustness against interventions whose strength is controlled via a regularization parameter and whose directions depend on the heterogeneity in the training data. Furthermore, we prove that as long as there are some shifts in the variances (i.e., the interventions given each environment are random variables), DRIG leads to robustness against perturbations in strictly (and often much) more directions than those protected by anchor regression; in fact, the DRIG robustness holds for general noise interventions, whereas anchor regression assumes additive noise interventions. When there are only mean shifts (i.e., the interventions are deterministic given each environment), DRIG is identical to anchor regression. We also discuss how DRIG with regularization parameter tending to infinity, which attains robustness against infinitely strong perturbations, leads to causality under more restrictive assumptions, highlighting the essence of causality-oriented robustness. 

In Section \ref{sec:adaptive}, we explore extensions of DRIG to semi-supervised settings. In particular, when we have access to samples from a test distribution of interest, we develop the extension DRIG-A that selects hyperparameters to adapt to the test distribution. In settings where we have access to a large set of unlabeled samples and a small set of labeled data from the test distribution, we present DRIG-A+. This method extends the DRIG formulation to have a matrix of hyperparameters, where the hyperparameters allow for much more flexible robustness; these hyperparameters are again chosen from the semi-supervised data. {We theoretically demonstrate that DRIG-A+ yields smaller test error (in population) as compared to the ordinary least squares (OLS) estimator obtained from the semi-supervised samples}. 

Finally, we conduct real-data analysis on single-cell and intensive health care data in \ref{sec:real_data_experiments}. A visual summary of our methodological contributions is presented in Figure~\ref{fig:directional_shift}(right), highlighting how DRIG (and its extensions) interpolate between in-distribution prediction and causality by exploiting heterogeneity in the training data.

\subsection{Related work}
There is a growing literature in exploiting heterogeneous data for causal inference \citep{peters2016causal,Ghassami,Rothenhausler2017CausalDF,huang2020causal,long2022generalized}, stablized variable selection~\citep{Pfister2019StabilizingVS,fan2023environment}, as well as robust predictions \citep{meinshausen15maximin,Magliacane2017DomainAB,sagawa2019distributionally,rothenhausler2021anchor,christiansen2021causal,RojasCarulla2015InvariantMF}. 
In a similar spirit, another line of work aims for out-of-distribution prediction from multi-environment data based on invariance notions \citep{arjovsky2019invariant,koyama2020out,krueger2021out,shi2021gradient,Ram2021FishrIG}; we discuss the connections to them in Appendix~\ref{app:invariance}.  
Most of these methods do not provide guarantees for finite robustness which is often more relevant to applications. Anchor regression~\citep{rothenhausler2021anchor} is a prominent method that can provably achieve finite robustness. In anchor regression, interventions are assumed to be additive and only affect the conditional means of the variables. Thus, the method is designed to exploit heterogeneity in this form, leading to robustness against additive mean shifts in the test data. In contrast, we consider a more general setting with general noise interventions. This flexibility allows us to exploit richer heterogeneity within the training data, which results in robustness against potentially much more perturbations and causal identification with data collected from fewer environments. 

In a concurrent work to our manuscript, \citet{kennerberg2023convergence} extend the framework of \citet{kania2022causal} to achieve finite robustness from multi-environment data. While the method proposed in \citet{kennerberg2023convergence} is similar to DRIG, our work differs in substantive ways. First, in our modeling framework, we allow for and exploit interventions on the response variable and on any potential latent confounders, which is more realistic and results in robustness against more general interventions. By contrast, in \citet{kennerberg2023convergence}, the environments arising from interventions appear in a much more restrictive way, excluding the above interventions. Second, we present precise connections with anchor regression -- in particular, we show how anchor regression is a special case of our method where only additive mean shifts are exploited and how we are able to obtain strictly more robust predictions; we also present an extension where we incorporate continuous anchor variables in our estimator. 
Third, we study (approximate) causal identifiability results in general settings, whereas \citet{kennerberg2023convergence} only consider the restrictive setting where there are no interventions on the response variable or on the latent variables. Finally, we propose adaptive extensions for more flexible robustness, often yielding substantially better prediction than other methods, as validated by both theoretical and numerical results.

\section{Our method DRIG}\label{sec:method}

{\subsection{Setup: Linear structural causal models}}
\begin{figure}
\centering
\begin{tabular}{cc}
\begin{tikzpicture}[every node/.style={scale=0.8},node distance={15mm},main/.style = {draw, circle, minimum size=0.5cm}]
	\node[main,inner sep=1.5pt,dotted](h) {$H$};
	\node[main,inner sep=1.5pt](x)[below left of=h]  {$X$};
	\node[main,inner sep=1.5pt](y)[below right of=h] {$Y$};
	\node[main,inner sep=1.5pt,red](e)[above left of=x] {$E$};
	\draw[<->] (x) -- (y);
	\draw[<->] (h) -- (x);
	\draw[<->] (h) -- (y);
	\draw[->,red] (e) -- (x);
	\draw[->,red] (e) -- (y);
	\draw[->,red] (e) -- (h);
\end{tikzpicture}
&~~~~~~~~~
\begin{tikzpicture}[every node/.style={scale=0.8},node distance={15mm},main/.style = {draw, circle, minimum size=0.5cm}]
	\node[main,inner sep=1.5pt,dotted](h) {$H$};
	\node[main,inner sep=1.5pt](x)[below left of=h]  {$X$};
	\node[main,inner sep=1.5pt](y)[below right of=h] {$Y$};
	\node[main,inner sep=1.5pt,red](e)[above left = 5.4mm and 3mm of x] {$E$};
	\node[main,inner sep=1.5pt,blue](a)[above right = 5.4mm and 3mm of y] {$A$};
	\draw[<->] (x) -- (y);
	\draw[<->] (h) -- (x);
	\draw[<->] (h) -- (y);
	\draw[->,red] (e) -- (x);
	\draw[->,red] (e) -- (y);
	\draw[->,red] (e) -- (h);
	\draw[->,blue] (a) -- (x);
	\draw[->,blue] (a) -- (y);
	\draw[->,blue] (a) -- (h);
    \draw[<->] (e) to [out=45,in=135,looseness=0.5] (a);
\end{tikzpicture}\vspace{-0.3in}\\\vspace{-0.05in}
\end{tabular}
\caption{\small Graphical models among covariates $X$, response variable $Y$, and latent variables $H$ ($X$ and $H$ may be multivariate): \textbf{(left)}: interventions $E$ on all components, \textbf{(right)}: discrete interventions $E$ and continuous interventions $A$ on all components. All these structures are allowed for DRIG.}
\label{fig:graph}
\end{figure}
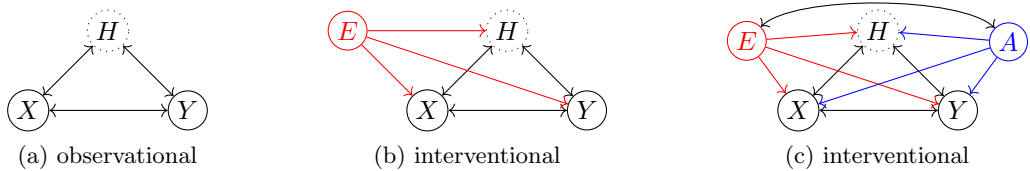
We suppose we have access to observations of variables under different environments, such as experimental conditions in which some of the variables may have been manipulated, that is, received interventions. To represent this setting, we consider covariates $X \in \mathbb{R}^p$ and a response variable $Y \in\mathbb{R}$. The interventions on these variables are generated randomly from a discrete random variable $E$ taking on values in the set $\mathcal{E}$; each $e \in \mathcal{E}$ represents a different environment that generates the random vectors $(X^e,Y^e)$. We posit that for every $e\in\mathcal{E}$, the random variables $(X^e,Y^e)$ satisfy the following linear structural causal model (SCM)
\begin{equation}\label{eq:scm_train}
    \begin{pmatrix}
    X^e \\ Y^e
    \end{pmatrix}
    =
    B^\star
    \begin{pmatrix}
    X^e \\ Y^e
    \end{pmatrix} 
    + \varepsilon^e.
\end{equation}

Here, $B^\star\in\bbR^{(p+1)\times (p+1)}$ is the adjacency matrix encoding the causal relations, namely $B^\star_{ij} \neq 0$ if $Z^e_j$ is a parent of $Z^e_i$ in the graph among observed variables $Z^e = (X^e,Y^e)$. The SCM \eqref{eq:scm_train} thus assumes that the causal structure among the observed variables does not change across  $e\in\mathcal{E}$. The row vector $B^\star_{p+1,1:p}$ encodes the (observable) causal parents of the response variable and the magnitude of their effects. Throughout, we will use 
$$b^\star := B^\star_{p+1,1:p}$$
to denote this vector and call it the \emph{causal parameter}. Further, $\varepsilon^e$ is a random vector with a bounded second moment, with potentially dependent components to account for latent confounding and dependencies in the interventions generated by $E$. We assume that the matrix $I-B^\star$ is invertible, which is guaranteed if the subgraph consisting of only the observed variables is acyclic. For any $j\in [p+1]$, the distribution of $\varepsilon_j^e$ is allowed to vary across $e \in \mathcal{E}$; this variation may result from a direct intervention on the variable $Z_j$  or an intervention on the latent variables, which are manifested through $\varepsilon_j^e$. An equal distribution of $\varepsilon_j^e$ for all $e \in \mathcal{E}$ indicates that $Z_j$ does not receive a direct intervention or an indirect intervention through a latent variable that affects $Z_j$, although its marginal distribution could still be changed due to interventions on its ancestors. Figure~\ref{fig:graph}(left) presents the graphical perspective of model \eqref{eq:scm_train}; $E$ is exogenous and cannot be descendants of $(X,Y)$ and any latent variables. Throughout, we assume the following on the noise variables $\varepsilon^e$.
\begin{assumption} $\exists$ an environment $0\in\mathcal{E}$ where $\mathbb{E}[\varepsilon^{0}{\varepsilon^{0}}^\top] \preceq \mathbb{E}[\varepsilon^{e}{\varepsilon^{e}}^\top]$ for every $e\in\mathcal{E}$.
\label{ass:strength}
\end{assumption}
Here, for two positive semidefinite matrices $A$ and $B$, we write $A \preceq B$ if and only if $B-A$ is positive semidefinite. Assumption~\ref{ass:strength} ensures that there exists an `observational' environment $0 \in \mathcal{E}$ with `smaller' interventions (as measured by the second moments) than the other environments. Letting $Z^e = (X^e,Y^e)$, this assumption can be expressed in terms of observed Gram matrices, namely: $\mathbb{E}[Z^{0}{Z^{0}}^\top] \preceq \mathbb{E}[Z^{e}{Z^{e}}^\top]$ for all $e\in\mathcal{E}$. An observational assumption is a common condition in the causal inference literature. Nevertheless, in Appendix~\ref{Sec:drig_without_observational}, we relax this condition while still guaranteeing that our estimator produces distributionally robust prediction models that interpolate between the OLS solution and the causal parameter. In short, our relaxed assumption ensures that the set of environments $\mathcal{E}$ can be divided into two: $\mathcal{E}_\mathrm{small}$ and $\mathcal{E}\setminus\mathcal{E}_\mathrm{small}$ where the interventions in $\mathcal{E}\setminus\mathcal{E}_\mathrm{small}$ are sufficiently stronger than those in $\mathcal{E}_\mathrm{small}$; see Appendix~\ref{Sec:drig_without_observational} for more details.


Our training data consists of $(X^e,Y^e)$ across all environments $e \in \mathcal{E}$. We consider out-of-distribution prediction on a test distribution generated according to the following SCM:
\begin{equation}\label{eq:scm_test}
    \begin{pmatrix}
    X^v \\ Y^v
    \end{pmatrix}
    =
    B^\star
    \begin{pmatrix}
    X^v \\ Y^v
    \end{pmatrix} 
    + v,
\end{equation}
Notably, the distribution of $v$ in the test data may follow a different distribution than $\{\varepsilon^e\}_{e \in \mathcal{E}}$ in the training data. Our objective is to develop a procedure that uses only the training data to learn a prediction model that performs well on test data generated according to \eqref{eq:scm_test}.

{\subsection{Our formulation}}
\label{sec:our_formulation}
We introduce our method DRIG at the population level; the empirical analog is described shortly. Specifically, suppose the random variables $(X^e,Y^e)$ are generated according to the SCM \eqref{eq:scm_train} for environments $e \in \mathcal{E}$. Given a scalar $\gamma \geq 0$, population DRIG minimizes
\begin{equation}
	b^\mathrm{opt}_\gamma = \argmin_b \mathcal{L}_\gamma(b),\quad \text{ where }
   \label{eqn:DRIG} 
\end{equation}
\begin{equation}\label{eq:obj}
	\cL_\gamma(b):= \min_{e\in\mathcal{E}} \mathbb{E}[\ell(X^{e},Y^{e}; b)] + \gamma\sum_{e \in \mathcal{E}}\omega^e \left(\mathbb{E}[\ell(X^{e},Y^e;b)] - \min_{e\in\mathcal{E}}\mathbb{E}[\ell(X^{e},Y^{e};b)]\right),
\end{equation}
and $\ell(x,y; b) := (y - b^\top x)^2$ is the squared loss. Here, $\omega^e \geq 0$ are weights that weigh the impact of each environment on the DRIG objective with $\sum_{e\in\mathcal{E}}\omega^e = 1$. Without any prior information on the test distribution or access to some labeled data from the test set, we suggest choosing the weight to be uniform across the environments, i.e. $\omega^e = 1/|{\mathcal{E}}|$ for each $e$, or in the finite sample version of DRIG (discussed shortly),  set them based on available data size in each environment; see Sections~\ref{sec:robust_gua} and~\ref{sec:adaptive} for additional discussions on $\omega^e$.

The risk $\cL_\gamma(b)$ is the squared loss in the environment with the smallest loss summed with the weighted average difference in the squared losses between every environment $e \in \mathcal{E}$ and the environment with the smallest loss; the regularization parameter $\gamma$ controls how much the latter component is penalized. By definition, the regularization term is non-negative. For $\gamma = 0$, DRIG is OLS on the environment with the smallest loss, named the observational OLS, as it is the observational setting under Assumption~\ref{ass:strength}; for $\gamma = 1$, DRIG coincides with the OLS solution on the pooled data, called the pooled OLS; for $\gamma\to\infty$, when $|\cE|=2$, we show in Appendix~\ref{app:causal_dantzig} that DRIG converges to the causal Dantzig estimator \citep{Rothenhausler2017CausalDF} which recovers the causal parameter under some conditions. To understand the intuition behind DRIG, we introduce the notion of gradient invariance. 
\begin{definition}[Gradient invariance] \label{def:gi}
A regression parameter $b$ is said to satisfy the gradient invariance condition if $\sum_{e \in \mathcal{E}}\omega^e\nabla_{b}\mathbb{E}[\ell(X^e,Y^e;b)] =\nabla_{b}\min_e\mathbb{E}[\ell(X^{e},Y^{e};b)]$ \footnote{Here, $\min_e\mathbb{E}[\ell(X^{e},Y^{e};b)]$ is almost everywhere differentiable. It is non-differentiable for $b$ where $\argmin_e\mathbb{E}[\ell(X^{e},Y^{e};b)]$ is not unique; then, one can use sub-differential of $\min_e\mathbb{E}[\ell(X^{e},Y^{e};b)]$ instead.}, that is the weighted average gradient of the loss function across the environments is the same as the gradient in the environment with the smallest loss.
\end{definition}
In the limit of $\gamma \to \infty$ and under some mild conditions, we show in Theorem~\ref{thm:solution_infty}  
that the DRIG solution $b^\mathrm{opt}_\gamma$ satisfies Definition~\ref{def:gi}.
We provide a thorough discussion on invariance in Appendix~\ref{app:invariance}, including the gradient invariance and other existing notions such as invariance of the conditional distribution, the conditional mean, or the risk. We highlight that Definition~\ref{def:gi} can be fulfilled by the causal parameter under more general cases, especially with the presence of latent confounders and interventions on $Y$ or on the latent variables.

In summary, by encouraging invariant gradients across the environments (to the extent controlled by the parameter $\gamma$), DRIG naturally interpolates between the ordinary least squares solution and the causal parameter. As we will discuss in Section~\ref{sec:robust}, the main benefit of the proposed DRIG estimator is robust prediction on test environments that are potentially far from the training environments, where the degree to which the test and training environments can differ is controlled by the parameter $\gamma$. 

\textbf{Finite-sample DRIG:} For each environment $e\in\mathcal{E}$, let $(X^e_1,Y^e_1),\dots,(X^e_{n_e},Y^e_{n_e})$ be i.i.d.\ samples of the random pair $(X^e,Y^e)$ distributed according to model \eqref{eq:scm_train}. Then, the finite-sample analog of the DRIG is given by $\hat{b}_\gamma \in \argmin_b \hat{\mathcal{L}}_\gamma(b)$,
where
\begin{equation}
\hat{\mathcal{L}}_\gamma(b) = \min_{e\in\mathcal{E}}\hat{\mathbb{E}}[\ell(X^e,Y^e;b)]^2 + \gamma\sum_{e \in\mathcal{E}}\omega^e \left(\hat{\mathbb{E}}[\ell(X^e,Y^e;b)]^2 - \min_{e\in\mathcal{E}}\hat{\mathbb{E}}[\ell(X^e,Y^e;b)]^2\right).
\label{eqn:finite_drig_objective}
\end{equation}
Here, $\hat{\mathbb{E}}$ denotes the empirical expectation computed over samples of $(X^e,Y^e)$ for every environment $e$, i.e., $\hat{\mathbb{E}}[\ell(X^e,Y^e;b)]^2 = \frac{1}{n_e}\sum_{i=1}^{n_e}(\ell(X^e_i,Y^e_i;b))^2$. We provide finite-sample consistency guarantees of the estimator \eqref{eqn:finite_drig_objective} in Appendix~\ref{sec:finite_sample}.

\textbf{Nonlinear DRIG:} In Section~\ref{sec:nonlinear}, we explore the extension of DRIG to nonlinear settings where we allow $\ell$ to be nonlinear in the objective \eqref{eq:obj}.

\subsection{Connections to anchor regression} 
\citet{rothenhausler2021anchor} posit the following linear SCM:
\begin{equation}
Z = \tilde{B}^\star{Z} + \varepsilon + MA,
\label{eqn:model_anchor}
\end{equation}
Here, $Z = (X,Y,H)$; $A$ are observed \emph{anchor} variables that are independent of the noise $\varepsilon$; and $H$ are latent variables. From a graphical perspective, $A$ are exogenous and cannot be descendant of any of the variables $(X,Y,H)$. Under this model, anchor regression minimizes
\begin{equation}
\cL_{\rm anchor,\gamma}(b) := \bbE[((I-P_A)(Y-b^\top X))^2] + \gamma\bbE[(P_A(Y-b^\top X))^2],
\label{eqn:method_anchor}
\end{equation}
with $P_A$ denoting the $L_2$-projection on the linear span from the components of $A$. 

When the anchors $A$ are discrete, our framework is a generalization of anchor regression; we further discuss in Appendix~\ref{sec:cont_anchor} how DRIG can be modified to accommodate continuous anchors as well (corresponding to Figure~\ref{fig:graph}(right)) and continue to be a generalization of anchor regression. Specifically, let $A$ take values in the set $\{a^e \in \mathbb{R}^{\text{dim}(A)}: e \in \mathcal{E}\}$. Then, setting $\varepsilon^e = \varepsilon+M{a}^e$, we conclude that the model \eqref{eqn:model_anchor} proposed in \citet{rothenhausler2021anchor} is a special case of our model \eqref{eq:scm_train} with substantial restrictions. First, in the anchor regression model, the dependence on the anchor variable and the latent confounders are restricted to be linear. Second, for different $e\in\mathcal{E}$, the noise variables $\varepsilon^e$ are restricted to be mean shifts of one another, which means that the interventions only affect the conditional mean of $(X,Y)$ given $A=a^e$. {Finally, the anchor regression model restricts the noise interventions to be additive, whereas our model {is more general}; for example, in our model, we allow for the interventions to affect the noise in a multiplicative manner, e.g., $\varepsilon^e = \varepsilon\cdot(M{a}^e)$.} 

Under model \eqref{eqn:model_anchor}, the anchor regression estimator \eqref{eqn:method_anchor} matches with the DRIG estimator \eqref{eqn:DRIG}, as formalized in the following proposition with the proof in Appendix~\ref{proof:prop_anchor_regression_ours}. 
\begin{proposition}
 Suppose the data is generated according to \eqref{eqn:model_anchor}. Let $A$ be discrete anchors taking values in the set $\{a^e \in \mathbb{R}^{\mathrm{dim}(A)}: e \in \mathcal{E}\}$. Suppose a reference environment $0 \in \mathcal{E}$ exists where $a^e=0$. Assuming that $\mathbb{P}(A = a^e) = \omega^e$, the anchor regression loss is the same as DRIG loss \eqref{eq:obj}, that is $\cL_{\rm anchor,\gamma}(b) = \cL_\gamma(b)$ for every regression parameter $b$.
 \label{prop:anchor_regression_ours}
 \end{proposition}
 This result states that under the restrictive model in \citet{rothenhausler2021anchor}, which only allows additive mean shifts, the DRIG estimator matches the one from anchor regression. However, the two estimators are different under more general interventions. DRIG is designed for the more general modeling framework \eqref{eq:scm_train} that allows for arbitrary noise interventions. Section~\ref{sec:robust} discusses how the additional flexibility of DRIG leads to more robust predictions. 

{ \subsection{Optimizing the DRIG Objective}}
 We use gradient descent to minimize the DRIG objective \eqref{eq:obj}; see Appendix~\ref{proof_convex} for details including a discussion on optimizing the finite-sample DRIG. As formalized next, the objective \eqref{eq:obj} is strictly convex, so gradient descent is guaranteed to find the optimal solution. The proof is in Appendix \ref{proof_convex}, where we also provide a finite-sample analysis.
\begin{proposition}
For $\gamma \geq 1$, the DRIG objective $\mathcal{L}_\gamma(b)$ is strictly convex with respect to $b$.
\label{prop:convex}
\end{proposition}
 Note that the convexity of DRIG holds as along as $\mathbb{E}[Z^0{Z^0}^\top] \preceq \mathbb{E}[Z^e{Z^e}^\top]$ (where $Z^e = (X^e,Y^e)$) without assuming the SCM \eqref{eq:scm_train}. In Appendix~\ref{sec:drig_convexity_without_observational}, we prove that the DRIG objective can be convex under a strictly weaker assumption than the condition $\mathbb{E}[Z^0{Z^0}^\top] \preceq \mathbb{E}[Z^e{Z^e}^\top]$. Moreover, in Appendix~\ref{sec:optimization}, we provide numerical experiments that demonstrate the robustness of gradient descent for minimizing the DRIG objective \eqref{eq:obj}.

\section{Distributional robustness}
\label{sec:robust}
\subsection{Robustness guarantees}
\label{sec:robust_gua}
We investigate how well the population DRIG \eqref{eqn:DRIG} prediction model generalizes to test environments generated by unseen interventions as in 
\eqref{eq:scm_test}, and compare its performance with other methods. In particular, each of these methods will be shown to minimize the worst-case risk over test noise distributions $v$ in a certain set $\mathcal{C} \subseteq \mathbb{R}^{p+1}$ of random variables, i.e.,
\begin{equation}
  \argmin_{b\in\mathbb{R}^p}\sup_{v \in \mathcal{C}} \mathbb{E}[\ell(X^v,Y^v;b)].
\label{eqn:worst-case}
\end{equation}
 Throughout, we suppose that the training data is generated according to the SCM \eqref{eq:scm_train}. Further, we suppose that the `observational' condition in Assumption~\ref{ass:strength} holds, although in Appendix~\ref{sec:drig_robustness_without_observational}, we show that our robustness guarantees hold with strictly weaker conditions. We define $\mu^e:=\mathbb{E}[\varepsilon^e]$ and $S^e := \mathbb{E}[\varepsilon^e{\varepsilon^e}^\top]$ as the first and second moment, respectively, of the noise variable for every training environment $e \in \mathcal{E}$. We further suppose that the test data is generated according to the SCM \eqref{eq:scm_test}. 
 The following theorem assesses the robustness of the DRIG prediction model with the proof in Appendix~\ref{proof_thm_main}.  
\begin{theorem}
The population DRIG $b_\gamma^\mathrm{opt}$ \eqref{eqn:DRIG} is the solution to the worst-case risk minimization \eqref{eqn:worst-case} with $\mathcal{C} = \mathcal{C}_\mathrm{DRIG}^\gamma$, where $\mathcal{C}_\mathrm{DRIG}^\gamma:=  \left\{v \in \mathbb{R}^{p+1}: \bbE[vv^\top ] \preceq S^{0} + \gamma\sum_{e \in \mathcal{E}} \omega^e\left(S^e - S^{0}\right)\right\}$.
\label{thm:main}
\end{theorem} 
 This result states that DRIG is robust against noise distributions $v$ that are in the set $\mathcal{C}_\mathrm{DRIG}^\gamma$. Furthermore, if the noise variable $v$ in the test data satisfies $\mathbb{E}[vv^\top] = S^{0} + \gamma\sum_{e\in\mathcal{E}}\omega^e(S^e-S^{0})$, then, in population, DRIG provides the best linear prediction model for the test data. The scalar $\gamma \geq 0$, which is a tuning parameter for our method DRIG, controls the strength of the noise interventions that our prediction model is robust against. The larger this parameter, the larger the set $\mathcal{C}^\gamma_\mathrm{DRIG}$, and the stronger the noise intervention $v$ can be. Furthermore, the column space of the matrix $S^0+\gamma\sum_{e \in\mathcal{E}}\omega^e(S^e-S^0)$ represents the ``directions" of the interventions that DRIG protects against with a controllable strength; the larger the dimension of this subspace, the more directions the DRIG is robust against. We provide further illustrations of the intervention class in Appendix~\ref{app:illus_perturb}. 

{The weights $\omega^e$ affect the robustness set $\mathcal{C}_\mathrm{DRIG}^\gamma$. Without any knowledge of the test distribution, we recommend choosing the weights as described in Section~\ref{sec:method}. We may have some domain knowledge, for example, that the test data is close to some environment(s). More commonly, we may have access to unlabeled and possibly some labeled samples from the test distribution. In such semi-supervised settings, the weights $\omega^e$ as well as the tuning parameter $\gamma$ may be chosen to calibrate to the test environment; see Section~\ref{sec:adaptive} for more discussion.

Additionally, if no test data is available, then the user must choose the parameter $\gamma$ (which also impacts the robustness set $\mathcal{C}_\mathrm{DRIG}^\gamma$) based on domain expertise; this situation is similar to most DRO methods where the radius of the robustness set must be pre-specified.
}


\textbf{Comparison to other methods:} We contrast the robustness guarantees provided by DRIG with the ones obtained by OLS estimates, the anchor regression estimate, group DRO~\citep{sagawa2019distributionally}, and the causal parameter $b^\star$. 
Recall that the OLS estimate on the reference environment and the pooled OLS estimate are the DRIG estimates with $\gamma = 0$ and $\gamma = 1$, respectively. Thus, appealing to Theorem~\ref{thm:main}, these estimates are minimizers of the worst case risk \eqref{eqn:worst-case} with $\mathcal{C}_\text{rOLS} := \mathcal{C}_\mathrm{DRIG}^{\gamma = 0}$ and $\mathcal{C}_\text{pOLS}:= \mathcal{C}_\mathrm{DRIG}^{\gamma = 1}$ with $\mathcal{C}_\text{rOLS} \subseteq \mathcal{C}_\text{pOLS} \subseteq \mathcal{C}_\mathrm{DRIG}^{\gamma}$ for any $\gamma \geq 1$. Thus, OLS 
on the reference environment does not protect against any perturbations {that exceed the perturbations in the reference environment alone}, and the pooled OLS protects against perturbations within the training heterogeneity; both approaches are inferior to DRIG in providing robust predictions under unseen (larger) test perturbations. 

When the noise interventions are additive, i.e. $\varepsilon^e = \varepsilon + \delta^e$ for every $e\in\mathcal{E}$ with $\varepsilon,\delta^e$ being independent, anchor regression improves the OLS by protecting against potentially stronger perturbations in
$\mathcal{C}^\gamma_{\rm anchor}$, where
	$\C_\mathrm{anchor}^\gamma = \left\{v \in \mathbb{R}^{p+1}: \bbE[vv^\top ] \preceq \sum_{e \in \mathcal{E}}\omega^e\left(S^e+(\gamma-1)\mu^e{\mu^e}^\top\right)\right\}$ as proved in Appendix~\ref{proof_robust_anchor}. Note that the perturbation strength $\gamma$ is only acting on the means $\mu^e$ {and thus anchor regression only protects against perturbations in the means}. 
Although anchor regression provides more robust predictions than OLS (formally $\mathcal{C}_\text{rOLS} 
 \subseteq \C_\mathrm{anchor}^\gamma$), it protects against a smaller set of perturbations than DRIG as $\C_{\text{anchor}}^\gamma \subseteq \C_{\mathrm{DRIG}}^\gamma$. In particular, since DRIG exploits both mean and variance shifts, it is robust against perturbations in strictly (and often much) more directions than anchor regression. For instance, when $|\mathcal{E}|=2$, anchor regression can only protect against perturbations $v$ that lie in a $2$-dimensional subspace (regardless of the number and strength of perturbations observed in the training data), while DRIG can protect against $v$ in arbitrary directions if all variables are intervened on {(formally if $S^1-S^0 \succ 0$)}. We will illustrate this comparison in Section~\ref{sec:exam}.

As described in the introduction, standard DRO methods, which minimize the worst-case prediction loss with respect to a divergence ball around the training distribution, lead to overly pessimistic models \citep{Duchi2020DistributionallyRL,sagawa2019distributionally}. To construct a realistic set of possible test distributions without being overly conservative, in settings where we have access to multiple environments, a class of DRO methods, known as group DRO, minimize the prediction loss over the worst-case group. Formally, in the context of linear models, group DRO is defined as $\argmin_b\max_{e\in\cE}\ \bbE[(Y^e-b^\top X^e)^2]$; this is equivalent to minimizing the  loss over the worst-case mixture of the distributions in the training environment.  
Suppose there exists an environment $m\in\cE$ such that $S^e\preceq S^m$ for all $e\in\cE$. Then, we show in Appendix~\ref{proof_robust_gdro} that group DRO is robust against the perturbation class $\mathcal{C}_{\mathrm{gDRO}}=\left\{v \in \mathbb{R}^{p+1}: \bbE[vv^\top ] \preceq   S^m\right\}$.  
Without assuming the existence of a dominating environment $m$, the perturbation class that group DRO protects against is not clear. Moreover, unlike DRIG (and anchor regression), group DRO does not have a tuning parameter that actively controls the size of the perturbation class; it is rather a passive interpolation between in-sample prediction and causality, merely relying on the training environments. 
Thus, group DRO cannot protect against test perturbations larger than training perturbations.


Finally, the causal parameter $b^\star$ is the solution to the worst-case risk minimization \eqref{eqn:worst-case} with $\mathcal{C} = \C_{\mathrm{causal}}$, where $\C_{\mathrm{causal}} = \left\{v \in \mathbb{R}^{p+1}: |\mathbb{E}[v_{p+1}v_j]| < \infty\allowbreak \text{ for} \allowbreak j \in [p+1]\right\}$ as proved in Appendix~\ref{proof_robust_causal}. To better understand the vectors inside $\C_{\mathrm{causal}}$, consider $\tilde{\mathcal{C}} = \{\varepsilon^0+\tilde{v}~|~ \tilde{v} \in \mathbb{R}^{p+1}, \tilde{v} \text{ independent of }\allowbreak\varepsilon^0, \tilde{v}_{p+1}\equiv 0\}$ where $\xi^0$ is the noise variable in the `observational' environment. The set $\tilde{\mathcal{C}}$ thus consists of independent additive interventions with no interventions on the latent variables and on $Y$, but allows for arbitrary intervention on the covariates $X$. 

DRIG may be preferred over the causal parameter for multiple reasons. First, as $\tilde{\mathcal{C}} \subseteq \C_{\mathrm{causal}}$, the causal parameter protects against arbitrary interventions on the covariates $X$, thus yielding overly conservative prediction models that come with a price of subpar predictive performance on moderately perturbed data. Second, the causal parameter is often not identifiable, especially when the interventions do not happen on all the variables.

In summary, DRIG is an attractive alternative for robust prediction over standard OLS estimators as well as anchor regression, group DRO, and the causal prediction model. 

\subsection{Illustrative examples}\label{sec:exam}
We give two simple examples to illustrate how our method performs under general additive interventions compared to existing methods. In particular, we compare our method DRIG (with $\gamma=5$), the observational OLS (DRIG with $\gamma=0$), the pooled OLS (DRIG with $\gamma=1$), causal parameter (DRIG with $\gamma=\infty$ if identifiable), and anchor regression (DRIG with mean shifts only and $\gamma=5$), all of which are special cases of DRIG. {We also consider group DRO \citep{sagawa2019distributionally}}. For simplicity, we consider a univariate $X\in\bbR$ and two training environments $e=0,1$. We compute the population versions of all estimators and also evaluate their test performance in population. We provide additional experiments for finite-sample estimators with multivariate covariates, multiple environments, as well as the oracle choice of $\gamma$ in Appendix~\ref{sec:simulations}.

\begin{example}[Covariate-intervened]\label{ex:robust1}
Data are distributed according to the SCMs: $\bbP^0: X^0 = \varepsilon_x; Y^0 = 2X^0 + \varepsilon_y$, $\bbP^1: X^1 =\varepsilon_x + \delta_x^1; Y^1 = 2X^1 + \varepsilon_y$, and $\bbP^v: X^v =\varepsilon_x + v_x; Y^v = 2X^v + \varepsilon_y$ 
Here, $(\varepsilon_x,\varepsilon_y)$ follows a bivariate Gaussian with means 0, variances 1, and covariance 0.5, intervention $\delta_x^1 \sim \cN(0.5, 1)$ only affects $X$, 
and $v_x \sim\cN(\mu_v,\sigma^2_v)$ represents a different intervention where $\mu_v^2+\sigma^2_v=1.25\alpha$ with a factor $\alpha$ controlling the test perturbation strength. 
\end{example}

Figure~\ref{fig:illus}(a) shows the mean squared errors (MSEs) of various methods in the perturbed test distribution $\bbP_v$ for varying perturbation strengths $\alpha$. The causal parameter is invariant (i.e., a constant MSE) for any perturbations on $X$, but is suboptimal when the perturbations are small or moderate. The observational OLS performs the best only when the test distribution is almost identical to the observational distribution and performs poorly when the perturbation grows.
DRIG achieves a trade-off between the causal parameter and observational OLS, leading to favorable robustness. In particular, under small or moderate perturbations, DRIG attains a lower test MSE than the causal parameter; when the perturbations become relatively strong, DRIG is superior to the OLS estimators. 
In this setting, DRIG with a finite $\gamma$ protects against the perturbation class $\{(v_x,v_y):\bbE[v_x^2]\le\tfrac{\gamma}{2}1.25,\ v_y=0\}$. Thus, the optimal $\gamma$ in DRIG should be $\gamma=2\alpha$, whereas we keep $\gamma$ fixed in our simulations. This highlights the robustness of DRIG to the choice of $\gamma$.

Anchor regression exploits heterogeneity in the means, thus generally outperforming OLS. However, since it can only exploit mean shifts, it tends to be inferior to DRIG. In Appendix~\ref{app:emp_results}, we show a case with weaker mean shifts. Here, anchor regression performs almost identically to pooled OLS, whereas DRIG exploits extra heterogeneity and outperforms both. Similarly, group DRO outperforms OLS when the test perturbation strength is large, although DRIG yields better predictions. 

\begin{example}[All-intervened]\label{ex:robust2}
Data are distributed according to the SCMs: $\bbP^0: X^0 = \varepsilon_x;Y^0 = 2X^0 + \varepsilon_y$, $\bbP^1: X^1 = \varepsilon_x + \delta_x^1;Y^1 = 2X^1 + \varepsilon_y + \delta_y^1$; and $\bbP^v: X^v = \varepsilon_x + v_x;Y^v = 2X^v + \varepsilon_y + v_y$. Here, $(\varepsilon_x,\varepsilon_y)$ is distributed similar to Example~\ref{ex:robust1}, $(\delta^1_x,\delta^1_y) \sim \mathcal{N}\left(\begin{pmatrix} 0.5\\0.1\end{pmatrix},\begin{pmatrix} 1 & 0.1\\0.1 &0.05\end{pmatrix}\right)$, and $\bbE[vv^\top] = \frac{\alpha}{2}
	\begin{pmatrix}
		1.25 & 0.15\\
		0.15 & 0.06
	\end{pmatrix}
	$
where $\alpha$ controls the test perturbation strength. 
\end{example}
As shown in Figure~\ref{fig:illus}(b), due to interventions on all variables, the causal parameter is no longer invariant and its prediction performance degrades as the test perturbation strength increases. In contrast, DRIG exhibits a significant advantage compared to all other methods.

\begin{figure}[!]
\centering
\begin{tabular}{cc}
	\includegraphics[width=0.34\textwidth]{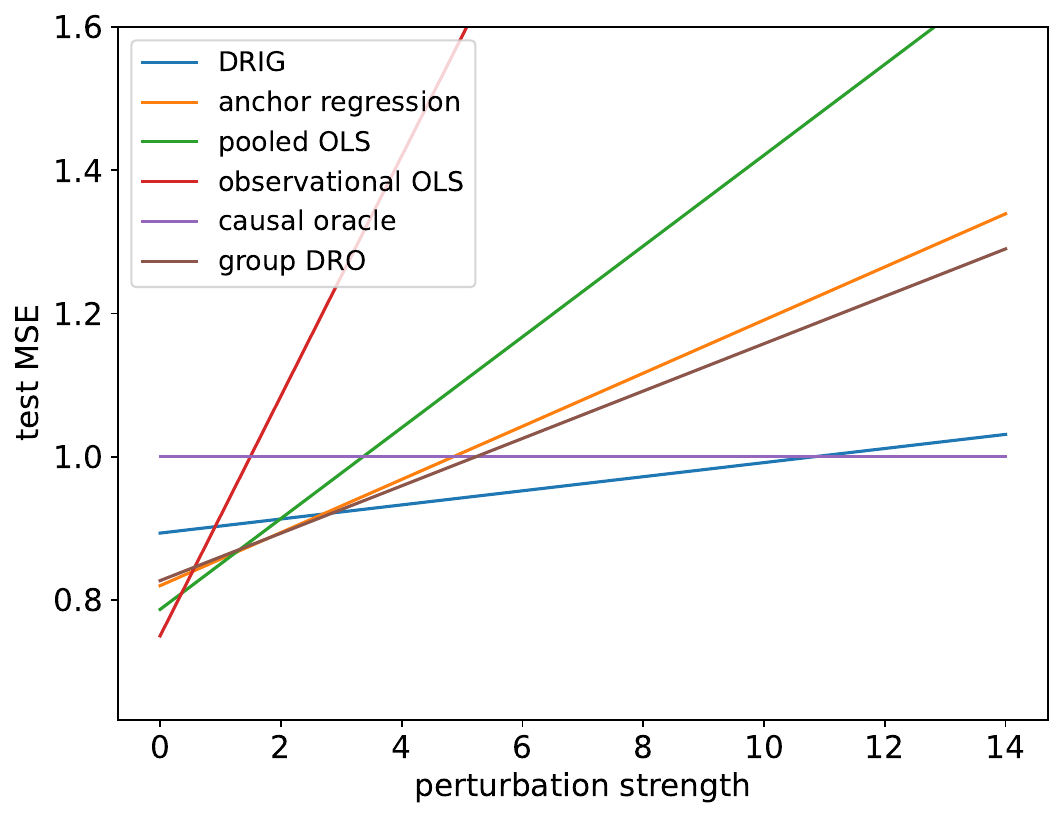}\hspace{0.3in} & \includegraphics[width=0.34\textwidth]{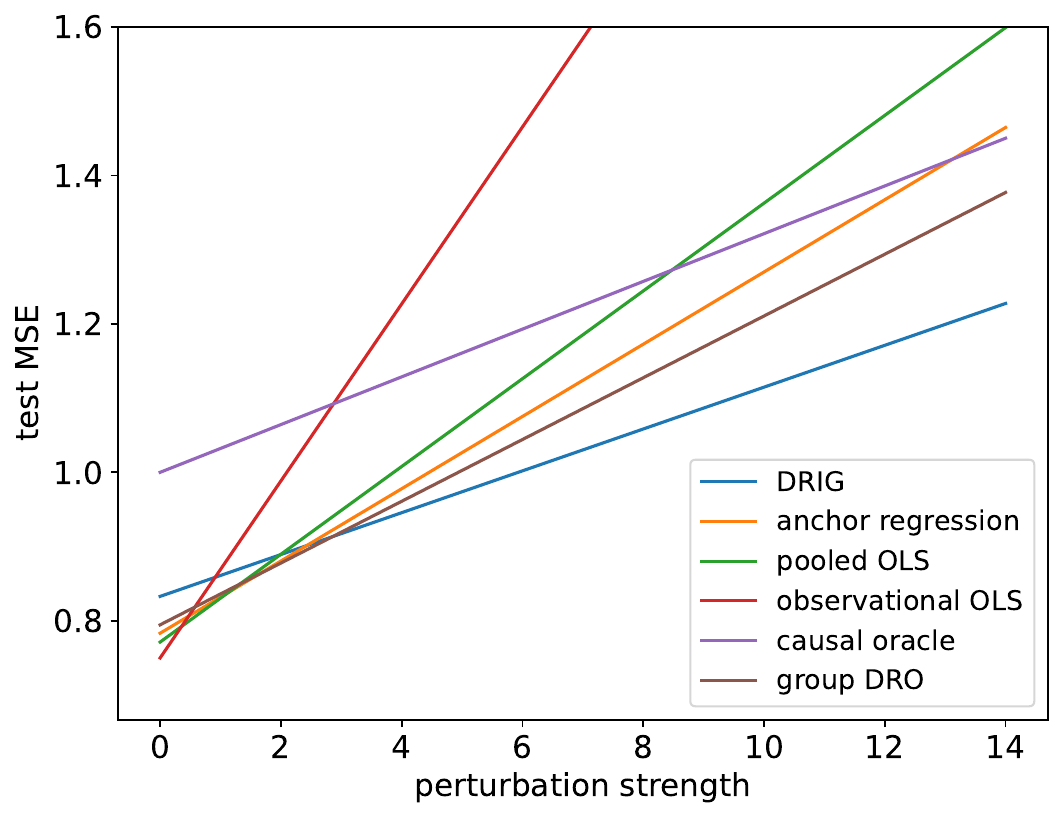}\vspace{-0.2in}\\
\end{tabular}
\caption{\small Test MSEs for varying perturbation strengths $\alpha$. \textbf{(left)}: perturbations on covariates only; \textbf{(right)}: perturbations on the covariate, response, and latent variables.}\label{fig:illus}
\end{figure}

\subsection{Infinite robustness and causality}
We analyze DRIG when $\gamma \to \infty$, and highlight how infinite robustness (as guaranteed by Theorem~\ref{thm:main}) connects to causality and invariance. Define $L^{\star} := \sum_{e \in \mathcal{E}} \omega^e (S^e-S^{0})$ and $C^\star = (I-B^\star)^{-1}$ with block forms, 
$ L^{\star} = \begin{pmatrix}L^\star_x & L^\star_{xy} \\
		{L^{\star}_{xy}}^\top & L^{\star}_{y}
	\end{pmatrix}$ and $C^\star = \begin{pmatrix}C^\star_{x} & C^\star_{xy} \\ C^{\star\top}_{yx} & C^\star_{y} \end{pmatrix},$ where $L^{\star}_y,C^\star_y \in \mathbb{R}$. We suppose that the data is generated according to the SCM \eqref{eq:scm_train}, and that the `observational' condition in Assumption~\ref{ass:strength} holds; we show in Appendix~\ref{sec:causal_weaker_assumptions} causal identifiability results of DRIG hold under a strictly weaker condition than Assumption~\ref{ass:strength}. Denote $b_\infty^{\mathrm{opt}}:=\lim_{\gamma\to\infty}b_\gamma^{\mathrm{opt}}$.
\begin{theorem} We have 
\begin{equation}\label{eq:b_insufficient}
        b_\infty^{\mathrm{opt}} = \argmin_{b\in\mathcal{I}}\min_e\bbE[(Y^e-b^\top X^e)^2],
\end{equation}
where $\mathcal{I} := \{b \in \mathbb{R}^{p}: b \text{ satisfies the gradient invariance condition in Definition~\ref{def:gi}}\}$ is a non-empty set. If additionally $\mathrm{rank}([C^\star{L^{\star}}{C^\star}^\top]_{1:p,1:p}) = p$, then $\mathcal{I}$ is a singleton, and 
\begin{equation}\label{eq:bias}
		b_\infty^{\mathrm{opt}} = b^\star + \left([C^\star{L^{\star}}{C^\star}^\top]_{1:p,1:p}\right)^{-1}(C^\star_xL^{\star}_{xy}+L^{\star}_{y}C_{xy}^\star).
	\end{equation}
\label{thm:solution_infty}
\end{theorem}
We prove Theorem~\ref{thm:solution_infty}
 in Appendix \ref{proof_prop_solution_infty}. The first part of the theorem states that DRIG with $\gamma \to \infty$ identifies -- among models in $\mathcal{I}$ that have invariant gradient -- the most predictive model in the reference environment. The second part states that if the aforementioned subspace is full dimensional, the set of gradient invariant models $\mathcal{I}$ is a singleton; appealing to \eqref{eq:b_insufficient}, the unique element in $\mathcal{I}$ is the solution of DRIG when $\gamma \to \infty$, and is characterized explicitly in \eqref{eq:bias}. We provide a thorough discussion on Theorem~\ref{thm:solution_infty} in Appendix~\ref{sec:causal-identify}. In particular, 
we investigate how $b_\infty^{\mathrm{opt}}$ is related to the causal parameter $b^\star$ under various scenarios of interventions and causal structures. To summarize, $b_\infty^{\mathrm{opt}}$ recovers $b^\star$ when assuming sufficient interventions on the covariates and no interventions on the response and latent variables (i.e., $\mathrm{rank}(L^\star_x)=p$ and $L^\star_{xy}=L^\star_y = 0$). In addition, we study the bias of $b_\infty^{\mathrm{opt}}$ in estimating the causal parameter, when allowing for interventions on $Y$ or the latent variables, or when encountering insufficient interventions on $X$. 

In general, causal identification requires stronger assumptions about the underlying data distribution than those needed for robust prediction; the robustness guarantee in Theorem~\ref{thm:main} remains valid regardless of the fulfillment of the identifiability conditions. This further highlights the merit of causality-oriented robustness for wider and more realistic applications. 
\section{Calibrating DRIG via semi-supervised data}\label{sec:adaptive}


We consider a semi-supervised domain adaptation setting, where we have a set of unlabeled test or target examples and possibly a small set of labeled test examples. Data from the target distribution provides some information on the strength of interventions we may encounter and thus making use of such information could allow us to calibrate our prediction model. 

As an example of a semi-supervised setting, consider the application in Section~\ref{sec:real_data_experiments}, where our training data consists of patient information and their heart rates 48 hours after entering the intensive care unit (ICU) across multiple hospitals. Suppose our goal is to perform real-time predictions of ICU patients' heart rates after 48 hours in a new hospital. From this new hospital, we may have covariate data on patients entering the ICU, and since there is a 48-hour delay, only a small amount of heart rate measurements. 

Throughout, we assume the training data is generated according to \eqref{eq:scm_train} and that Assumption~\ref{ass:strength} holds. Suppose the test distribution ${P}_\mathrm{test}$ is generated according to the SCM \eqref{eq:scm_test} with an unknown intervention variable $v$. We let ${P}_\mathrm{test}^x$ be the marginal distribution of the covariates $X$. We 
assume that we are given a collection of i.i.d.\ labeled test samples $\{(X^v_i,Y^v_i)\sim {P}_\mathrm{test},i=1\dots,n_l\}$ with a small (or possibly zero) $n_l$ and a collection of i.i.d.\ unlabeled test samples $\{X^v_i\sim {P}_\mathrm{test}^x,i=1\dots,n_u\}$ with $n_u$ fairly large. Let $G^e = \mathbb{E}[Z^e{Z^e}^\top]$ for every $e\in\mathcal{E}$ with $Z^e = (X^e,Y^e)$, $G^v_x:=\bbE[ X^v{X^v}^\top]$, $G_{xy}^v:=\bbE[X^vY^v]$, and their estimates based on the test samples $\hat{G}^v_x:=\frac{1}{n_u}\sum_{i=1}^{n_u}X^v_i{X^v_i}^\top$, $\hat{G}_{xy}^v:=\frac{1}{n_l}\sum_{i=1}^{n_l}X^v_iY^v_i$, and $\hat{G}^v_{y}:=\frac{1}{n_l}\sum_{i=1}^{n_l}(Y^v_i)^2$ with $\hat{G}^v := \begin{pmatrix}\hat{G}_x^v & \hat{G}_{xy}^v \\\hat{G}_{yx} & \hat{G}_y\end{pmatrix}$. 

A naive prediction model is based on OLS under the test distribution. The population version of the test OLS is given by $\argmin_b (-b,1)\mathbb{E}[vv^\top](-b,1)^\top$ and the associated estimator based on the labeled and unlabeled test samples is given by $\argmin_b (-b,1)\hat{G}^v(-b,1)^\top$ where $\hat{G}^v$ is the plug-in estimator for $G^v$. Naturally, if the number of labeled and unlabeled test samples $n_l,n_u$ tend to infinity, the finite sample OLS minimizes the test MSE with high probability. However, in our setting of a small number of labeled test samples, finite-sample OLS can have a high variance and perform poorly.
Our objective is to calibrate DRIG to achieve a small test MSE under ${P}_\mathrm{test}$ by exploiting both the heterogeneity within the training data and the limited test samples.

{\subsection{DRIG-A: selecting weights $\omega^e$ and $\gamma$}}
\label{sec:weights_gamma}
DRIG \eqref{eqn:DRIG} can be equivalently reformulated as $\argmin_{b}(-b,1)[G^{0}+\sum_{e\in\mathcal{E}}\tilde{\omega}^e(G^e-G^0)](-b,1)^\top$ where $\tilde{\omega}^e = \gamma\omega^e$ for each $e \in \mathcal{E}$. Let $J:=G^{0}+\sum_{e\in\mathcal{E}}\tilde{\omega}^e(G^e-G^0)$; we hide the dependency of $J$ on the weights $\tilde{\omega}^e$. Naturally, based on semi-supervised data, we can choose $\tilde{\omega}^e$ to align the DRIG estimate to the OLS estimate via the following convex optimization problem:  
\begin{eqnarray*}
\begin{aligned}
\tilde{\omega}^{e}_\mathrm{opt} &:= \argmin_{\{\tilde{\omega}^e\}_{e\in\mathcal{E}} \geq 0} n_u\|J_{1:p,1:p} - \hat{G}^v_{1:p,1:p}\|_F^2 + n_l(2\|J_{1:p,p+1} - \hat{G}^v_{1:p,p+1}\|_F^2+\|J_{p+1,p+1} - \hat{G}^v_{p+1,p+1}\|_F^2).
\end{aligned}
\end{eqnarray*}
Further, set $\gamma = \sum_{e\in\mathcal{E}}\tilde{\omega}^e_\mathrm{opt}$ and $\omega^e = \tilde{\omega}^e_\mathrm{opt}/\gamma$. We then supply this choice of hyperparameters $\gamma,\omega^e$ to \eqref{eqn:DRIG}. The resulting estimator again satisfies a similar robustness guarantees as Theorem~\ref{thm:main}; see Appendix~\ref{sec:robustness_drig-A} for more results including discussion on finite-sample consistency guarantees. Notice that DRIG-A may also be applied even without any labeled samples (when $n^l = 0$). Its numerical results are shown in the single-cell application below. 


\subsection{DRIG-A+: More hyperparameters, more flexible robustness}
\label{sec:adaptive_hyperparameters}
In Section~\ref{sec:weights_gamma}, we described how $\gamma,\omega^e$ may be chosen to adapt DRIG to a test environment of interest. In essence, this approach aims to choose $\gamma,\omega^e$ to adjust the shape and size of the set $\mathcal{C}_\mathrm{DRIG}^\gamma$ such that the second moment $\mathbb{E}[vv^\top]$ of the intervention $v$ in the test environment lies close to its boundary -- if $\mathbb{E}[vv^\top]$ lies exactly on the boundary, then DRIG yields the best linear prediction model in population. However, when the number of environments is much smaller than the number of observed variables, we may not have enough degrees of freedom to make $\mathbb{E}[v{v}^\top]$ be close to the boundary of $\mathcal{C}_\mathrm{DRIG}^\gamma$.

To remedy this potential issue -- particularly when the number of unlabeled test samples $n^u$ is large and the number of labeled test samples $n^l$ is not too small -- we propose an extension of the original formulation that allows for more flexible control over the shape and size of the perturbation class.  We consider a matrix of hyperparameters $\Gamma$ in the form $\Gamma=\mathrm{diag}(\Gamma_x,\gamma_y)$
with $\Gamma_x\in\bbR^{p\times p}$ and $\gamma_y\in\bbR$. Given a positive semidefinite matrix $\Gamma$, we define the population version of the modified DRIG estimator, dubbed DRIG-A+, as
\begin{equation}
b^\mathrm{opt}_\Gamma  := \argmin_{b} \Big\{\min_{e\in\mathcal{E}}\bbE[(Y^e-b^\top X^e)^2] + \sum_{e \in \mathcal{E}}\omega^e\big(\bbE[\gamma_y Y^e - b^\top\Gamma_x X^e]^2 - \min_{e\in\mathcal{E}}\bbE[\gamma_y Y^e - b^\top\Gamma_x X^e]^2\big)\Big\}.
\label{eqn:drig_A}
\end{equation}

Note that when $\Gamma=\gamma I_{p+1}$ with a scalar $\gamma\ge0$, DRIG-A+ estimator $b^{\mathrm{opt}}_\Gamma$ reduces to the original DRIG estimator $b^{\mathrm{opt}}_\gamma$ in \eqref{eqn:DRIG} . Thus, the DRIG-A+ method is a generalization of DRIG with potentially more hyperparameters. As we show in the following theorem, the additional parameters provide flexibility in controlling both the size and shape of the perturbation class. We define $\tilde\Gamma:=(I-B^\star)\Gamma(I-B^\star)^{-1}$ for notational clarity.

\begin{theorem}\label{thm:robust_adap} 
The DRIG-A+ estimator
$b^\mathrm{opt}_\Gamma$ is the solution to the worst-case risk minimization \eqref{eqn:worst-case} with $\mathcal{C} = \mathcal{C}_\mathrm{DRIG\text{-}A+}^\Gamma := \left\{v \in \mathbb{R}^{p+1}:\bbE[vv^\top]\preceq S^0+\tilde\Gamma\sum_{e\in\mathcal{E}}\omega^e(S^e-S^0)\tilde\Gamma^\top\right\}$.
\end{theorem}
We prove Theorem~\ref{thm:robust_adap} in Appendix~\ref{proof_thm_robust_adap}. The result states that the DRIG-A prediction model is robust against test perturbations that are in the set $\mathcal{C}^\Gamma_\mathrm{DRIG\text{-}A+}$; both the size and shape of the perturbation class $\mathcal{C}^\Gamma_\mathrm{DRIG\text{-}A+}$ can be adjusted by an appropriate choice of $\Gamma$.  It is worth noting that while DRIG-A+ can often provide more robustness compared to the original DRIG formulation, it in general moves further away from causality. In particular, we show in Appendix~\ref{app:causal_driga} that when $\Gamma_x/\gamma_y \neq I$, \eqref{eqn:drig_A} does not recover the causal parameter even when it is identifiable (e.g., the setting in Corollary~\ref{cor:identifiable}). This phenomenon highlights the trade-off between prediction and causality: DRIG-A+, compared to DRIG, is designed more towards the goal of prediction (see Figure~\ref{fig:directional_shift},right).

In principle, one can select the matrix $\Gamma$ based on some prior or expert knowledge on the relation between the test and training data. More generally, we can use semi-supervised data from the test distribution to choose $\Gamma$. A naive prediction model is based on the OLS estimator $\hat{b}_{\mathrm{tOLS}} := ({\hat{G}^v_x})^{-1}\hat{G}_{xy}^v$ under the test distribution. In Appendix \ref{eq:selecting_Gamma}, we describe how to specify $\Gamma$ in \eqref{eqn:drig_A} so that the test MSE achieved by DRIG-A+ is smaller than the one achieved by $\hat{b}_{\mathrm{tOLS}}$. {In our scheme, we let $\Gamma = \text{diag}(\Gamma_x,\gamma_y)$, where we choose the matrix $\Gamma_x \in \mathbb{R}^{p \times p}$ using the large amount of unlabled samples, and the scalar $\gamma_y$ using the labeled samples. Since only one hyperparameter is chosen from labeled samples and the rest are chosen from unlabeled samples, DRIG-A+ can be useful in many semi-supervised settings.}

For coefficient $b \in \mathbb{R}^p$, denote the population test MSE by $\mathcal{L}_\mathrm{test}(b)=\bbE[(Y^v - b^\top{X}^v)^2]$. The following theorem, with proof in Appendix~\ref{proof_thm_adap_mse}, highlights the advantage of using the DRIG-A+ estimator with $\hat{\Gamma}=\mathrm{diag}(\hat\Gamma_x,\hat\gamma_y)$
over $\hat{b}_{\mathrm{tOLS}}$. While the result considers access to the training distributions, similar results can be established for finite training samples.

\begin{theorem}\label{thm:adap_mse}
Assume $p>1$ and $\var(X^vY^v)\succ \big(\bbE[X^vY^v]-\bbE[X^0Y^0]\big)\big(\bbE[X^vY^v]-\bbE[X^0Y^0]\big)^\top$. {Assume further that $X^v,Y^v$ have bounded second moments and $\|b\|\le B$ for some $B>0$.} Then there exist positive integers $N_l$ and $N_u$ such that when $n_u\ge N_u$ and $n_l\leq N_l$, we have $\bbE [\mathcal{L}_\mathrm{test}({b}_{\hat{\Gamma}}^\mathrm{opt})]<\bbE [\mathcal{L}_\mathrm{test}(\hat{b}_{\mathrm{tOLS}})]$, where the expectation is taken over all test samples.
\end{theorem}
The first condition 
indicates that the variance of the cross term $X^vY^v$ exceeds the expected difference between the cross terms on the test and observational distributions.
Then, Theorem~\ref{thm:adap_mse} implies that our DRIG-A+ estimator ${b}_{\hat{\Gamma}}^\mathrm{opt}$ is favored over $\hat{b}_{\mathrm{tOLS}}$ in terms of the test MSE, when we have sufficiently many unlabeled samples and not sufficiently many labeled samples from the test distribution. It is not yet clear when the gap between the two MSEs is significant. In Appendix \ref{sec:illustration_DRIG_A+}, we use simulations to empirically demonstrate the advantage of our adaptive estimator; see also real data analysis in Section~\ref{sec:real_data_experiments}.

\section{Real data analysis }
\label{sec:real_data_experiments}
\subsection{Single-cell data}
\citet{replogle2022mapping} published a large-scale single-cell RNA sequencing dataset where they performed genome-scale Perturb-seq targeting on all expressed genes with CRISPR perturbations across millions of human cells. We utilize the dataset on the RPE1 cells, as it focuses on putatively important genes and tend to respond more to interventions. After preprocessing the data following \citet{chevalley2022causalbench}, we arrive at 10 genes with the highest expression level as the observed variables. We regard one gene as the response variable and the others as covariates, with the reasoning given in Appendix~\ref{app:singlecell}. Our training data contains 11,485 observational data and 10 interventional environments in each of which one of the 10 genes is intervened on. The sample sizes of the interventional environments range from 100 to 500. 

Moreover, we have hundreds of additional environments, each of which involves the intervention on one hidden gene (i.e., a gene that is not among the 10 observed ones). These environments, potentially different from the training environments, serve as the test distributions to assess the robustness 
of prediction models.

We apply DRIG and anchor regression with different $\gamma$ {as well as group DRO}, and evaluate the estimated models on the test environments. Among the hundreds of test environments, we select 50 environments where the observational OLS performs the worst, indicating the presence of large distributional shifts. Figure~\ref{fig:sc_test} presents the boxplots of the MSEs on the 50 test environments for different methods with varying $\gamma$. DRIG with an increasing $\gamma$ achieves a smaller worst-case test MSE, which is consistent with Theorem~\ref{thm:main}. Similarly, anchor regression also demonstrates similar robustness behavior, although it generally performs worse than DRIG, exhibiting a larger worst-case MSE. This discrepancy {indicates}
that shifts among different environments arise due to random interventions that affect not only the means but also the variances, and DRIG is able to better exploit the rich heterogeneity. {Group DRO is inferior to DRIG or anchor regression with positive regularizations.} Recall that DRIG with $\gamma=0$ yields the observational OLS and DRIG with $\gamma=1$ is the pooled OLS. We observe that all the shown quantiles of MSEs decrease as $\gamma$ increases, especially the worst-case error, indicating that the OLS estimators are inferior to DRIG. Overall, the results here highlight the superiority of DRIG in handling distribution shifts and achieving robust predictions.



\begin{figure}
\centering
\includegraphics[width=0.8\textwidth]{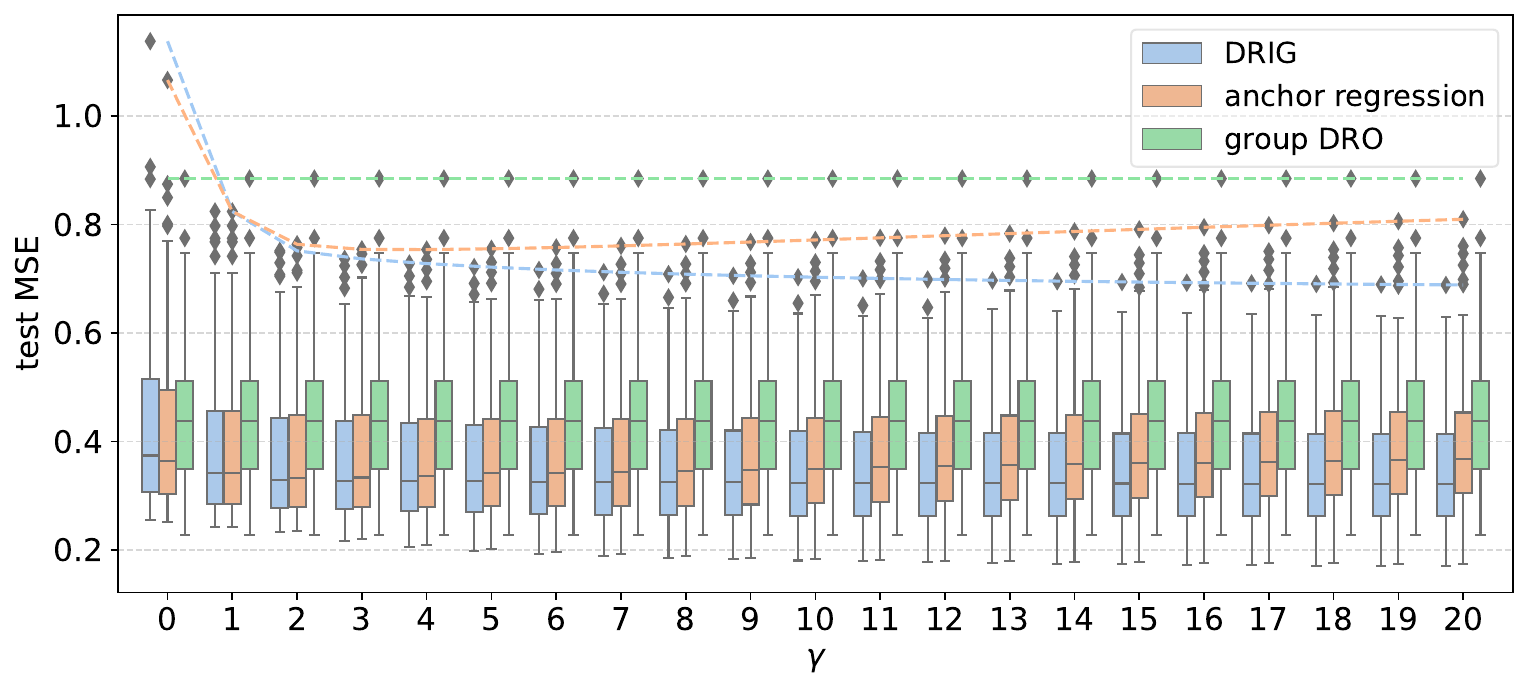} \vspace{-0.2in}
\caption{\small Boxplots of the MSEs on 50 test environments for each method with varying $\gamma$, with the worst-case MSE shown in the dashed lines on top. }\label{fig:sc_test}
\end{figure}

\begin{figure}
\centering
	\begin{tabular}{@{}cc@{}}
		\includegraphics[width=0.44\textwidth]{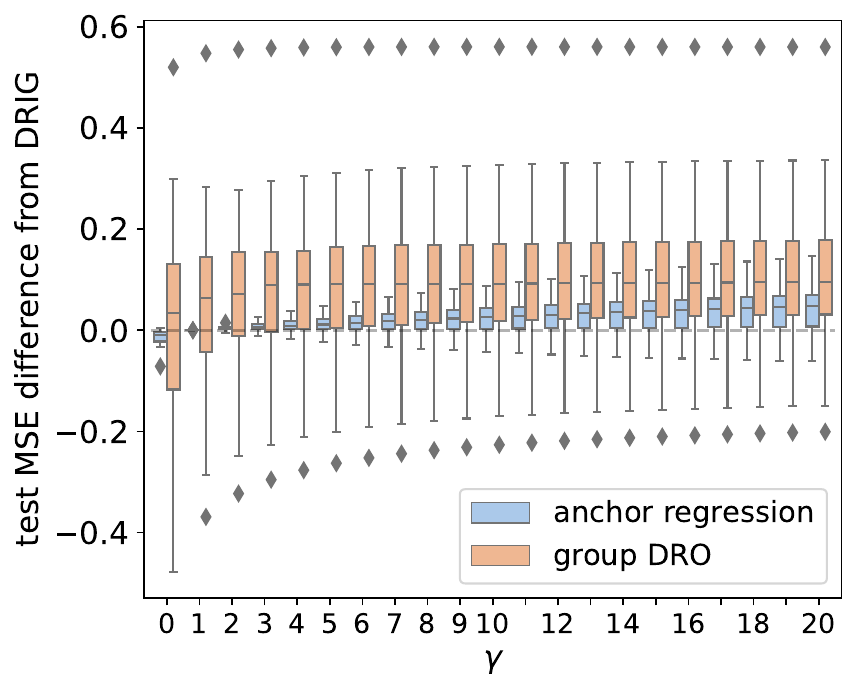} & \includegraphics[width=0.5\textwidth]{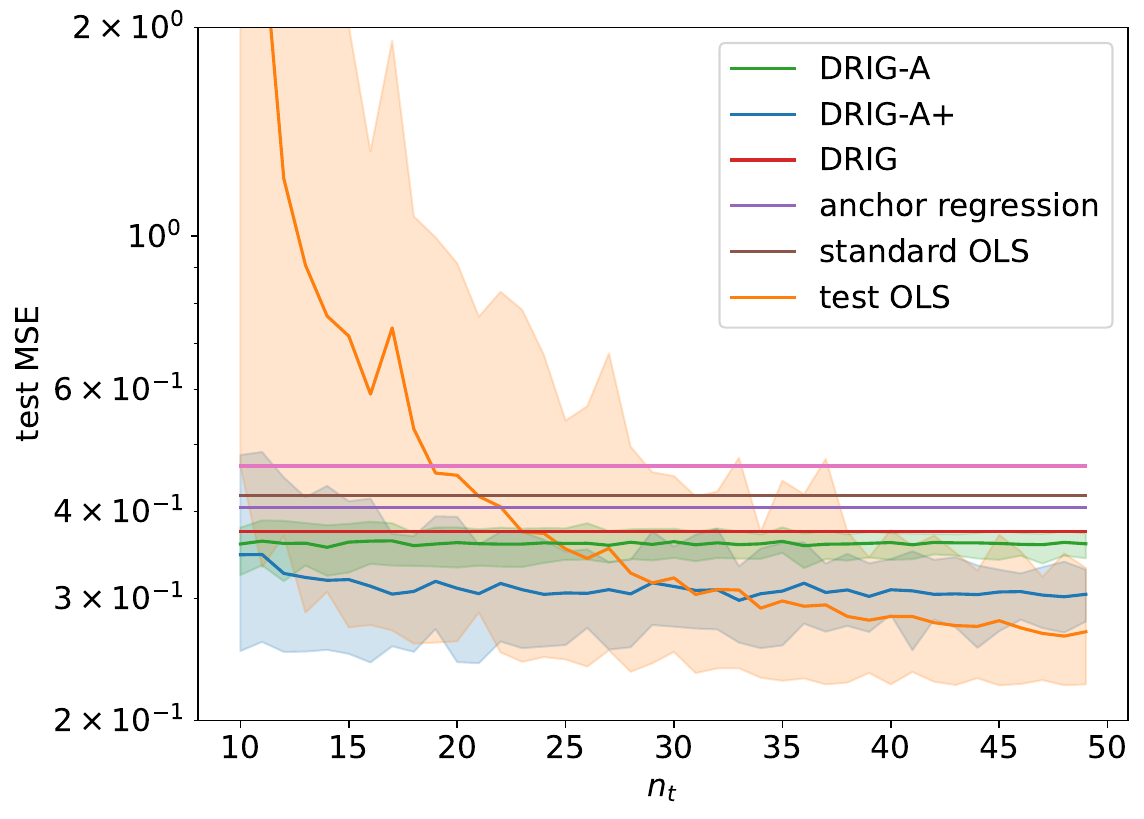} \vspace{-0.2in}
	\end{tabular}
\caption{\small (\textbf{left}) The difference of test-MSE of anchor regression and group DRO with the test MSE of DRIG for all $50$ test environments. (\textbf{right}) Performance of DRIG-A and DRIG-A+ for varying labeled sample sizes, in comparison to test-OLS and other methods that rely only on the training data. DRIG and anchor regression use fixed $\gamma=10$. Lines represent the mean and 2.5\% and 97.5\% quantiles.}\label{fig:sc_adap}
\end{figure}

In addition, it is worth noting that the robustness measured by the worst-case performance tends to stabilize with a moderate value of $\gamma$. For example, in this case, once $\gamma$ exceeds a certain threshold, such as $\gamma>5$, the performance becomes relatively stable. This suggests that there is less concern about meticulously selecting the value of $\gamma$ in order to achieve better robustness than standard approaches like OLS. Nevertheless, in Appendix~\ref{app:singlecell}, we investigate the performance of DRIG on test environments by some specific interventions. This indicates that the choice of $\gamma$ could still have a potentially crucial impact on the performance for particular test distributions and brings up the issue of selection of $\gamma$. 

{We further investigate how the methods compare on the same environment. Figure~\ref{fig:sc_adap}(left) shows the boxplots of the differences between MSEs of a competitive method and that of DRIG for each environment. DRIG leads to better prediction performance on most environments, especially with a larger $\gamma$. }


{When a small labeled sample from the test environment is available, our adaptive methods DRIG-A and DRIG-A+ can enhance prediction performance without manually selecting $\gamma$. For evaluation, to ensure a larger test sample size, we pool the aforementioned 50 test environments together as our new test domain, which is a mixture of various interventions. Given a test sample size $n_l$, we randomly draw a subsample from the test domain and apply DRIG-A, DRIG-A+, and test OLS. As shown in Figure~\ref{fig:sc_adap}, with a relatively small number of labeled test data, DRIG-A and DRIG-A+ outperform all other methods that rely solely on the training data including DRIG. Since DRIG-A+ offers much more flexibility than DRIG-A for adapting to the test environment, we see that DRIG-A+ yields more robust predictions. DRIG-A+ exhibits superior performance and greater stability compared to test OLS. Finally, as the number of labeled test data increases, the advantage of DRIG-A+ over the test OLS diminishes, aligned with our theoretical result in Theorem~\ref{thm:adap_mse}.}

{\subsection{{Intensive care unit data}}}
Our second case study is based on two large electronic health record databases. The first is MIMIC-III~\citep{johnson2016mimic} which contains deidentified data for ICU admissions to the Beth Israel Deaconess Medical Center in Boston. The second is eICU~\citep{pollard2018eicu} collected from a large number of hospitals located within the United States excluding the hospital of MIMIC-III. We consider a regression task with the outcome being the average heart rate of patients between 48-72 hours after ICU admission and covariates including various clinical and laboratory measurements and patient demographics. After preprocessing, we end up with 31 covariates, 784 observations from eICU among four regions in the US (four training environments), and 67 observations from MIMIC-III (test environment). More details about the datasets and preprocessing are given in Appendix~\ref{app:exp_detail}. Our goal is to learn a prediction model from the training environments that performs well in the test environment. Note that here, the observational assumption does not hold.

Figure~\ref{fig:icu_results}(left) shows the test MSEs on MIMIC-III for different methods. With any proper regularization, DRIG exhibits a clear advantage over group DRO and anchor regression; indeed, anchor regression and group DRO do not improve over the pooled OLS (DRIG with $\gamma=1$). To further investigate how the prediction models perform for each single test observation,  Figure~\ref{fig:icu_results}(right) presents the boxplot of test-MSEs across each of the test observations for different values of $\gamma$. Comparing the worst-case or upper quantile test MSEs across the methods, we again see that DRIG outperforms competing methods.

\begin{figure}
	\centering
	\begin{tabular}{@{}c@{}c@{}}
		\includegraphics[height=0.26\textwidth]{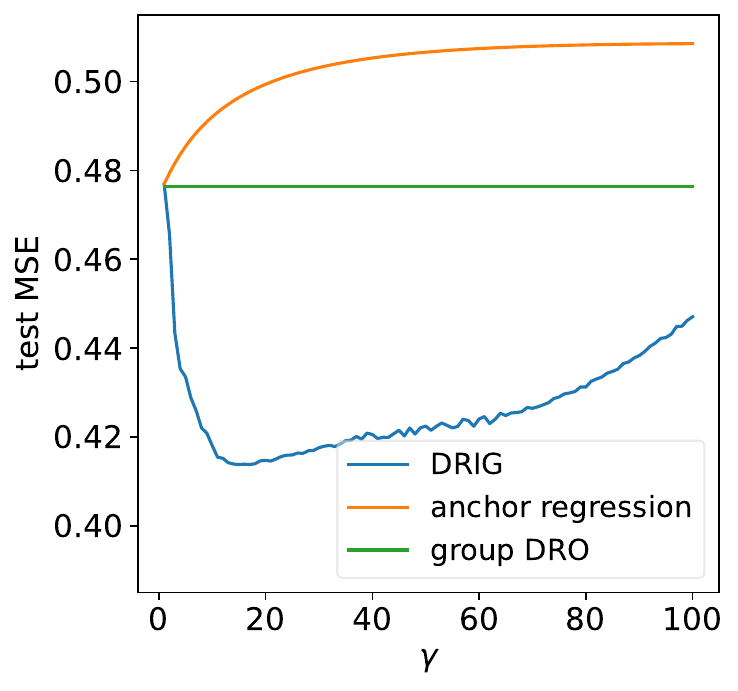} & \includegraphics[height=0.26\textwidth]{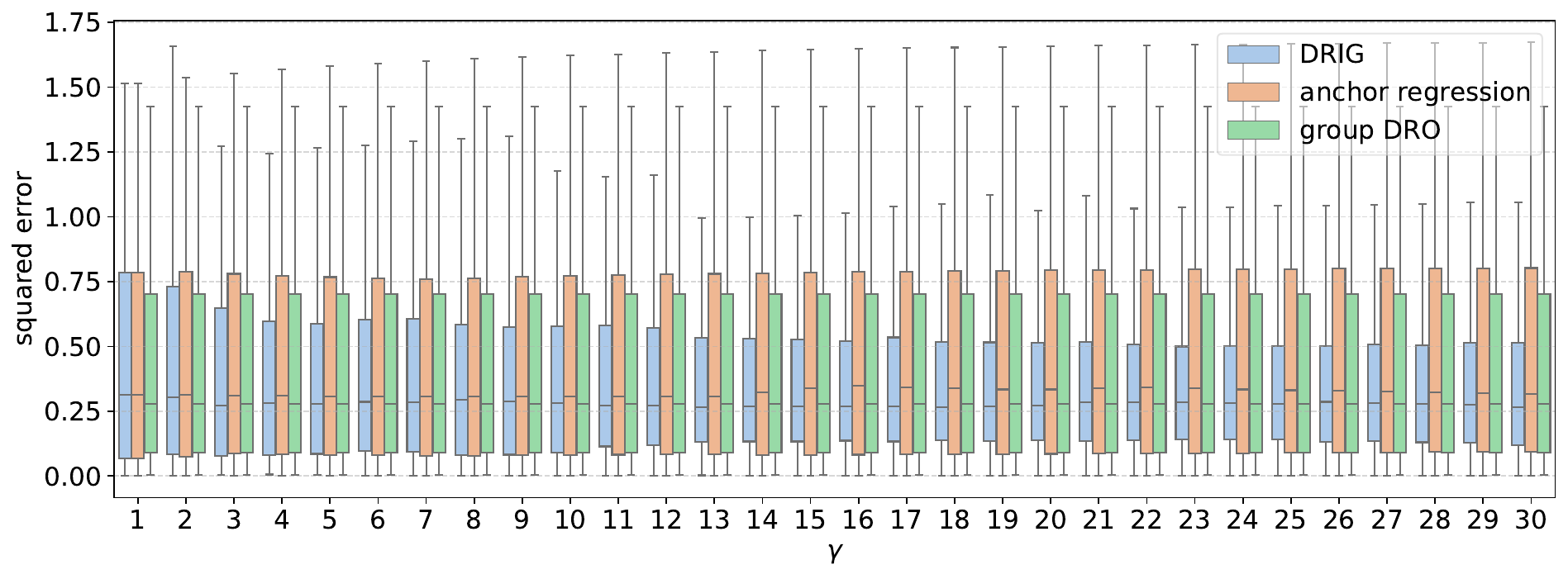} \vspace{-0.2in}\\
	\end{tabular}
	\caption{\small Results for ICU data. (\textbf{left}) MSE on the test environment as a function of the tuning parameter $\gamma$ for each method. (\textbf{right}) the squared prediction error of DRIG, anchor regression, and group DRO across each individual in the test environment for different values of $\gamma$. }\label{fig:icu_results}
\end{figure}

\section{Conclusion and future work}
We proposed DRIG, a procedure that exploits general noise interventions to obtain distributionally robust prediction models. While DRO formalizes robustness based on a postulated distance measure, DRIG is based on causal modeling and focuses on structural and data-dependent distribution shifts. 


A number of interesting future directions arise from our work. First, while we established robustness guarantees for DRIG in linear settings, the gradient invariance principle, as well as the DRIG formulation, are general and also applicable to nonlinear models with some promising numerical results (see Section~\ref{sec:method}). Hence, investigating nonlinear extensions would be of significant interest. 
Second, DRIG can produce predictions in the form of point estimates; incorporating uncertainty with corresponding prediction intervals would further expand the applicability of our methods.

%% file: parts/appendix.tex
\renewcommand{\thesection}{\Alph{section}}
\setcounter{section}{0}

{\section{DRIG without the observational assumption}}
\label{Sec:drig_without_observational}
\begin{assumption}
There are environment(s) $\mathcal{E}_\mathrm{small} \subset \mathcal{E}$ such that: $S^{e'} \preceq S^e$ for every $e' \in \mathcal{E}_\mathrm{small}$ and $e \in \mathcal{E}\setminus\mathcal{E}_\mathrm{small}$, and for every $e' \in \mathcal{E}_\mathrm{small}$, $\sum_{e\in\mathcal{E}}\omega^e(S^e-S^{e'}) \succeq 0$.
\label{ass:strength_2}
\end{assumption}
Assumption~\ref{ass:strength_2} ensures that the set of environments $\mathcal{E}$ can be divided into two: $\mathcal{E}_\mathrm{small}$ and $\mathcal{E}\setminus\mathcal{E}_\mathrm{small}$ where the interventions in $\mathcal{E}\setminus\mathcal{E}_\mathrm{small}$ are sufficiently stronger than those in $\mathcal{E}_\mathrm{small}$. A special case of the aforementioned setting is when there exists an `observational' environment $0 \in \mathcal{E}$ with $S^0 \preceq S^e$ for every $e\in\mathcal{E}$, which is a common condition in the causal inference literature, although Assumption~\ref{ass:strength} much less restrictive. Letting $Z^e = (X^e,Y^e)$,  Assumption~\ref{ass:strength} can be expressed in terms of the Gram matrix of the observed data, namely: $\mathbb{E}[Z^{e'}{Z^{e'}}^\top] \preceq \mathbb{E}[Z^{e}{Z^{e}}^\top]$ for all $e'\in\mathcal{E}_\text{small}$ and $e \in \mathcal{E}\setminus\mathcal{E}_\mathrm{small}$, and for every $e' \in \mathcal{E}_\mathrm{small}$, $\sum_{e\in\mathcal{E}}\omega^e(\mathbb{E}[Z^e{Z^e}^\top]-\mathbb{E}[Z^{e'}{Z^{e'}}^\top]) \succeq 0$.

\subsection{Convexity of DRIG}
\label{sec:drig_convexity_without_observational}
\begin{theorem}Suppose Assumption~\ref{ass:strength_2} is satisfied. Then, for any $\gamma \geq 1$, the DRIG objective is convex.
\label{thm:convexity_general}
\end{theorem}
We prove Theorem~\ref{thm:convexity_general} in Appendix~\ref{proof_convex}. Note that the convexity of DRIG holds without assuming the linear structural equation model \eqref{eq:scm_train}, as long as $\mathbb{E}[Z^{e'}{Z^{e'}}^\top] \preceq \mathbb{E}[Z^{e}{Z^{e}}^\top]$ for all $e'\in\mathcal{E}_\text{small}$ and $e \in \mathcal{E}\setminus\mathcal{E}_\mathrm{small}$, and for every $e' \in \mathcal{E}_\mathrm{small}$, $\sum_{e\in\mathcal{E}}\omega^e(\mathbb{E}[Z^e{Z^e}^\top]-\mathbb{E}[Z^{e'}{Z^{e'}}^\top]) \succeq 0$.

\subsection{Robustness guarantees of DRIG}
\label{sec:drig_robustness_without_observational}

 For notational simplicity, for any $\bar{e}\in\mathcal{E}$, we define $\mathcal{L}_\gamma^\mathrm{\bar{e}}(b) := \sum_{e \in \mathcal{E}}\omega^e (\gamma\mathbb{E}[\ell(X^{e},Y^e;b)] +(1-\gamma)\mathbb{E}[\ell(X^{\bar{e}},Y^{\bar{e}};b)])$ and $b^\mathrm{opt,\bar{e}}:= \argmin_{b}\mathcal{L}_\gamma^\mathrm{\bar{e}}(b)$. The following theorem assesses the robustness of the DRIG prediction model. 
{\begin{theorem}Let $\bar{e} \in \argmin_{e\in\mathcal{E}_\mathrm{small}} \mathcal{L}_\gamma^\mathrm{{e}}(b^{\mathrm{opt,{e}}})$. If $\bar{e} \in\argmin_{e\in\mathcal{E}_{\mathrm{small}}}\mathbb{E}[\ell(X^e,Y^e; b^{\mathrm{opt,\bar{e}}})]$, then, the DRIG estimator $b_\gamma^\mathrm{opt}$ is the minimizer of \eqref{eqn:worst-case} with $\mathcal{C} = \mathcal{C}_\mathrm{DRIG}^\gamma$, where:
$$\mathcal{C}_\mathrm{DRIG}^\gamma:=  \left\{v \in \mathbb{R}^{p+1}: \bbE[vv^\top ] \preceq S^{\bar{e}} + \gamma\sum_{e \in \mathcal{E}} \omega^e\left(S^e - S^{\bar{e}}\right)\right\}.$$
\label{thm:main1}
\end{theorem}
 We prove Theorem~\ref{thm:main1}  in Supplementary~\ref{proof_thm_main}. This result states that under some assumptions, DRIG protects against noise interventions $v$ that are in the set $\mathcal{C}_\mathrm{DRIG}^\gamma$. The assumptions of Theorem~\ref{thm:main1} are strictly weaker than the observational assumption; In Appendix \ref{sec:numerical_ass}, we numerically illustrate settings where the Assumptions of Theorem~\ref{thm:main1} are satisfied but an observational condition is not satisfied.}

\subsection{Connections to causal parameter}
\label{sec:causal_weaker_assumptions}
We analyze DRIG when $\gamma \to \infty$, and highlight how infinite robustness (as guaranteed by Theorem~\ref{thm:main}) connects to causality and invariance. For every $\bar{e}\in\mathcal{E}$, define $L^{\star,\bar{e}} := \sum_{e \in \mathcal{E}} \omega^e (S^e-S^{\bar{e}})$ and $C^\star = (I-B^\star)^{-1}$ with block forms, 
$ L^{\star,\bar{e}} = \begin{pmatrix}L^{\star,\bar{e}}_x & L^{\star,\bar{e}}_{xy} \\
		{L^{\star,\bar{e}}_{xy}}^\top & L^{\star,\bar{e}}_{y}
	\end{pmatrix}$ and $C^\star = \begin{pmatrix}C^\star_{x} & C^\star_{xy} \\ C^{\star\top}_{yx} & C^\star_{y} \end{pmatrix},$ where $L^{\bar{e},\star}_y,C^\star_y \in \mathbb{R}$. The following theorem characterizes the solution of DRIG with $\gamma \to \infty$, denoted by $b_\infty^{\mathrm{opt}}:=\lim_{\gamma\to\infty}b_\gamma^{\mathrm{opt}}$.

\begin{theorem}We have 
\begin{equation}\label{eq:b_insufficient_2}
        b_\infty^{\mathrm{opt}} = \argmin_{b\in\mathcal{I}}\min_e\bbE[(Y^e-b^\top X^e)^2],
\end{equation}
where $\mathcal{I} := \{b \in \mathbb{R}^{p}: b \text{ satisfies the gradient invariance condition in Definition~\ref{def:gi}}\}$ is a non-empty set. If additionally $\mathrm{rank}([C^\star{L^{\star,\bar{e}}}{C^\star}^\top]_{1:p,1:p}) = p$ for every $\bar{e}\in\mathcal{E}_\mathrm{small}$ and ${L^{\star,\bar{e}}_{xy}}= 0 = L^{\star,\bar{e}}_y$, then $b_\infty^{\mathrm{opt}} = b^\star$.
\label{thm:solution_infty_2}
\end{theorem}

We prove Theorem~\ref{thm:solution_infty_2} in Appendix~\ref{proof_prop_solution_infty}.

\section{Finite-sample consistency guarantees of DRIG}
\label{sec:finite_sample}
 Note that $\hat{b}_\gamma$ in \eqref{eqn:finite_drig_objective} is as an estimate for the population parameter $b^\mathrm{opt}_\gamma$, and $\hat{\mathcal{L}}_\gamma(\hat{b}_\gamma)$ is an estimate for $\mathcal{L}_\gamma(b^\mathrm{opt}_\gamma)$, which according to Theorem~\ref{thm:main}, is the worst-case risk over a class of noise interventions.

We provide finite-sample consistency guarantees for the finite-sample DRIG estimator. Specifically, we demonstrate convergence of prediction models $\|\hat{b}_\gamma-b^\mathrm{opt}_\gamma\|_2$ as well worst-case loss functions $|{\mathcal{L}}_\gamma(b_{\gamma}^\mathrm{opt})-\hat{\mathcal{L}}_\gamma(\hat{b}_\gamma)|$. For simplicity, we assume that the random variable $\varepsilon^e$ in \eqref{eq:scm_train} is Gaussian, although the analysis can readily be extended to sub-Gaussian distributions. We let $\psi_e$ be the spectral norm of the joint Gram matrix of $(X^e,Y^e)$. Let $\psi_{\max}=\max_{e\in\cE}\psi_e$ and $n_{\min}=\min_{e\in\cE}n_e$. Furthermore, let $\tau_{\mathrm{min}}$ be the minimum eigenvalue of the matrix $\sum_{e\in\mathcal{E}}\omega^e[G^{e} - \frac{(\gamma-1)}{\gamma}G^{0}]$ where $G^{e}$ is the second moment of the vector $(X^e,Y^e)$. 
\begin{theorem} Suppose $n_e \geq p\max\{1,64\psi_e^2,\frac{64}{\min\{\tau_{\mathrm{min}},1\}^2}(\|b^\mathrm{opt}_\gamma\|_2+1)(\max_e\psi_e+1)^2\}$ for all $e \in \mathcal{E}$. Then with probability exceeding $1-|\mathcal{E}|\exp(-p/2)$, for any $\gamma \geq 1$, we have $\|\hat{b}_\gamma - b^\mathrm{opt}_\gamma\|_2 \leq \frac{32(\|b^\mathrm{opt}_\gamma\|_2+1)}{\min\{\tau_{\mathrm{min}},1\}}(1+\psi_{\max})\sqrt{\frac{p}{n_{\min}}}$ and $|\hat{\mathcal{L}}_\gamma(\hat{b}_\gamma) - {\mathcal{L}}_\gamma(b_{\gamma}^{opt})| \leq \frac{480(\|b^\mathrm{opt}_\gamma\|_2+1)^3}{\min\{\tau_{\mathrm{min}},1\}}(1+\psi_{\max})^2\gamma\sqrt{\frac{p}{n_{\min}}}$.
\label{thm:theorem_finite}
\end{theorem}
The proof of Theorem~\ref{thm:theorem_finite} is presented in Supplementary~\ref{proof_finite_sample}. {Note the scaling with the factor $\sqrt{p/n_e}$ in the second statement. This is due to the fact that $b_{\gamma}^{\mathrm{opt}}$ gives residuals which are not independent nor orthogonal (in population) to the covariates $X$.}

\section{DRIG-A robustness guarantees}
\label{sec:robustness_drig-A}
Consider first the population setting where we have access to the distribution of the training enviornments, although the number of test-samples may be finite. Let $\mathcal{D}_\mathrm{test}$ be the test samples (both labeled and unlabeled). Then, the optimal weights $\tilde{\omega}^e$ that are estimated by DRIG-A can be expressed as:
$$\tilde{\omega}^e = f_e(\{G^e\}_{e\in\mathcal{E}};\mathcal{D}_\mathrm{test}),$$ 
for some function $f_e$. Let:
$$\mathcal{L}_\mathrm{DRIG-A}(b) = \min_{e\in\mathcal{E}}\mathbb{E}[\ell(X^e,Y^e;b)] + \sum_{e\in\mathcal{E}}\tilde{\omega}^e(\mathbb{E}[\ell(X^e,Y^e;b)]-\min_{e\in\mathcal{E}}\mathbb{E}[\ell(X^e,Y^e;b)],$$
be the population DRIG objective after plugging in $\tilde{\omega}^e$, and 
\begin{eqnarray}
\begin{aligned}
{b}_\mathrm{DRIG-A} := \argmin_b \mathcal{L}_\mathrm{DRIG-A}(b).
\end{aligned}
\end{eqnarray}
\begin{theorem}The estimator ${b}_\mathrm{DRIG-A}$ is the solution the worst-case risk \eqref{eqn:worst-case} with $\mathcal{C} = \mathcal{C}_\mathrm{DRIG-A}$ where,
$\mathcal{C}_\mathrm{DRIG-A} = \left\{v:\mathbb{E}[vv^\top] \preceq S^0 +  \sum_{e\in\mathcal{E}}f_e(\{G^e\}_{e\in\mathcal{E}};\mathcal{D}_\mathrm{test})(S^e-S^0)\right\}$.
\label{thm:drig-A-robust}
\end{theorem}
The proof of Theorem~\ref{thm:drig-A-robust} is similar to that of Theorem~\ref{thm:main} and is left out for brevity. Note that $i)$ the result depends on access to full training distributions, $ii)$ the robustness set is random (as it depends on finite test samples). To have a finite-sample result, consider:
\begin{eqnarray*}
\begin{aligned}
\tilde{\omega}^{e,\star}_{\mathrm{opt}} := \argmin_{\{\tilde{\omega}^e\}_{e\in\mathcal{E}} \geq 0} n_u\|J_{1:p,1:p} - G^v_{1:p,1:p}\|_F^2 + n_l(2\|J_{1:p,p+1} - {G}^v_{1:p,p+1}\|_F^2+\|{J}_{p+1,p+1} - {G}^v_{p+1,p+1}\|_F^2),
\end{aligned}
\end{eqnarray*}
where $n^l$ and $n^u$ may be viewed as controlling the mixture proportion of labeled samples vsunlabed samples in population. Here, $G^v = \mathbb{E}[vv^\top]$. Let, 
\begin{eqnarray*}
\begin{aligned}
\mathcal{L}_\mathrm{DRIG-A}^\star(b) &:= \min_{e\in\mathcal{E}}\mathbb{E}[\ell(X^e,Y^e;b)] + \sum_{e\in\mathcal{E}}\tilde{\omega}^{e,\star}_{\mathrm{opt}} (\mathbb{E}[\ell(X^e,Y^e;b)]-\min_{e\in\mathcal{E}}\mathbb{E}[\ell(X^e,Y^e;b)],\\
{b}^\star_\mathrm{DRIG-A} &:= \argmin_b \mathcal{L}_\mathrm{DRIG-A}^\star(b).
\end{aligned}
\end{eqnarray*}
Let $\hat{G}^e$ be the empirical Gram matrix of $(X^e,Y^e)$. Consider the empirical analog of the optimization procedure for estimating $\tilde{w}^e$:
\begin{eqnarray*}
\begin{aligned}
\hat{\tilde{\omega}}^{e}_{\mathrm{opt}} := \argmin_{\{\tilde{\omega}^e\}_{e\in\mathcal{E}} \geq 0} n_u\|\hat{J}_{1:p,1:p} - \hat{G}^v_{1:p,1:p}\|_F^2 + n_l(2\|\hat{J}_{1:p,p+1} - \hat{G}^v_{1:p,p+1}\|_F^2+\|\hat{J}_{p+1,p+1} - \hat{G}^v_{p+1,p+1}\|_F^2).
\end{aligned}
\end{eqnarray*}
Here, $\hat{J} = \hat{G}^0 + \sum_{e\in\mathcal{E}}\tilde{\omega}^e(\hat{G}^e-\hat{G}^0)$, with $\hat{G}^e$ representing the empirical Gram matrix of the data in enviornment $e$. Then, finite-sample DRIG would minimize:
\begin{eqnarray*}
\begin{aligned}
\hat{\mathcal{L}}_\mathrm{DRIG-A}(b) &:= \min_{e\in\mathcal{E}}\hat{\mathbb{E}}\ell(X^e,Y^e;b) + \sum_{e\in\mathcal{E}}\hat{\tilde{\omega}}^{e}_{\mathrm{opt}} (\hat{\mathbb{E}}[\ell(X^e,Y^e;b)]-\min_{e\in\mathcal{E}}\hat{\mathbb{E}}[\ell(X^e,Y^e;b)],\\
\hat{b}_\mathrm{DRIG-A} &:= \argmin_b \hat{\mathcal{L}}_\mathrm{DRIG-A}(b).
\end{aligned}
\end{eqnarray*}
As $n^e \to \infty$ and $n^u,n^l \to \infty$, then, $\hat{\tilde{\omega}}^{e}_{\mathrm{opt}} \to {\tilde{\omega}}^{e,\star}_{\mathrm{opt}}$, and an empirical average converges to the corresponding expected value. As a result, $\hat{b}_\mathrm{DRIG-A} \to {b}^\star_\mathrm{DRIG-A}$ and $\hat{\mathcal{L}}_\mathrm{DRIG-A}(\hat{b}_\mathrm{DRIG-A}) \to \mathcal{L}^\star_\mathrm{DRIG-A}({b}^\star_\mathrm{DRIG-A})$. Standard finite sample analysis yields the following convergence rates:
\begin{eqnarray*}
\begin{aligned}
 \|\hat{b}_{\mathrm{DRIG-A}}- b^\star_\mathrm{opt}\|_2 &\leq \mathcal{O}(p|\mathcal{E}|/\sqrt{n_\text{min}}),\\
 \hat{\mathcal{L}}_\mathrm{DRIG-A}^\mathrm{opt}(\hat{b}_{\mathrm{DRIG-A}})-\mathcal{L}^\star(b^\star_\mathrm{opt}) &\leq \mathcal{O}(p|\mathcal{E}|/\sqrt{n_\text{min}}),
 \end{aligned}
 \end{eqnarray*}
where $n_\text{min} = \min\{\min_e n^e, n^u,n^l\}$. We omit the proof for brevity.

\section{Illustrations of the perturbation class}\label{app:illus_perturb}
We provide some illustrations of the perturbation class that DRIG is robust against, i.e.,
\begin{equation*}
    \mathcal{C}_\mathrm{DRIG}^{\gamma} = \left\{v \in \mathbb{R}^{p+1}: \bbE[vv^\top ] \preceq S^0 + \gamma\sum_{e \in \mathcal{E}} \omega^e\left(S^e - S^0\right)\right\}.
\end{equation*}
The column space of the matrix $U:=\sum_{e \in\mathcal{E}}\omega^e(S^e-S^0)$ represents the ``directions" of the perturbations that DRIG protects against with a controllable strength via $\gamma$. Specifically, denote by $U=Q\Lambda Q^\top$ the spectral decomposition of $U$, where $\Lambda=\mathrm{diag}(\lambda_1,\dots,\lambda_{p+1})$ and $Q=(q_1,\dots,q_{p+1})$ with $(\lambda_i,q_i)$ being an eigenvalue/eigenvector pair. Let $r=\mathrm{rank}(U)$ be the rank of the matrix $U$ so that $\lambda_i>0$ for $i\le r$ and $\lambda_i=0$ for $i\ge r+1$; here, the eigenvectors $q_1,\dots,q_r$ span the column space of $U$. Then for all $v\in\mathcal{C}_\mathrm{DRIG}^{\gamma}$, we have $\bbE[Q^\top v(Q^\top v)^\top]\preceq \Lambda$, implying $q_i^\top v\equiv 0$ for $i\ge r+1$.  That is, the DRIG estimator can only be robust to perturbations that lie in the column space of $U$. As such, the larger the dimension of this column space, the more directions the DRIG estimator is robust against. 

\begin{example}
\label{ex:1}
	We consider two covariates $X_1,X_2$, uniform weights $\omega^e\equiv1/|\cE|$, and interventions on only the covariates. We first assume there is one interventional environment $e=1$ apart from the observational environment $e=0$ with $\delta^0=0$, where both covariates are perturbed. If only the mean is affected (the anchor regression setting), i.e., $\delta^1=\mu^1$ for some deterministic vector $\mu^1\ne0$, we have $U=\gamma\mu^1{\mu^1}^\top/2$ with rank 1. If the variance is affected, i.e., $\delta^1\sim\cN(\mu^1,S^1)$ we have $U=\gamma(\mu^1{\mu^1}^\top+S^1)/2$ which is in general full-rank. The perturbations that we are potentially robust against in the cases of mean shifts and variance shifts are depicted in Figures~\ref{fig:pert_direction}(a) and \ref{fig:pert_direction}(b), respectively.
	
	Next, we assume that only $X_1$ is perturbed in the interventional environment $e=1$. Thus $\delta^1_2=0$, and the matrix  $U=\gamma(\mu^1{\mu^1}^\top+S^1)/2$ has rank equal to one since the second diagonal entry is zero. Now if we have one more interventional environment $e=2$ where $X_2$ receives an intervention, it then holds that the matrix $U=\gamma(\mu^1{\mu^1}^\top+S^1 + \mu^2{\mu^2}^\top+S^2)/3$ is full-rank. The perturbations that DRIG is robust against in these two cases also follow the same pattern as in Figures~\ref{fig:pert_direction}(a) and \ref{fig:pert_direction}(b), respectively.
\end{example}

\begin{figure}[h]
\centering
\begin{tabular}{cc}
	\includegraphics[page=1, clip, trim=28cm 13cm 28cm 18cm, width=0.25\textwidth]{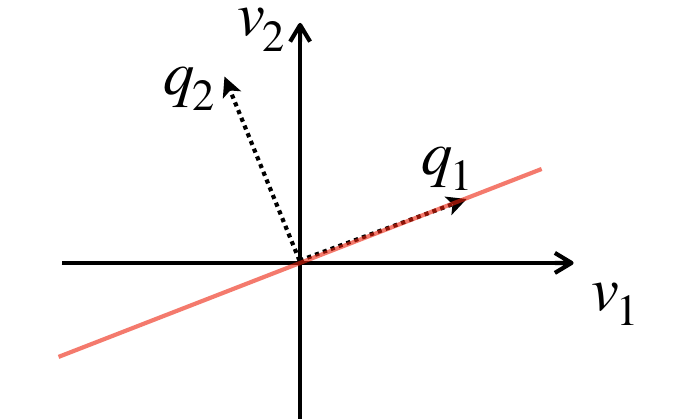} \hspace{0.5in} &
	\includegraphics[page=2, clip, trim=28cm 13cm 28cm 18cm, width=0.25\textwidth]{fig/perturb-direction}\vspace{-0.1in}\\
	\small{(a) mean shifts or fewer environments} \hspace{0.5in} & \small{(b) variance shifts or more environments}
\end{tabular}
\caption{\small Perturbations that DRIG is {controllably} robust against in scenarios in Example~\ref{ex:1}.}
\label{fig:pert_direction}
\end{figure}

\section{Incorporating continuous exogenous variables}
\label{sec:cont_anchor}
Our modeling framework \eqref{eq:scm_train} contains only interventions through a discrete anchor (environment) variable $E$. We can also incorporate interventions due to continuous anchor variables $A$ which are exogenous. Specifically, for every environment $e \in \mathcal{E}$, the data $(X^e,Y^e)$ is generated according to the following modified SCM:
\begin{equation}\label{eq:scm_anc_train}
    \begin{pmatrix}
    X^e \\ Y^e \\H^e 
    \end{pmatrix}
    =
    \tilde{B}^\star
    \begin{pmatrix}
    X^e \\ Y^e\\H^e
    \end{pmatrix} 
    + \varepsilon^e + MA^e,
\end{equation}
with the matrix $I - \tilde{B}^\star$ being invertible. Here, $A^e$ denotes the observed continuous anchor variable in environment $e$ with $A^e$ being a random variable following the conditional distribution of $A$ given $E=e$. For every $e\in\mathcal{E}$, $(\varepsilon^e,A^e)$ are jointly independent. Figure~\ref{fig:graph}(c) presents the graphical perspective of the model \eqref{eq:scm_anc_train}; the variables $A$ and $E$ are exogenous and cannot be descendants of any of the variables $(X,Y,H)$. 

For every environment $e \in \mathcal{E}$, we define $\tilde{Y}^e=Y^e-\bbE[Y^e|A^e]$ and $\tilde{X}^e=X^e-\bbE[X^e|A^e]$. The population version of the modified DRIG estimator (to account for continuous anchors) is
\begin{equation}
	{b}^{\mathrm{opt}}_{\lambda,\gamma} = \argmin_b \mathcal{L}_{\lambda,\gamma}(b).
\label{eqn:mod_estimator}
\end{equation}
Here, $\lambda,\gamma\geq 0$ are regularization parameters and the objective $\cL_{\lambda,\gamma}(b)$ is
\begin{equation*}
	\cL_{\lambda,\gamma}(b):= \tilde{\cL}_\gamma(b) +\lambda\sum_{e \in \mathcal{E}}\omega^e\bbE[\bbE(Y^e-b^\top X^e|A^e)]^2,
 \label{eqn:modified_objective}
\end{equation*}
where $\tilde{\cL}_\gamma(b)$ is the original DRIG objective function in \eqref{eq:obj} applied to the transformed data $(\tilde{X}^e,\tilde{Y}^e)$.

\subsection{Robustness guarantees with discrete and continuous exogenous variables}
Above we introduced a generalization of DRIG \eqref{eqn:mod_estimator} for incorporating both discrete and continuous exogenous variables. We now assess the robustness of this estimator, and establish once again that our estimator has stronger robustness guarantees than anchor regression. Throughout, we suppose that the training data is generated according to the SCM \eqref{eq:scm_anc_train} and the test data is generated according to the SCM \eqref{eq:scm_test}. Let $\tilde{S}^e=\bbE[\varepsilon^e{\varepsilon^e}^\top|A^e]$.  For simplicity, we also assume there is an observational environment $0\in\mathcal{E}$ with $\tilde{S}^0 \preceq \tilde{S}^e$ for every $e\in\mathcal{E}$.
\begin{theorem}
The modified DRIG estimator 
${b}^{\mathrm{opt}}_{\lambda,\gamma}$ in \eqref{eqn:mod_estimator} is the minimizer of the distributional robust objective \eqref{eqn:worst-case} with $\mathcal{C} = \mathcal{C}^{\lambda,\gamma}_\mathrm{DRIG}$, where 
\begin{equation*}
\C^{\lambda,\gamma}_\mathrm{DRIG} = \left\{v\in\mathbb{R}^{p+1}:\bbE[vv^\top]\preceq \tilde{S}^0+\sum_{e\in\cE}\omega^e\Big[\gamma(\tilde{S}^e - \tilde{S}^0) +\lambda\Big(\bbE[\bbE(\varepsilon^e|A^e)\bbE(\varepsilon^e|A^e)^\top] +  M\bbE[A^e{A^e}^\top]M^\top\Big)\Big]\right\}.
\end{equation*}
\label{thm:continous_discrete}
\end{theorem}
We prove Theorem~\ref{thm:continous_discrete} in Supplementary~\ref{proof_thm_discrete_continuous}. This result states that the modified DRIG estimator
$b^{\mathrm{opt}}_{\lambda,\gamma}$ protects against perturbations in the class $C^{\lambda,\gamma}_\mathrm{DRIG}$. Notice that if the environment (discrete) variables $E$ are independent of the continuous anchors $A$, then the perturbation class simplifies to $$\C^{\lambda,\gamma}_\mathrm{DRIG} = \{v \in \mathbb{R}^{p+1}: \bbE[vv^\top]\preceq\allowbreak S^0+\sum_{e\in\mathcal{E}}\omega^e\allowbreak\left(\allowbreak\gamma(S^e-S^0)+\lambda\mu^e{\mu^e}^\top + \lambda M\bbE[AA^\top]M^\top\right)\},$$
where $\mu^e = \mathbb{E}[\varepsilon^e|A]$. Furthermore, when there are no continuous anchors, we recover the result of Theorem~\ref{thm:main}. 

The anchor regression estimator \eqref{eqn:method_anchor} proposed in \cite{rothenhausler2021anchor} can be applied to data generated according to the model \eqref{eq:scm_anc_train}. Appealing to Theorem 1 of \cite{rothenhausler2021anchor}, we can conclude that anchor regression with turning parameter $\lambda$ protects against perturbations in the set 
$$\C_{\mathrm{anchor}}^{\lambda} = \left\{v \in \mathbb{R}^{p+1}: \mathbb{E}[vv^\top] \preceq \sum_{e\in\cE}\omega^e\Big[\tilde{S}^e +\lambda\Big(\bbE[\bbE(\delta^e|A^e)\bbE(\delta^e|A^e)^\top] +  M\bbE[A^e{A^e}^\top]M^\top\Big)\Big]\right\}.$$
Thus, analogous to the discrete exogenous setting, our estimator \eqref{eqn:mod_estimator} in the continuous and discrete exogenous setting is robust against strictly more directions than those protected by anchor regression as $\C^{\lambda,\gamma}_\mathrm{DRIG} \supseteq \C_{\mathrm{anchor}}^{\lambda}$.

\section{Connections to other invariance notions}\label{app:invariance}

We devote a comprehensive discussion on existing notions of invariance in the literature, and how they are related to the gradient invariance notion in our work. Throughout, we assume that the data is generated according to the SCM \eqref{eq:scm_train}.
 
The notion of invariance dates back to \citet{haavelmo1943statistical} who realized the invariant property of the causal variables. Formally, a subset $\cS\subseteq\{1,\dots,p\}$ of covariates is said to be \emph{conditionally invariant} if the distribution of the response $Y^e$ given $X^e_\cS$ is the same for all $e\in\cE$. In the SCM \eqref{eq:scm_train}, when there are no interventions on $Y$ or $H$ so that the distribution of $\varepsilon_y^e$ is the same for all $e\in\cE$, the parental set of $Y$, denoted by $\mathrm{pa}(Y)$, satisfies the conditional invariance in that $Y^e|X^e_{\mathrm{pa}(Y)}$ is the same for all $e\in\cE$. This property was explored in the reverse direction by \citet{peters2016causal} for discovering the parental set of $Y$. However, the conditional invariance may sometimes fail to identify the causal parameter; in particular, the conditional invariance property does not hold for the causal parameter when $X$ and $Y$ are confounded by a latent variable \citep{Rothenhausler2017CausalDF}. In recent literature, several alternative notions of invariance have been proposed; these are then used for causal discovery or distributional robustness. Below we list several representatives followed by a discussion. 


The first alternative proposed in \citet{arjovsky2019invariant} looks at the invariance of the conditional mean or the solution of $L_2$ risk minimization within each environment, instead of the conditional distribution. Formally, a subset $\cS^*\subseteq\{1,\dots,p\}$ of covariates is said to be \emph{solution invariant} if there exists $b^*\in\bbR^p$ supported on $\cS$ such that
\begin{equation*}
	b^*\in\argmin_{b}\bbE[\ell(X^{e}_{\cS^*},Y^e; b)],\quad\forall e\in\cE
\end{equation*}
where $X_\cS\in\bbR^p$ denotes the random vector that copies the coordinates of $X$ in $\cS$ and has zero components elsewhere. 
Based on this notion of invariance, \citet{arjovsky2019invariant} then proposed a method called invariant risk minimization (IRM) for out-of-distribution generalization. In the variable selection setting, IRM interpolates between the pooled OLS and solution invariance.
Formally, IRM solves the following problem
\begin{equation*}
	\min_{\cS,b}\left\{\frac{1}{|\cE|}\sum_{e\in\cE}\bbE[\ell(\tilde{X}^{e}_\cS,Y^e; b)] + \frac{\lambda}{|\cE|}\sum_{e\in\cE}\left[\bbE[\ell(\tilde{X}^{e}_\cS,Y^e; b)]-\min_{b'}\bbE[\ell(\tilde{X}^{e}_\cS,Y^e; b')]\right]\right\},
\end{equation*}
 where $\lambda$ is a hyperparamter that controls the regularization strength with $\lambda\to\infty$ enforcing the solution invariance, whenever it is achievable.

Apart from the conditional distribution and the conditional mean, another alternative considers the invariance of the risk of a prediction model from $X$ to $Y$. Specifically, a regression coefficient $b\in\bbR^p$ is said to fulfill \emph{risk invariance} if  
	the risk $\bbE[(Y^e-b^\top X^e)^2]$ is the same for all $e\in\cE$.
\citet{krueger2021out} proposed to regularize the pooled OLS towards risk invariance:
\begin{equation*}
	\min_b \frac{1}{|\cE|}\sum_{e\in\cE}\bbE[\ell(X^e,Y^e; b)] + \lambda\mathrm{Var}(\{\bbE[\ell(X^e,Y^e; b)]:e\in\cE\}),
\end{equation*}
where $\mathrm{Var}$ denotes here the empirical variance over all $e \in \cE$.

The last notion of invariance that we would like to highlight is the most closely related to our gradient invariance. We say a regression coefficient $b\in\bbR^p$ satisfies \emph{full gradient invariance} if 
	$\nabla_b\bbE[\ell(X^e,Y^e;b)]$ is the same for all $e\in\cE$.
This notion was introduced by \citet{Rothenhausler2017CausalDF} with the name inner-product invariance since in linear models, inner-product invariance is equivalent to $\bbE[X^e(Y^e-b^\top X^e)]$ being the same for all $e\in\cE$. 
\citet{Rothenhausler2017CausalDF} then proposed the causal Dantzig to identify the causal parameter by exploiting full gradient invariance in the setting with two environments. As we have seen earlier, DRIG with $\gamma\to\infty$ and $|\mathcal{E}|=2$ recovers the causal Dantzig.
A similar invariance notion was also explored in the context of out-of-distribution generalization by \citet{koyama2020out,shi2021gradient,Ram2021FishrIG}. Specifically, the authors in \citet{koyama2020out,Ram2021FishrIG} propose the following formulation
\begin{equation*}
	\min_b \frac{1}{|\cE|}\sum_{e\in\cE}\bbE[\ell(X^e,Y^e; b)] + \lambda\mathrm{trace}(\mathrm{Var}(\{\nabla_b\bbE[\ell(X^e,Y^e; b)]:e\in\cE\})),
\end{equation*} 
enforcing full gradient invariance when the regularization parameter $\lambda$ tends to infinity. \citet{shi2021gradient} also enforces full gradient invariance via regularization based on the inner products among pairs of gradients. Our gradient invariance in Definition~\ref{def:gi} is a relaxed version of the full gradient invariance. In particular, instead of enforcing the gradients in all environments to be the same, we require only a weighted average of the gradients to be stable in the sense of equaling the gradient in the reference environment. Thus, gradient invariance is strictly weaker than the full gradient invariance except when there are two environments, where the two notions are identical.

Under data generated according to the linear SCM \eqref{eq:scm_train}, among all the preceding invariance notions, our notion of gradient invariance necessitates the weakest conditions to identify the causal parameter (see below for a more detailed discussion). Furthermore, although gradient invariance is not strictly satisfied with a finite regularization parameter $\gamma$, DRIG achieves distributional robustness against moderate interventions. In contrast, all the aforementioned methods do not have finite robustness guarantees.


\subsection{Necessary conditions for invariance conditions to identify the causal parameter}
We discuss the necessary conditions for the above notions of invariance to identify the causal parameter under the linear SCM with multiple environments $e\in\cE$ and additive interventions, which is a special case of \eqref{eq:scm_train}:
\begin{equation*}
    \begin{pmatrix}
    X^e \\ Y^e
    \end{pmatrix}
    =
    B^\star
    \begin{pmatrix}
    X^e \\ Y^e
    \end{pmatrix} 
    + \varepsilon + \delta^e.
\end{equation*}
That is, we investigate when the causal parameter satisfies a certain type of invariance. We summarize the conclusions in Table~\ref{tab:invariance} and the following, which indicates that our gradient invariance requires the weakest conditions among all. The proof is given below.
\begin{itemize}
	\item In the simplest case without latent confounder and intervention on $Y$, all invariance conditions true for the causal parameter. Additionally under some sufficient conditions, e.g. when there are sufficient interventions on $X$ as illustrated in Section~\ref{sec:case1},  and all methods can identify the causal parameter. The existence of latent confounders and interventions on $Y$ bring in complications for causal identification. 
	\item When there are latent confounders, the conditional and solution invariance fail to hold for the causal parameter, while the risk and gradient invariance remain valid if $Y$ is not intervened on.
	\item Interventions on $Y$ causes even more trouble, under which only the full and our gradient invariance can be fulfilled by the causal parameter under some conditions on the interventions and the structural relationship between $Y$ and $X$. Compared to the full gradient invariance that requires the inner-product of interventions to be exactly the same across all environments, our gradient invariance requires a strictly weaker condition in that in some environments, the interventions on $X$ and $Y$ could have different correlations, although their weighted average has to be stable. In addition, the full gradient invariance does not allow $Y$ to have children in $X$, that is, the structural relationship from $X$ to $Y$ can only be causal rather than anti-causal. Intuitively, this protects the gradients from varying due to interventions on $Y$ that does not propagate to some of $X$. Nevertheless, our gradient invariance could relax this assumption if $Y$ is intervened in a stable way across environments.
\end{itemize}

\begin{table}
\centering
\caption{Whether the causal parameter satisfies a certain type of invariance under different cases of interventions and latent effects. Superscripts means some additional conditions are needed. 1: $\bbE[\Delta_x^e\delta_y^e]$ is the same for all $e\in\cE$ and $B^\star_{yx}=0$ (i.e., $Y$ is childless in $X$); 2: $\sum_{e\in\cE}\omega^e\bbE[\Delta_x^e\delta_y^e]=\bbE[\Delta_x^0\delta_y^0]$ and $\sum_{e\in\cE}\omega^e\bbE[{\delta_y^e}^2]=\bbE[{\delta_y^0}^2]$ or $B^\star_{yx}=0$.}\vskip 0.1in
\label{tab:invariance}
\begin{tabular}{c|CC|CC}
\toprule
intervention on $Y$ & \multicolumn{2}{c|}{$\delta_y^e=0$} & \multicolumn{2}{c}{$\exists e,e'\in\cE:\bbE[\delta_y^e]\neq0,\bbE[(\delta_y^{e'})^2]\neq0$} \\
latent confounder & w/o & w/ & w/o & w/ \\
\midrule
conditional invariance & \cmark & \xmark & \xmark & \xmark \\
risk invariance  & \cmark & \cmark & \xmark & \xmark\\
solution invariance & \cmark & \xmark & \xmark & \xmark \\
full gradient invariance & \cmark & \cmark & \cmark$^1$ & \cmark$^1$\\
gradient invariance & \cmark & \cmark & \cmark$^2$ & \cmark$^2$\\
\bottomrule
\end{tabular}
\end{table}

\begin{proof}
	
According to model \eqref{eq:scm_train}, we have $Y^e=b^{\star\top} X + \varepsilon_y^e$.

\paragraph{Case I. distribution of $\varepsilon_y^e=0$ the same for all $e\in\mathcal{E}$, without latent confounder.}
In this case, we have $Y^e=b^{\star\top} X + \varepsilon_y$, where $\varepsilon_y$ is independent of $X^e$ and has the same marginal distribution across all environments. Hence the conditional distribution of $Y^e$ given $X^e_{\mathrm{pa}(Y)}=x_{\mathrm{pa}(Y)}$ which is the distribution of $b^{\star\top}_{\mathrm{pa}(Y)}x_{\mathrm{pa}(Y)} + \varepsilon_y$ remains invariant for all $e$, which suggests the conditional invariance holds for the parental set of $Y$. 

The optimal solution given the parental set is $\bbE[Y^e|X^e_{\mathrm{pa}(Y)}]=b^{\star\top}_{\mathrm{pa}(Y)}X^e_{\mathrm{pa}(Y)}$. Hence the parental set and $b^\star$ satisfy the solution invariance. 

The $L_2$ risk of the causal parameter is given by $\bbE[(Y^e-b^{\star\top} X)^2]=\bbE[(Y^e-b^{\star\top}_{\mathrm{pa}(Y)}X^e_{\mathrm{pa}(Y)})^2]=\bbE[\varepsilon_y^2]$ which is the same for all $e$, so we conclude the risk invariance.

The gradient of the $L_2$ risk for each $e$ evaluated at $b^\star$ is $\bbE[X^e(Y^e-b^{\star\top}X^e)]=\bbE[X^e\varepsilon_y^e]=0$. Hence we conclude the full and our gradient invariance.

\paragraph{Case II. $\delta_y^e=0$, with latent confounder.}
The conditional distribution of $Y^e|X^e_{\mathrm{pa}(Y)}=x_{\mathrm{pa}(Y)}$ is the conditional distribution of $\varepsilon_y|X^e_{\mathrm{pa}(Y)}=x_{\mathrm{pa}(Y)}$, shifted by a constant $b^{\star\top}_{\mathrm{pa}(Y)}x_{\mathrm{pa}(Y)}$, which in general varies for different interventions on $X^e_{\mathrm{pa}(Y)}$. The conditional mean $\bbE[Y^e|X^e_{\mathrm{pa}(Y)}]=b^{\star\top}_{\mathrm{pa}(Y)}X^e_{\mathrm{pa}(Y)}+\bbE[\varepsilon_y|X^e_{\mathrm{pa}(Y)}]$, similarly, depends on $e$ as well. So both the conditional and solution invariance in general fail to hold for the causal parameter.

We have $\bbE[(Y^e-b^{\star\top}_{\mathrm{pa}(Y)}X^e_{\mathrm{pa}(Y)})^2]=\bbE[\varepsilon_y^2]$, suggesting the risk invariance. To see the gradient invariance, recalling the model \eqref{eq:scm_train}, we have
\begin{equation*}
	X^e = C^\star_x(\varepsilon_x+\Delta_x^e) + C^\star_{xy} \varepsilon_y.
\end{equation*}
Thus, the gradient at the causal parameter is given by $\bbE[X^e(Y^e-b^{\star\top}X^e)]=\bbE[(C^\star_x(\varepsilon_x+\Delta_x^e) + C^\star_{xy} \varepsilon_y)(\varepsilon_y)]=C^\star_x\bbE[\varepsilon_x\varepsilon_y]+C^\star_{xy}\bbE[\varepsilon_y^2]$, which is free of $e$. So we conclude the full gradient invariance which also implies our gradient invariance. 

\paragraph{Case III. $\bbE[\delta_y^e]\not\equiv c,\bbE[{\delta_y^{e}}^2]\not\equiv c$, w/ or w/o latent confounders.}
The conditional distribution of $Y^e|X^e_{\mathrm{pa}(Y)}=x_{\mathrm{pa}(Y)}$ is the conditional distribution of $b^{\star\top}_{\mathrm{pa}(Y)}x_{\mathrm{pa}(Y)}+\varepsilon_y+\delta_y^e$ given $X^e_{\mathrm{pa}(Y)}=x_{\mathrm{pa}(Y)}$ which apparently varies for different $e$ regardless of the existence of the latent confounders. The conditional expectation $\bbE[Y^e|X^e_{\mathrm{pa}(Y)}]=b^{\star\top}_{\mathrm{pa}(Y)}X^e_{\mathrm{pa}(Y)}+\bbE[\varepsilon_y|X^e_{\mathrm{pa}(Y)}]+\bbE[\delta_y^e]$ depends on $e$. The risk is now given by $\bbE[\varepsilon_y^2]+\bbE[{\delta_y^e}^2]$ which also depends on $e$. In contrast, the gradient becomes $$C^\star_x(\bbE[\varepsilon_x\varepsilon_y]+\bbE[\Delta_x^e\delta_y^e])+C^\star_{xy}(\bbE[\varepsilon_y^2]+\bbE[{\delta_y^e}^2]).$$
Under the conditions that $\bbE[\Delta_x^e\delta_y^e]\equiv c$ and $C_{xy}^\star=0$, we have the full gradient invariance. Our gradient invariance, in this case, is equivalent to say
\begin{equation*}
	C^\star_x\sum_{e\in\cE}\omega^e\bbE[\Delta_x^e\delta_y^e] + C_{xy}^\star\sum_{e\in\cE}\omega^e\bbE[{\delta_y^e}^2] = C^\star_x\bbE[\Delta_x^0\delta_y^0] + C_{xy}^\star\bbE[{\delta_y^0}^2].
\end{equation*}
So it is adequate to assume $\sum_{e\in\cE}\omega^e\bbE[\Delta_x^e\delta_y^e]=\bbE[\Delta_x^0\delta_y^0]$ and $\sum_{e\in\cE}\omega^e\bbE[{\delta_y^e}^2]=\bbE[{\delta_y^0}^2]$ or $C_{xy}=0$, for the causal parameter to satisfy our gradient invariance. 
\end{proof}

\section{Numerical exploration of the assumptions of Theorem~\ref{thm:main1}}
\label{sec:numerical_ass}
We consider the setup with three environments, two environments with small interventions and an environment with large interventions. Denote $e =1,2$ to be the two environments with small interventions and $e = 3$ to be the environment with large interventions. We set $p = 9$ and generate three Gram matrices $G^e \in \mathbb{R}^{p+1{\times}p+1}$, corresponding to data from each environment as follows:
\begin{eqnarray*}
\begin{aligned}
G^1&= (p+1)\times(p+1)~\text{ matrix with iid normal entries}; G_1 \leftarrow G^1{G^1}^\top\\
\zeta_1&= (p+1)\times(p+1)~\text{ matrix with iid normal entries}; \zeta_1 \leftarrow \zeta_1\zeta_1^\top/20\\
\zeta_2&= (p+1)\times(p+1)~\text{ matrix with iid normal entries}; \zeta_2 \leftarrow \zeta_2\zeta_2^\top/20\\
G^2 &= G^1+\zeta_1-\zeta_2\\
\zeta_3&= (p+1)\times(p+1)~\text{ matrix with iid normal entries}; \zeta_3 \leftarrow \zeta_3\zeta_3^\top\\
G^3&=\zeta_3+ G^1+G^2
\end{aligned}
\end{eqnarray*}
Note that by construction, $G^e \succ 0$ with high probability. Further, for every such matrix, there exists a SCM \eqref{eq:scm_train} such that the Gram matrix of $(X^e,Y^e)$ is $G^e$. Moreoever, $\mathbb{E}[(Y^e-X^eb)^2] = (b,1)G^e(b,1)^\top$. Furthermore,  Assumption~\ref{ass:strength_2} can be stated completely in terms of Gram matrices. 

Let $\mathcal{E} = \{1,2,3\}$, $\omega^e = 1/3$ for each $e\in\mathcal{E}$, and $\gamma = 4$. We generate 10000 instances of $G^e$ according to the scheme described above. All the instances do not satisfy the `observational' assumption (i.e. $\not\exists e' \text{ such that }G^{e'} \preceq G^{e} \text{ for all }e \in\mathcal{E}$). Furthermore, all instances satisfy Assumption~\ref{ass:strength_2} with $\mathcal{E}_\text{small} = \{1,2\}$. Out of the 10000 instances, 3480 satisfy the assumptions of Theorem~\ref{thm:main}. 

This numerical illustration shows that there are many instances where the observational assumption is not satisfied, and Assumption~\ref{ass:strength} and the assumptions of Theorem~\ref{thm:main1} are satisfied, highlighting that these assumptions are much less restrictive than the `observational' assumption.

\section{Approximate robustness guarantees of DRIG}
\label{sec:approximate_robustness}
Consider the sets:
 \begin{eqnarray*}
\begin{aligned}
\mathcal{C}_{1,\gamma}:= & \left\{v \in \mathbb{R}^{p+1}: \bbE[vv^\top ] \preceq \left[K^\star_1 + \gamma\sum_{e \in \mathcal{E}} \omega^e\left(S^e - K^\star_1\right)\right]_+\right\},\\
\mathcal{C}_{2,\gamma}:= & \left\{v \in \mathbb{R}^{p+1}: \bbE[vv^\top ] \preceq K^\star_2 + \gamma\sum_{e \in \mathcal{E}} \omega^e\left(S^e - K^\star_2\right)\right\},
\end{aligned}
\end{eqnarray*}
where, 
\begin{eqnarray*}
\begin{aligned}
K^\star_1&= \argmin_{K \in \mathbb{R}^{p+1 \times p+1}}\|K\|_2\quad\text{subject-to}\quad K = S^e\text{ for some }e \in \mathcal{E} \\
K^\star_2  &= \argmax_{K\in \mathbb{R}^{p+1 \times p+1}} \|K\|_2 \quad\text{subject-to}\quad K \preceq S^e \text{ for all }e \in \mathcal{E}.
\end{aligned}
\end{eqnarray*}
Here, for a symmetric matrix $A$ with eigenvector/eigenvalue pairs $(u_i,\lambda_i)$, $[A]_{+} = \sum_{i}\max\{\lambda_i,0\}u_iu_i^\top$ represents the positive part of the matrix. Furthermore, $\|A\|_2$ represents the spectral norm of $A$. Since $K^\star_2 \preceq K^\star_1$, we have for every $\gamma \geq 1$, $\mathcal{C}_{1,\gamma} \subseteq \mathcal{C}_{2,\gamma}$. Thus, for every regression parameter $b \in \mathbb{R}^{p}$ and $\gamma \geq 1$, we have: $\mathcal{L}_{\mathcal{C}_{1,\gamma}}^\mathrm{robust}(b) \leq \mathcal{L}_{\mathcal{C}_{2,\gamma}}^\mathrm{robust}(b)$. The following theorem assesses how the DRIG loss $\mathcal{L}_\gamma(b)$ is related to objectives $\mathcal{L}_{\mathcal{C}_{1,\gamma}}^\mathrm{robust}(b)$ and $\mathcal{L}_{\mathcal{C}_{2,\gamma}}^\mathrm{robust}(b)$, and characterizes the robustness properties of the DRIG prediction model to perturbations in the test environment. For simplicity, we omit constants and specify them in Appendix~\ref{proof_thm_main}.
\begin{theorem}
For every $\gamma \geq 1$ and regression parameter $b \in \mathbb{R}^{p}$, the DRIG objective \eqref{eq:obj} is between $\mathcal{L}_{\mathcal{C}_{1,\gamma}}^\mathrm{robust}(b)$ and $\mathcal{L}_{\mathcal{C}_{2,\gamma}}^\mathrm{robust}(b)$, i.e.: $\mathcal{L}^\mathrm{robust}_{\mathcal{C}_{1,\gamma}}(b)\leq \mathcal{L}_\gamma(b) \leq \mathcal{L}^\mathrm{robust}_{\mathcal{C}_{2,\gamma}}(b)$. Furthermore, suppose $K^\star_1 + \gamma\sum_{e \in \mathcal{E}} \omega^e\left(S^e - K^\star_1\right) \succeq 0$ and $\frac{(1-\gamma)\|K^\star_2-K^\star_1\|_2}{\sigma_\text{min}(I-B^\star)} < 1$. Then, the distance between the solution $b^\mathrm{opt}_\gamma$ of \eqref{eqn:DRIG} and the minimizer of \eqref{eqn:worst-case} with respect to the set $\mathcal{C}_{1,\gamma}$ and $\mathcal{C}_{2,\gamma}$ is bounded:
\begin{eqnarray*}
\max_{\mathcal{C}\in\{\mathcal{C}_{1,\gamma},\mathcal{C}_{\gamma,2}\}}\|b^\mathrm{opt}_\gamma - \argmin_{b\in\mathbb{R}^p} \mathcal{L}_{\mathcal{C}}^\mathrm{robust}(b)\|_2 \leq c'\sqrt{\gamma\|K^\star_1 - K_2^\star\|_2},
\end{eqnarray*}
with $\max_{\mathcal{C}\in\{\mathcal{C}_{1,\gamma},\mathcal{C}_{\gamma,2}\}}\mathcal{L}_{\mathcal{C}}^\mathrm{robust}(b^\mathrm{opt}_\gamma) - \min_{b\in\mathbb{R}^p} \mathcal{L}_{\mathcal{C}}^\mathrm{robust}(b) \leq c\gamma\|K^\star_1 - K_2^\star\|_2$ for some constants $c,c'$. 
\label{thm:main2}
\end{theorem}

We prove Theorem~\ref{thm:main} in Supplementary~\ref{proof_thm_main}. The first part of the theorem states that the DRIG loss is sandwiched between two distributional robust objectives, one with respect to the set $\mathcal{C}_1$ and the other with respect to the set $\mathcal{C}_2$. A key quantity in the second part of our result is $\|K^\star_1-K^\star_2\|_2$: the smaller this quantity, the closer the DRIG estimate $b^\mathrm{opt}_\gamma$ is to minimize the worst-case risk \eqref{eqn:worst-case} with respect to the set $\mathcal{C}_{1,\gamma}$. As a setting where $\|K^\star_1-K^\star_2\|_2$ is small, suppose there exists a collection of environment $\mathcal{E}_\mathrm{small} \subset \mathcal{E}$ with small interventions, i.e. $S^e \preceq S^f$ for all $e \in \mathcal{E}_\mathrm{small}$ and $f \in \mathcal{E}\setminus\mathcal{E}_\mathrm{small}$, and $\|S^e-S^{e'}\|_2 \leq \epsilon$ for all $e,e' \in \mathcal{E}_\mathrm{small}$ and some small $\epsilon$. Then, it is straightforward to show that $\|K^\star_1-K^\star_2\|_2 \leq \epsilon$.

\section{Causal identification via DRIG}\label{sec:causal-identify}
We investigate causal identifiability with the DRIG estimator \eqref{eqn:DRIG} when $\gamma \to \infty$. 

In Section \ref{sec:case1}, we show that if there are sufficient interventions on the covariates $X$, then $\mathrm{rank}([C^\star{L^\star}{C^\star}^\top]_{1:p,1:p}) = p$ and the set of models $\mathcal{I}$ with invariant gradients is a singleton. In this setting, according to \eqref{eq:bias}, the optimal solution of DRIG when $\gamma \to \infty$ is a biased version of the causal parameter $b^\star$, where the bias is given by $([C^\star{L^\star}{C^\star}^\top]_{1:p,1:p})^{-1}(C^\star_xL^\star_{xy}+L^\star_{y}C_{xy}^\star)$. We analyze in Section \ref{sec:case1} the magnitude of this bias under various structural assumptions. In Section~\ref{sec:insuff_interv_x}, we consider the setting where there are insufficient interventions on the covariates $X$ but impose structural assumptions so that $C^\star_xL^\star_{xy}+L^\star_{y}C_{xy}^\star = 0$; here, the set of models $\mathcal{I}$ with invariant gradients typically consists of multiple elements, and we identify the most predictive model according to \eqref{eq:b_insufficient}.

Throughout, we assume additive interventions, i.e. assume the following model for $\varepsilon^e$:
$$\varepsilon^e = \varepsilon + \delta^e,$$
where $\varepsilon$ is independent of $\delta^e$, and $\delta^e$ represents additive interventions. Note that for a variable $j$, $\delta_j^e$ not being identically zero implies that either variable $j$ has received a direct intervention, or there has been an intervention on the latent variable. 

\subsection{Sufficient interventions on the covariates $X$}
\label{sec:case1}
 Recalling that the matrix $L^\star_x$ encodes interventions on the covariates we impose conditions on $L^\star_x$. In particular, in Section~\ref{sec:no_latent_interv}, we assume no interventions on the response or latent variables, leading to a identifiable case for the causal parameter; in Section~\ref{sec:interventions_latent}, we allow for interventions on the latent variable and the response variable and study the approximate causal identifiability by quantifying the bias with respect to the causal parameter.

\subsubsection{No interventions on the response variable $Y$ or latent variables $H$}
\label{sec:no_latent_interv}

By making structural assumptions on the underlying graphical model, the result of Theorem~\ref{thm:solution_infty} can be specialized to attain full causal identifiability, namely the DRIG estimator recovering the causal parameter.  
\begin{corollary}[causally identifiable and robust]\label{cor:identifiable}Suppose that $\delta^e_{p+1} \equiv 0$ for every $e \in \mathcal{E}$ and $L_x^\star \succ0$. Then, we have that 
$$ b_\infty^{\mathrm{opt}}=b^\star \quad \text{ and }\quad \lim_{\gamma \to \infty}\mathcal{L}_{\gamma}(b_\infty^{\mathrm{opt}}) = E[(\varepsilon_y^e)^2],$$
where $\varepsilon_y^e := \varepsilon^e_{p+1}$ represents the component of the noise $\varepsilon^e$ corresponding to $Y$.
\end{corollary}
See Supplementary~\ref{proof_cor_identifiable} for the proof. Corollary~\ref{cor:identifiable} states that under some assumptions, the causal parameter $b^\star$ can be identified by the DRIG estimator with $\gamma \to \infty$. The assumption $\delta^e_{p+1} \equiv 0$ for every $e \in \mathcal{E}$ requires that there are no interventions on the response $Y$ or any latent variables $H$, that is $E$ does not point to $H$ or $Y$ in the graphical model~\ref{fig:graph}(b). The assumption $L^\star_{x} \succ 0$ ensures that there are interventions on all the covariates $X$, that is $E$ points to every covariate in $X$.  Under these conditions, the invertibility assumption in Theorem~\ref{thm:solution_infty} is satisfied, and the matrices $L^\star_{xy}$ and $L^\star_y$ are both equal to zero.  We note that a similar result as Corollary~\ref{cor:identifiable} was also established in \cite{Rothenhausler2017CausalDF} without touching upon the objective that quantifies the robustness, although \cite{Rothenhausler2017CausalDF} only considers the specialized settings discussed above, and does not provide guarantees on approximate identifiability under more general settings (as we do in subsequent sections). 

The assumption that the interventions do not directly affect the response variable or the latent variables is common for identifiability in the causal inference literature. Similarly, the assumption that the covariates all receive an intervention is also prevalent, although the manifestation of this assumption is different in our setting than in instrumental variable regression or in anchor regression. To take a closer look at the latter condition, namely $L^\star_x \succ 0$, note that $L^\star_x = \sum_{e \in \mathcal{E}}\omega^e(S^e-S^0)_{1:p,1:p}$ where as defined in Section~\ref{sec:robust_gua}, $S^e := \mathbb{E}[\varepsilon^e{\varepsilon^e}^\top]$. Thus the condition that $L^\star_x$ is positive definite can be satisfied with data from {two environments} (a reference environment and an additional environment). In particular, as long as $(S^{e})_{1:p,1:p}  \succ S^{0}_{1:p,1:p}$ for the non-reference environment $e$, we have that $L^\star_x \succ 0$, and can guarantee identifiability. In contrast, 
instrumental variable regression or anchor regression on data from SCM \eqref{eq:scm_train} can only guarantee identifiability if $\sum_{e \in \mathcal{E}}\omega^e({\mu^e}{\mu^e}^\top)_{1:p,1:p} \succ 0$. In other words, these methods require {at least $p$ environments} to recover the causal parameter, which is generally far larger than the number of environments required by DRIG. Conceptually, the improvement in identifiability offered by DRIG comes from the fact that it exploits both mean and variance shifts, whereas the other two methods only exploit mean shifts. A similar attribute of DRIG led to substantial improvement in using DRIG for obtaining robust predictions over other methods (see Section~\ref{sec:robust_gua}).

Besides identifying the causal parameter, the optimal objective function, which is the worst-case risk according to Theorem~\ref{thm:main}, is finite and depends on the variance of the exogenous noise associated with $Y$. Recall that the causal parameter is robust against arbitrary interventions on $X$, namely the perturbation class $\mathcal{C}_{\mathrm{causal}}$. Thus, the prediction model $b^\star$ is guaranteed to have a bounded mean squared error under arbitrarily strong interventions on $X$, which is appealing in some applications. 

\paragraph{Independent interventions on the response variable} Previously, we assumed that there are no interventions on $Y$, so that $L^\star_y>0$. We next relax this condition, and allow independent interventions on $Y$. Formally, we assume that $\mathbb{E}[\delta_x^e\delta_y^e] = 0$ for every $e \in \mathcal{E}$; this assumption will be satisfied if there are no interventions on the latent variables $H$, and if the interventions on $X$ and $Y$ are independent. As with Corollary~\ref{cor:identifiable}, we assume that there are interventions on all the covariates $X$ (i.e., $L_x \succ 0)$. Under these assumptions, we have $L^\star_{xy} = 0$, and the result of Theorem~\ref{thm:solution_infty} can be specialized to attain (approximate) causal identifiability even when $Y$ is intervened on.
\begin{corollary}[independent interventions on $Y$] \label{corr:independent_Y}
Suppose that $\mathbb{E}[\delta_x^e\delta_y^e] = 0$ for every $e \in \mathcal{E}$, and that $L^\star_x \succ 0$ and $L^\star_y>0$. Then,
\begin{equation}\label{eq:bias_bound1}
    	\left\|b_\infty^{\mathrm{opt}} - b^\star\right\|_\infty \leq \frac{\|C^\star_{xy}\|_\infty}{\min_{\|u\|_\infty=1}\|(C^\star_xL^\star_x {C^\star_x}^\top/L^\star_{y} + C^\star_{xy}{C^\star_{xy}}^\top)u\|_\infty}.
	\end{equation}
Further, assuming that $Y$ is not an ancestor of any covariate $X$, then we have
\begin{equation*}
	b_\infty^{\mathrm{opt}} = b^\star \quad \text{ and } \quad \lim_{\gamma \to \infty}\mathcal{L}_\gamma(b_\infty^{\mathrm{opt}}) \to \infty.
\end{equation*}
\end{corollary}
See Supplementary~\ref{proof_corr_independent_Y} for the proof. 
Corollary~\ref{corr:independent_Y} states that under the setting where the interventions on $Y$ are independent of those on $X$ and when all covariates are intervened on, the DRIG estimator with $\gamma \to \infty$ approximates the causal parameter at the resolution in \eqref{eq:bias_bound1}. Notice that the approximation becomes tighter the smaller $L^\star_y$ or equivalently the weaker the interventions on $Y$. Corollary~\ref{corr:independent_Y} further states that if the response $Y$ is a descendant of all the covariates, then we have full identifiability, regardless of the intervention strength on the response variable. However, in contrast to Corollary~\ref{cor:identifiable}, now the objective function evaluated at the optimum is approaching infinity as $\gamma\to\infty$. In other words, even though DRIG can identify the causal parameter when there are interventions on $Y$, it does not protect against arbitrarily strong interventions on both $X$ and $Y$. Specifically, all linear prediction models, which includes the causal parameter, would attain an infinite worst-case error.

Nevertheless, the following proposition shows that the causal parameter is robust against another perturbation class which consists of arbitrarily strong interventions on $X$ but bounded interventions on $Y$. This is a slight generalization of the robustness result of the causal parameter discussed in Section~\ref{sec:robust_gua}. See Supplementary~\ref{proof_prop_int_y_robust} for the proof.
 \begin{proposition}\label{prop:int_y_robust}
Suppose that the test data is generated according to the SCM \eqref{eq:scm_test}. Under the assumptions in Corollary~\ref{corr:independent_Y}, for any $c\geq0$, we have
	\begin{equation*}
		b^\star=\argmin_b\sup_{v\in\bbR^{p+1}:\bbE[v_y^2]\leq c}\bbE[(Y-b^\top X)^2],
	\end{equation*}
	where $v_y$ is the component of $v$ corresponding to $Y$. 
\end{proposition}

\subsubsection{Interventions on the latent variables with dense latent effects}
\label{sec:interventions_latent}
When there are interventions on the latent variables or on the response variable that is the parent of some covariates, the assumptions in Section~\ref{sec:case1} are not satisfied, and thus identifiability cannot be guaranteed. Nonetheless, we will demonstrate in this section that under some assumptions on the strength of perturbations on the covariates, and structural assumptions on the latent variables, we can guarantee that the DRIG estimator with $\gamma \to \infty$ can approximately identify the causal parameter $b^\star$. To formally state assumptions needed for approximate identifiability, we model the effects of those latent variables that vary explicitly:
$$\begin{pmatrix}X^e \\ Y^e\end{pmatrix} = B^\star\begin{pmatrix}{X}^e \\Y^e \end{pmatrix} + \Gamma^\star{H}^e+\varepsilon+\delta^e \quad;\quad H^e = H + \eta^e,$$
where $H \in \mathbb{R}^h$ represents the unperturbed latent variables and $\eta^e$ represents interventions on these latent variables. The matrix $\Gamma^\star \in \mathbb{R}^{p \times h}$ encodes the effect of the latent variables on the observed variables. As the latent effects and their perturbations are fully captured by the term $\Gamma^\star{H}^e$, the quantity $\delta^e$ represents the perturbations on only the observed variables, and is independent of $H^e$. Finally, $\varepsilon$ is an independent noise term that is independent of both $\delta^e$ and $H^e$. For simplicity, we assume that $e = 0$ is an observational setting with $\delta^0 \equiv 0$ and $\eta^0 \equiv 0$. \\

Before describing the assumptions needed for our theoretical guarantees, we present some notations. Specifically, we denote $\sigma_{\mathrm{max}}(\cdot)$ and $\sigma_{\mathrm{min}}(\cdot)$ as the maximum and minimum singular value of an input matrix. 

\begin{assumption} Our analysis is based on the setting where the number of covariates $p$ is tending to infinity, and makes the following assumptions:
\begin{enumerate}[label=\textit{A\arabic*}]
\setlength{\itemsep}{2pt}
\setlength{\parskip}{2pt}
        \item The sub-graph among the observed variables is a DAG.\label{assumptiondag} 
        \item The latent variables $H$ are ancestors of the observed variables. \label{assumption:latent_exogenous}
        \item The number of latent variables {$h$} is much smaller than the number of observed variables: $h = o(p)$.\label{assumption:dense_latent_number}
	\item The latent effects are dense, that is:
$\max_{i \in [p]}\|\Proj_{\text{col-space}(\Gamma^\star)}e_i\|_2^2 = \mathcal{O}(h/p)$. \label{assumption:dense_latent}
	\item The latent effects are bounded, i.e., $\|\Gamma^\star\|_2^2 = \mathcal{O}(h)$.\label{assumption: bounded_latent}
	\item The interventions on the covariates $X$ are sufficiently strong: $\sigma_{\mathrm{min}}(L_x^\star) > \frac{4\|L_{xy}^\star\|_2 \sigma_\mathrm{max}(I-B^\star)^2}{\sigma_{\mathrm{min}}(I-B^\star)}$. \label{assumption: sufficient_strong}
	\item The causal coefficients are not too large, i.e., $d\max_{i,j}|B^\star_{ij}| < 1/2$, where $d$ is the largest number of incoming and outgoing edges among the nodes in the subgraph among observed variables. \label{assumption:bound_causal_coeffs}
\end{enumerate}
\end{assumption}
Assumption~\ref{assumptiondag} requires that there are no cycles in the graph among the observed variables. Assumption~\ref{assumption:latent_exogenous} assumes that the latent variables $H$ act exogenously on the observed variables. Assumption~\ref{assumption:dense_latent_number} requires that the number of latent variables is much smaller than number of observed variables. Assumption~\ref{assumption:dense_latent} can be interpreted as the effects of the latent variables spread across all the observed variables. The quantity $\max_{i \in [p]}\|\Proj_{\text{col-space}(\Gamma^\star)}e_i\|_2$ in this condition is an incoherence parameter \citep{Chandrasekaran2011RankSparsityIF} measuring the ``diffuseness" of the latent effects, where $\Proj_{\text{col-space}(\Gamma^\star)}$ is the projection onto the column-space of $\Gamma^\star$ and $e_i$ is a standard coordinate basis. The smaller the value of $\max_{i \in [p]}\|\Proj_{\text{col-space}(\Gamma^\star)}e_i\|_2$, the less concentrated the effect of the latent variables on any single observed variable. As $\max_{i \in [p]}\|\Proj_{T}(e_i)\|_2 \in [\sqrt{\text{dim}(T)/{p}},1]$ for any subspace $T \subseteq \mathbb{R}^p$, Assumption~\ref{assumption:dense_latent} ensures that the latent effects are sufficiently diffuse. Assumption~\ref{assumption: bounded_latent} requires that the latent effects are bounded; for example entries of $\Gamma^\star$ being distributed as $\mathcal{N}(0,1/p)$ satisfies this condition. Assumption~\ref{assumption: sufficient_strong} requires sufficiently strong interventions on the covariates $X$. Finally, Assumption~\ref{assumption:bound_causal_coeffs} ensures that the strength of the causal effects among observed variables is not too large. 

\begin{proposition}(approximate identifiability with interventions on the latent variables) Suppose that Assumptions~\ref{assumption:latent_exogenous}-\ref{assumption:bound_causal_coeffs} are satisfied. As the number of covariates $p$ tends to infinity, we have:
\begin{equation*}
\left\|b_\infty^{\mathrm{opt}} - b^\star\right\|_\infty = \mathcal{O}\left(\frac{h^{5/2}\max_e\|\mathrm{Cov}({\eta^e})\|_\infty + \max_e \mathbb{E}[(\delta^e_y)^2]}{\sigma_{\mathrm{min}}(L_x^\star)}\right).
\end{equation*}
\label{corollary:dense}
\end{proposition}
We prove Proposition~\ref{corollary:dense} in Supplementary~\ref{proof_corollary_dense}. 
This result states that while identifiability may not be possible in the setting where there are interventions on the latent variables and on the response variable $Y$, the DRIG estimator with $\gamma \to \infty$ can approximate the causal parameter $b^\star$ up to some resolution. Specifically, note that $\mathrm{Cov}({\eta^e})$ is the covariance matrix of the latent perturbations $\eta^e$, $\mathbb{E}[(\delta^e_y)^2]$ encodes the variance of perturbations on the response variable $Y$, and $L^\star_{x}$ encodes perturbation strengths on the covariates $X$. Thus, Proposition~\ref{corollary:dense} claims that the stronger the perturbations on the covariates $X$ (i.e., larger $\sigma_{\mathrm{min}}(L_x^\star)$) relative to perturbations on the latent variables and on the response variable, the better the DRIG estimate approximates the causal parameter $b^\star$.

\subsection{Insufficient interventions on $X$}\label{sec:insuff_interv_x}
So far, we have assumed that there are interventions on all the covariates $X$, so that the set of models $\mathcal{I}$ in Theorem~\ref{thm:solution_infty} that satisfy the invariant gradient condition is a singleton. We next relax this condition, resulting in multiple models that exhibits invariant gradients. 

For simplicity, throughout the following discussion, we assume that there are no interventions on the response variable $Y$ or on the latent variables $H$ so that $L^\star_{xy} = 0$ and $L^\star_y = 0$, and only focus on insufficient interventions on $X$.  We denote $\varepsilon_x$ and $\varepsilon_y$ as the components of $\varepsilon$ corresponding to the covariates and the response variable, respectively. 



\begin{proposition}\label{prop:bias_bound3}
	Suppose $L^\star_{xy} = 0$ and $L^\star_y = 0$. Then, $\cI=\{b^\star+b':\Delta_x b'=0\}$ where $\Delta_x := \sum_{e \in \mathcal{E}}\omega^e(\mathbb{E}[X^e{X^e}^\top]-\mathbb{E}[X^0{X^0}^\top])$. Furthermore, we have 
	\begin{equation}\label{eq:bias3}
        b_\infty^{\mathrm{opt}}=b^\star+D\bbE[X^0\varepsilon_y],
	\end{equation}
	where $D:=\lim_{\gamma\to\infty}\big(\bbE [X^0{X^0}^\top]+\gamma\Delta_x\big)^{-1}$. Finally, 
    \begin{equation}\label{eq:bias_bound3}
		\left\|b_\infty^{\mathrm{opt}}-b^\star\right\|_\infty \leq \|D\|_\infty({\|C^\star_x\bbE[\varepsilon_x\varepsilon_y]\|_\infty+\|C^\star_{xy}\bbE[\varepsilon_y^2]\|_\infty}).
	\end{equation}
\end{proposition}

We prove Proposition~\ref{prop:bias_bound3} in Supplementary~\ref{proof:bias_bound3}. It first states that when there are not sufficient interventions on $X$ so that $\Delta_x$ is not positive definite, the set $\cI$ is not a singleton but an equivalence class. Then by \eqref{eq:b_insufficient}, DRIG with $\gamma\to\infty$ is searching for the best predictive solution among this equivalence class. Next, formula \eqref{eq:bias3} and bound \eqref{eq:bias_bound3} quantify the closeness of the causal parameter to the DRIG estimator $b^\mathrm{opt}_\gamma$ when $\gamma \to \infty$. The bias in estimating the causal parameter stems from two sources. First, under insufficient interventions on the covariates, the matrix $\Delta_x$ is not positive definite so $D\neq0$. Second, when there are latent confounders or when some covariates are descendants of $Y$, we have $\bbE[\varepsilon_x\varepsilon_y]\neq0$ or $C^\star_{xy}\neq0$, respectively. Nevertheless, we will show next that under some structural assumptions, DRIG can achieve partial identifiability, and produces a smaller bias than both pooled and observational OLS estimators. 

For simplicity, we consider a specialized setting where the covariates are jointly independent and so are the interventions on them, that is, $\bbE[X^0X^{0\top}]$ and $\Delta_x$ are both diagonal matrices. Then, it is straightforward to show that the bias $\|b^\mathrm{opt}_\gamma-b^\star\|_\infty$ is monotonically decreasing with respect to $\gamma\ge0$. Further assume there exists $i\in\{1,\dots,p\}$ such that the $i$th diagonal entry of $\Delta_x$ and the $i$th component of $\bbE[X^0\varepsilon_y]$ are nonzero, i.e., when the intervention happens to a covariate that is confounded with $Y$. Then, the bias $\|b^\mathrm{opt}_\gamma-b^\star\|_\infty$ is strictly decreasing with respect to $\gamma\ge0$, which implies that DRIG with $\gamma>1$ always has a smaller bias than observational and pooled OLS. Moreover, for any coordinate $i$ such that the $i$th diagonal entry of $\Delta_x$ is nonzero or the $i$th component of $\bbE[X^0\varepsilon_y]$ is zero, we have $\lim_{\gamma \to \infty}b_{\gamma,i}^{\mathrm{opt}}=b^\star_i$. In other words, DRIG with $\gamma\to\infty$ identifies the causal parameter associated with the $i$-th covariate (i) if there is no latent confounder between $X_i$ and $Y$, or (ii) if there is an intervention on this covariate. Thus, even under insufficient interventions on $X$, DRIG can still leverage the limited amount of interventions to partially eliminate the bias caused by the latent confounding effects and partially identify the causal effects.

{\section{Nonlinear DRIG}}
\label{sec:nonlinear}
Let $\tilde{X} = (X,Y)\in \mathbb{R}^{p}$. Consider a nonlinear SCM:
$$\tilde{X}_i = f^\star_{i}(\tilde{X}_{\mathrm{pa}(i)};\varepsilon^e)\quad i \in \{1,2,\dots,p+1\}, e\in\mathcal{E},$$
where $\mathrm{pa}(i) \subset \{1,2,\dots,p+1\}\setminus{i}$ denotes the parental set of node $i$ in graph among the observed variables. Then, the nonlinear population DRIG minimizes:
\begin{equation}
\begin{aligned}
    f_\gamma^\mathrm{nl} \in \argmin_{f\in\mathcal{F}}\min_{e\in\mathcal{E}}\mathbb{E}[\ell(X^{e},Y^{e}; f)] + \gamma\sum_{e \in \mathcal{E}}\omega^e \left(\mathbb{E}[\ell(X^{e},Y^e;f)] - \min_{e\in\mathcal{E}}\mathbb{E}[\ell(X^{e},Y^{e};f)]\right).
\end{aligned}
\end{equation}
where $\mathcal{F}$ is a nonlinear function class; for example, splines or neural networks. Optimization can then be implemented via gradient descent algorithms similar to the case for linear models. 

To investigate the robustness property of the nonlinear formulation, we conduct numerical experiments while theoretical justifications would be worthwhile for future research. We note that distribution shifts that involve changes in the support of the covariates (a.k.a., out-of-support covariate shift) is a fundamentally challenging problem for nonparametric regression that requires specific techniques or structural assumptions~\citep{shen2024engression}. To avoid this complication, we consider settings where the covariates follow a linear structural causal model in \eqref{eq:scm_train} up to a nonlinear function. This allows nonlinear causal relationships between the covariates and the response as well as among covariates. Specifically, let $Z$ be some latent features that follows the SCM
\begin{equation*}
    \begin{pmatrix}
    Z^e \\ Y^e
    \end{pmatrix}
    =
    B^\star
    \begin{pmatrix}
    Z^e \\ Y^e
    \end{pmatrix} 
    + \varepsilon^e,
\end{equation*}
for each environment $e\in\cE$. The observed covariates $X$ are nonlinear, invertible transformations of latent features $Z$, i.e., $X=g(Z)$. Hence we can equivalently write
\begin{equation*}
    \begin{pmatrix}
    g^{-1}(X^e) \\ Y^e
    \end{pmatrix}
    =
    B^\star
    \begin{pmatrix}
    g^{-1}(X^e) \\ Y^e
    \end{pmatrix} 
    + \varepsilon^e,
\end{equation*}
where the causal relationships between $X$ and $Y$ as well as among $X$ are in general nonlinear. 

In our numerical setting, we consider two cases of transformation: cube root $g(z)=z^{1/3}$ and softplus $g(z)=\log(1+\exp(Z))$. For simplicity we consider univariate $Z$ and $X$. We implement DRIG and other methods with a polynomial class of degree 3, which leads to correct specification in the cube root case but slight misspecification in the softplus case. The SCMs and intervention schemes for $(Z,Y)$ is the same as in Example~\ref{ex:robust2}, while the predictor we use for the model is $X=g(Z)$. The regularization coefficient for DRIG and anchor regression are fixed to $\gamma=5$. 

Figure~\ref{fig:nonlinear_results} shows the test MSEs for varying perturbation strength (See Example~\ref{ex:robust2}). We see that in both settings, nonlinear DRIG performs the best among all. While in the softplus case, DRIG with linear models (DRIG-lin) performs reasonably well as linear function is a good approximation of the softplus function, DRIG-lin suffer much more in the cube root case due to the lack of nonlinearity.  

Note that in Figure 8, all the methods except DRIG-lin are nonlinear. 

\begin{figure}
\centering
\begin{tabular}{@{}c@{}c@{}}
	\includegraphics[width=0.48\textwidth]{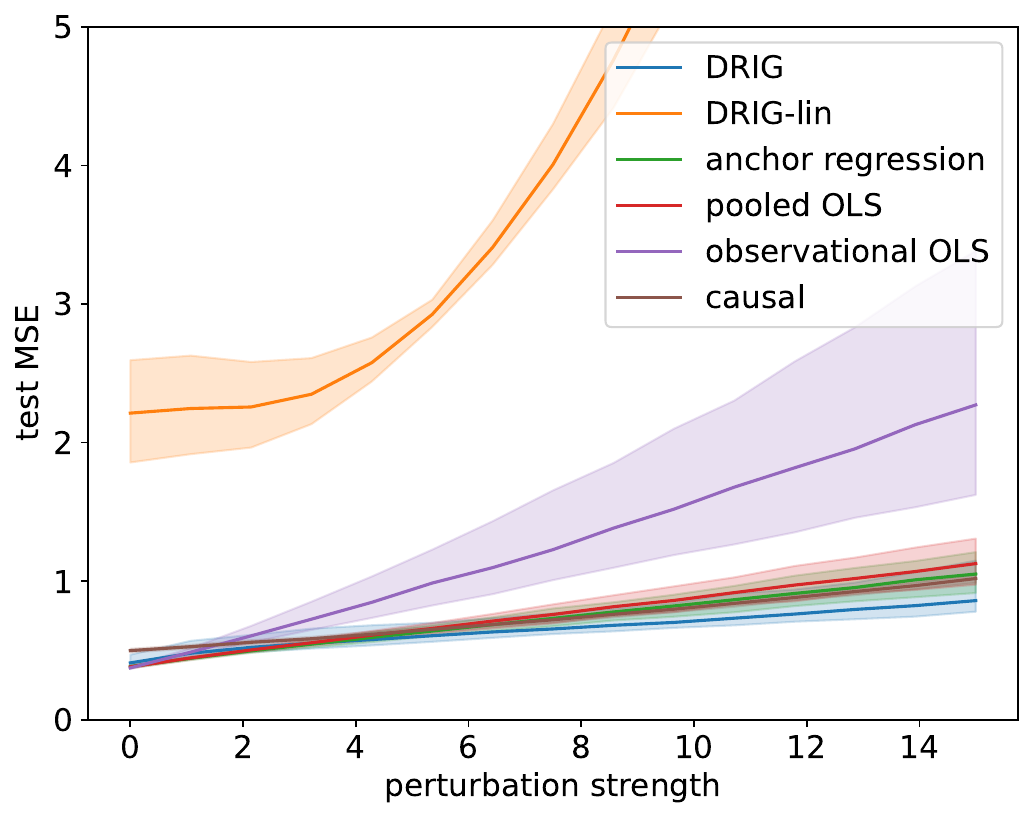} & \includegraphics[width=0.49\textwidth]{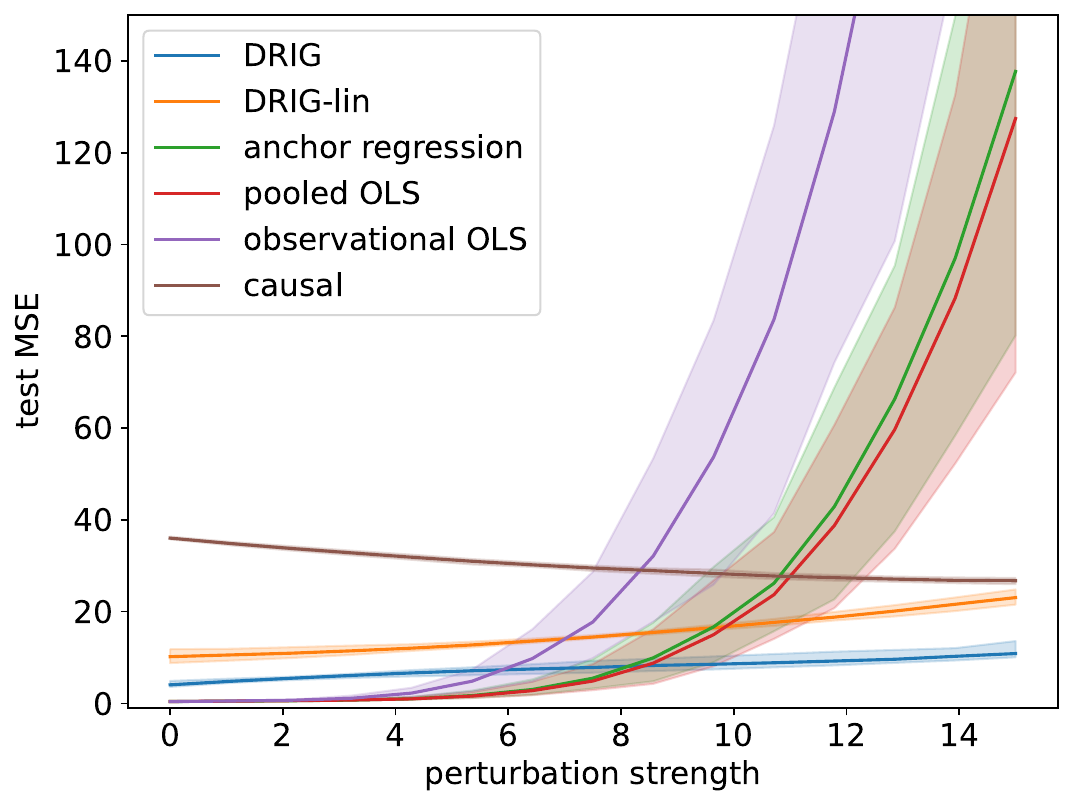}
\end{tabular}
\caption{Results for cube root (left) and softplus (right). {Lines represent the mean and 2.5\% and 97.5\% qunatiles.} }
\label{fig:nonlinear_results}
\end{figure}

\section{Selecting $\Gamma$ in DRIG-A+}
\label{eq:selecting_Gamma}
Note that when we take $\Gamma=\mathrm{diag}(\Gamma_x,\gamma_y)$, the DRIG-A+ estimator has the closed form solution $b^\mathrm{opt}_\Gamma = [\bbE X^0{X^0}^\top +\Gamma_x\Delta_x\Gamma_x]^{-1}[\bbE X^0Y^0+\gamma_y\Gamma_x\Delta_{xy}]$, 
where $\Delta_x:=\sum_{e\in\cE}\omega^e[\bbE X^e{X^e}^\top-\bbE X^0{X^0}^\top]$ and $\Delta_{xy}:=\sum_{e\in\cE}\omega^e[\bbE X^eY^e - \bbE X^0Y^0]$. Thus, compared to the population test OLS, DRIG-A+ replaces ${G^v_x}$ with $\bbE X^0{X^0}^\top +\Gamma_x\Delta_x\Gamma_x$ and $G_{xy}^v$ with $\bbE X^0Y^0+\gamma_y\Gamma_x\Delta_{xy}$. As the gram matrix ${G^v_x}$ can be accurately estimated with a large unlabeled samples, we set $\Gamma_x$ so that $\bbE X^0{X^0}^\top +\Gamma_x\Delta_x\Gamma_x = {G^v_x}$, which yields
${\Gamma}^{\star}_x := \Delta_x^{-1/2}\left[\Delta_x^{1/2}\left(G^v_x - \bbE X^0{X^0}^\top\right)\Delta_x^{1/2}\right]^{1/2}\Delta_x^{-1/2}.$
Given $\Gamma_x = {\Gamma}^\star_x$, we then select $\gamma_y$ to minimize the population test MSE of $b^\mathrm{opt}_\Gamma$, which gives
${\gamma}^\star_y := \frac{({(G^v_x)}^{-1/2}\Gamma^\star_x\Delta_{xy})^\top}{\|{(G^v_x)}^{-1/2}\Gamma^\star_x\Delta_{xy}\|^2}{(G^v_x)}^{-1/2}(G_{xy}^v-\bbE X^0Y^0).$
Then based on the finite test samples, we define $\hat\Gamma_x$ and $\hat\gamma_y$ as the plug-in estimators of ${\Gamma}^{\star}_x$ and ${\gamma}^\star_y$, where we replace $G^v_x$ and $G_{xy}^v$ by $\hat{G}^v_x$ and $\hat{G}_{y}^v$, respectively. We derive the above formulas in Supplementary~\ref{app:Gamma}.

\section{Proofs}
\subsection{Connections to causal Dantzig}\label{app:causal_dantzig}
When $|\cE|=2$, as $\gamma\to\infty$ and $0 \in \mathcal{E}$ is an observational environment with $S^0 \preceq S^1$, DRIG formulation \eqref{eq:obj} becomes 
\begin{equation*}
	\min_b \mathbb{E}[\ell(X^{1},Y^1;b)] - \mathbb{E}[\ell(X^{0},Y^{0};b)]
\end{equation*}
Setting the gradient of the above objective function to 0 yields
\begin{equation*}
	\left[\bbE X^1{X^1}^\top - \bbE X^0{X^0}^\top\right]b=\bbE X^1Y^1-\bbE X^0Y^0
\end{equation*}
which is the population version of the causal Dantzig estimator.

\subsection{Proof of Proposition~\ref{prop:anchor_regression_ours}}
 \label{proof:prop_anchor_regression_ours}
\begin{proof}[Proof of Proposition~\ref{prop:anchor_regression_ours}]

Denote by $(X^e,Y^e)$ the random variables follow the conditional distribution of $(X,Y)$ given $A=a^e$. 
Then we have
\begin{equation*}
\begin{split}
	\cL_{\mathrm{anchor},\gamma}(b) &=\bbE[((I-P_A)(Y-b^\top X))^2] + \gamma\bbE[(P_A(Y-b^\top X))^2]\\ 
	&= \sum_{e\in\cE}\omega^e\bbE[(Y^e-b^\top X^e-\bbE(Y^e-b^\top X^e))^2] + \gamma\sum_{e\in\cE}\omega^e[\bbE(Y^e-b^\top X^e)]^2 \\
	&=\sum_{e\in\cE}\omega^e\bbE[(Y^e-b^\top X^e)^2] + (\gamma-1)\sum_{e\in\cE}\omega^e[\bbE(Y^e-b^\top X^e)]^2.
\end{split}
\end{equation*}
Since $S^0 \preceq S^e$ for all $e\in\mathcal{E}$, the DRIG loss function as
\begin{equation*}
	\cL_\gamma(b) = \sum_{e\in\cE}\omega^e\bbE[(Y^e-b^\top X)^2] + (\gamma-1)\sum_{e\in\cE}\omega^e\left(\bbE[(Y^e-b^\top X^e)^2] - \bbE[(Y^0-b^\top X^0)^2]\right).
\end{equation*}
Note that the difference between the two loss functions lies in the second terms.

For any regression coefficient $b$, define the vector $w$ as
\begin{equation}\label{eq:def_w}
    w := [(I-B^\star)_{p+1,:}^{-1}-b^\top (I-B^\star)_{1:p,:}^{-1}]^\top.
\end{equation}
We note from the SCM \eqref{eq:scm_train} that
\begin{equation*}
	Y^e-b^\top X^e = w^\top\varepsilon^e,
\end{equation*}
and
\begin{equation*}
	\bbE(Y^e-b^\top X^e) = w^\top \mu^e,
\end{equation*}
with $\mu^e=Ma^e$ in this case with deterministic perturbations (here, we have used the fact that $\mathbb{E}[\varepsilon^0]=0$ and $\varepsilon^e = \varepsilon^0 + \mu^e$.). 
Then we have
\begin{equation*}
	[\bbE(Y^e-b^\top X^e)]^2 = w^\top\mu^e{\mu^e}^\top w
\end{equation*}
and
\begin{equation*}
	\bbE[(Y^e-b^\top X^e)^2] - \bbE[(Y^0-b^\top X^0)^2] = w^\top\mu^e{\mu^e}^\top w.
\end{equation*}
Thus, the two loss functions are equal.
\end{proof}

\subsection{Proof of convexity of population and finite-sample DRIG}
\label{proof_convex}
We first prove Proposition~\ref{prop:convex} in the setting where Assumption~\ref{ass:strength_2} is satisfied (a strictly weaker assumption than Assumption~\ref{ass:strength}) and discuss assumptions when finite-sample DRIG is convex.
\begin{proof}
We note from the SCM \eqref{eq:scm_train} that
\begin{equation*}
	Y^e-b^\top X^e = w^\top\varepsilon^e.
\end{equation*}
where $w$ is a linear function of $b$ and is defined in \eqref{eq:def_w}. Thus, 
\begin{equation}
	\mathbb{E}[(Y^e-b^\top X^e)^2] = w^\top\mathbb{E}[\varepsilon^e{\varepsilon^e}^\top]{w} = w^\top{S}^ew.
    \label{eq:simp_loss}
\end{equation}
Thus, the DRIG objective can be equivalently written as:
\begin{eqnarray*}
\begin{aligned}
\mathcal{L}_\gamma(b) &= \gamma{w}^\top\left[\sum_{e\in\mathcal{E}}\omega^e{S}^e\right]w + (1-\gamma)\min_{e'\in\mathcal{E}}w^\top{S}^{e'}{w},\\
&= \gamma{w}^\top\left[\sum_{e\in\mathcal{E}}\omega^eS^e\right]w + (1-\gamma)\min_{e'\in\mathcal{E}_\mathrm{small}}w^\top{S}^{e'}{w},\\
&= \max_{e'\in\mathcal{E}_\mathrm{small}}\gamma{w}^\top\left[\sum_{e\in\mathcal{E}}\omega^eS^e\right]w + (1-\gamma)w^\top{S}^{e'}{w},\\
&= \max_{e'\in\mathcal{E}_\mathrm{small}}{w}^\top\left[\left[\gamma\sum_{e\in\mathcal{E}}\omega^eS^e\right]+(1-\gamma)S^{e'}\right]w.
\end{aligned}
\end{eqnarray*}
Here, the second inequality follows from the fact that $S^{e'} \preceq S^e$ for every $e'\in\mathcal{E}_\mathrm{small}$ and $e\in\mathcal{E}\setminus\mathcal{E}_\mathrm{small}$; the third equality follows from $\gamma \geq 1$. By the assumptions of the proposition, $\left[\gamma\sum_{e\in\mathcal{E}}\omega^eS^e\right]+(1-\gamma)S^{e'} \succeq 0$ for every $e'\in\mathcal{E}_\mathrm{small}$. Thus, since $w$ is a linear function of $b$, then, for every $e \in \mathcal{E}_\mathrm{small}$, ${w}^\top\left[\left[\gamma\sum_{e\in\mathcal{E}}\omega^eS^e\right]+(1-\gamma)S^{e'}\right]w$ is a convex function of $b$. Since point-wise maximum of convex functions are convex, $\mathcal{L}_\gamma(b)$ is convex.
\end{proof}
Let $\hat{G}^e = \frac{1}{n_e}\sum_{i = 1}^{n_e}\begin{pmatrix}X^e_i\\Y^e_i\end{pmatrix}\begin{pmatrix}X^e_i\\Y^e_i\end{pmatrix}^\top$ be the gram matrix. We then have the following statement regarding the convexity of the finite-sample DRIG loss in \eqref{eqn:finite_drig_objective}
\begin{proposition} Suppose there exists a set of environments $\mathcal{E}_\mathrm{small}\subset \mathcal{E}$ such that for every $e'\in\mathcal{E}_\mathrm{small}$ and $e\in\mathcal{E}\setminus\mathcal{E}_\mathrm{small}$, we have $\hat{G}^{e'} \preceq \hat{G}^e$. Furthermore, suppose that for every $e'\in\mathcal{E}_\mathrm{small}$, $\hat{G}^{e'} \preceq \sum_{e\in\mathcal{E}}\omega^e\hat{G}^e$. Then, for $\gamma \geq 1$, the finite-smaple DRIG loss in \eqref{eqn:finite_drig_objective} is convex. 
\end{proposition}
\begin{proof}
It is straightforward to see that:
$$\hat{\mathbb{E}}[\ell(X^e,Y^e;b)] = \tilde{w}^\top\hat{G}^e\tilde{w},$$
where $\tilde{w}^\top =(1,-b)$. Thus, the finite-sample DRIG objective can be equivalently written as: 
\begin{eqnarray*}
\begin{aligned}
\hat{\mathcal{L}}_\gamma(b) &= \gamma\tilde{w}^\top\left[\sum_{e\in\mathcal{E}}\omega^e\hat{G}^e\right]\tilde{w} + (1-\gamma)\min_{e'\in\mathcal{E}}\tilde{w}^\top\hat{G}^{e'}\tilde{w},\\
&= \gamma\tilde{w}^\top\left[\sum_{e\in\mathcal{E}}\omega^e\hat{G}^e\right]\tilde{w} + (1-\gamma)\min_{e'\in\mathcal{E}_\mathrm{small}}\tilde{w}^\top\hat{G}^{e'}\tilde{w},\\
&= \max_{e'\in\mathcal{E}_\mathrm{small}}\gamma\tilde{w}^\top\left[\sum_{e\in\mathcal{E}}\omega^e\hat{G}^e\right]\tilde{w} + (1-\gamma)\tilde{w}^\top\hat{G}^{e'}\tilde{w},\\
&= \max_{e'\in\mathcal{E}_\mathrm{small}}\tilde{w}^\top\left[\left[\gamma\sum_{e\in\mathcal{E}}\omega^e\hat{G}^e\right]+(1-\gamma)\hat{G}^{e'}\right]\tilde{w}.
\end{aligned}
\end{eqnarray*}
Here, the second inequality follows from the fact that $\hat{G}^{e'} \preceq \hat{G}^e$ for every $e'\in\mathcal{E}_\mathrm{small}$ and $e\in\mathcal{E}\setminus\mathcal{E}_\mathrm{small}$; the third equality follows from $\gamma \geq 1$. By the assumptions of the proposition, $\left[\gamma\sum_{e\in\mathcal{E}}\omega^e\hat{G}^e\right]+(1-\gamma)\hat{G}^{e'} \succeq 0$ for every $e'\in\mathcal{E}_\mathrm{small}$. Thus, since $w$ is a linear function of $b$, then, for every $e \in \mathcal{E}_\mathrm{small}$, $\tilde{w}^\top\left[\left[\gamma\sum_{e\in\mathcal{E}}\omega^e\hat{G}^e\right]+(1-\gamma)\hat{G}^{e'}\right]w$ is a convex function of $b$. Since point-wise maximum of convex functions are convex, $\hat{\mathcal{L}}_\gamma(b)$ is convex.
\end{proof}

\subsection{Proof of Theorem~\ref{thm:main}}
\label{proof_thm_main}

\begin{proof}[Proof of Theorem~\ref{thm:main}] We prove Theorem~\ref{thm:main1}, and note that  Assumption~\ref{ass:strength} is strictly stronger than \ref{ass:strength_2}, and that $\bar{e} = 0$ when Assumption~\ref{ass:strength} is satisfied to conclude that Theorem~\ref{thm:main1} implies Theorem~\ref{thm:main} under Assumption~\ref{ass:strength}. For any regression coefficient $b$, define the vector $w$ as in \eqref{eq:def_w}. Note that for the SCM \eqref{eq:scm_test} $Y^v-b^\top X^v=w^\top{v}$, where $w$. Then, we have for any set $\mathcal{C} = \left\{v \in \mathbb{R}^{p+1}~|~\mathbb{E}[vv^\top] \preceq M\right\}$,  
\begin{align*}
    \mathcal{L}_{\mathcal{C}}(b)=\sup_{v \in \mathcal{C}} \bbE[(Y^v-b^\top X^v)^2] =\sup_{v \in \mathcal{C}} w^\top \bbE[vv^\top ]w = w^\top{M}w
\end{align*}
Consider the DRIG objective $\mathcal{L}_\gamma(b)$. Using the relation \eqref{eq:simp_loss}, we have that:
\begin{eqnarray*}
\begin{aligned}
\mathcal{L}_\gamma(b) &= \min_{e\in\mathcal{E}} w^\top{S}^ew+ \gamma\sum_{e\in\mathcal{E}}\omega^e(w^\top{S}^ew - \min_{e\in\mathcal{E}}w^\top{S}^e{w})
\\&= w^\top\left[ \gamma\sum_{e\in\mathcal{E}}\omega^eS^e\right]w +(1-\gamma)\min_{e\in\mathcal{E}}w^\top{S}^e{w}\\
&=  w^\top\left[ \gamma\sum_{e\in\mathcal{E}}\omega^eS^e\right]w +(1-\gamma)\min_{e\in\mathcal{E}_\mathrm{small}}w^\top{S}^e{w}
\end{aligned}
\end{eqnarray*}
Here, the last inequality follows from the data-generating assumption. Thus, for each $b$, there exists $\bar{e}(b) \in \mathcal{E}_\mathrm{small}$ such that:
\begin{eqnarray*}
\begin{aligned}
\mathcal{L}_\gamma(b) 
&= w^\top\left[ \gamma\sum_{e\in\mathcal{E}}\omega^eS^e\right]w +(1-\gamma)w^\top{S}^{\bar{e}(b)}{w},
\end{aligned}
\end{eqnarray*}
where $w$ depends on $b$. Then, 
\begin{eqnarray*}
\begin{aligned}
\min_b \mathcal{L}_\gamma(b) 
&= \min_{b}w^\top\left[ \gamma\sum_{e\in\mathcal{E}}\omega^eS^e\right]w +(1-\gamma)w^\top{S}^{\bar{e}(b)}{w}\\
&\geq \min_{b}\min_{\tilde{e}\in\mathcal{E}} w^\top\left[ \gamma\sum_{e\in\mathcal{E}_\mathrm{small}}\omega^eS^e\right]w +(1-\gamma)w^\top{S}^{\tilde{e}}{w} \\
&= \min_{\tilde{e}\in\mathcal{E}_\mathrm{small}}\min_{b}\mathcal{L}_\gamma^{\tilde{e}}(b) = \mathcal{L}_\gamma^{\bar{e}}(b_\mathrm{opt}^{\bar{e}}) \\
&=  w(b_\mathrm{opt}^{\bar{e}})^\top\left[ \gamma\sum_{e\in\mathcal{E}_\mathrm{small}}\omega^eS^e\right]w(b_\mathrm{opt}^{\bar{e}}) +(1-\gamma)w(b_\mathrm{opt}^{\bar{e}})^\top{S}^{\tilde{e}}{w}(b_\mathrm{opt}^{\bar{e}})
\end{aligned}
\end{eqnarray*}
Now notice that:
\begin{eqnarray*}
\begin{aligned}
\mathcal{L}_\gamma(b_\mathrm{opt}^{\bar{e}}) &= w(b_\mathrm{opt}^{\bar{e}})^\top\left[ \gamma\sum_{e\in\mathcal{E}}\omega^eS^e\right]w(b_\mathrm{opt}^{\bar{e}}) +(1-\gamma)w(b_\mathrm{opt}^{\bar{e}})^\top{S}^{\bar{e}(b)}{w}(b_\mathrm{opt}^{\bar{e}}) \\
&\leq w(b_\mathrm{opt}^{\bar{e}})^\top\left[ \gamma\sum_{e\in\mathcal{E}}\omega^eS^e\right]w(b_\mathrm{opt}^{\bar{e}}) +(1-\gamma)w(b_\mathrm{opt}^{\bar{e}})^\top{S}^{\bar{e}}{w}(b_\mathrm{opt}^{\bar{e}})\\
&=\mathcal{L}_\gamma^{\bar{e}}(b_\mathrm{opt}^{\bar{e}})
\end{aligned}
\end{eqnarray*}
Thus, we have concluded that: 
\begin{eqnarray*}
\min_b \mathcal{L}_\gamma(b) = \min_{b}\mathcal{L}_\gamma^{\bar{e}}(b) = \min_b w^\top\left[S^{\bar{e}}+ \gamma\sum_{e\in\mathcal{E}}\omega^e(S^e-{S}^{\bar{e}})\right]{w},
\end{eqnarray*}
\end{proof}

\subsection{Proof of Theorem~\ref{thm:solution_infty}}
\label{proof_prop_solution_infty}
Since Assumption~\ref{ass:strength} is strictly stronger than Assumption~\ref{ass:strength_2}, the first part of Theorem~\ref{thm:solution_infty} follows from the first part of Theorem~\ref{thm:solution_infty_2}. So we prove Theorem~\ref{thm:solution_infty_2}. 

{\textbf{Proof of Theorem~\ref{thm:solution_infty_2}}} Recall our block notations $B^\star = \begin{pmatrix}B^\star_{x} & b^\star \\ B^{\star\top}_{yx} & 0 \end{pmatrix}$ where $b^\star=b^\star$, and $C^\star = \begin{pmatrix}C^\star_{x} & C^\star_{xy} \\ C^{\star\top}_{yx} & C^\star_{y} \end{pmatrix}$.
Denote by $\cL_{\mathrm{reg}}(b)$ the regularization term in the objective function \eqref{eq:obj}. 

When $\gamma \to \infty$, it is straightforward to check that if $\cL_{\mathrm{reg}}(b)$ has a minimizer, then, DRIG solves the following optimization problem
\begin{eqnarray}
\begin{aligned}
\min_b &\quad \min_e \mathbb{E}[\ell(X^e,Y^e;b)] \\
\text{subject to }&:\quad b \in \argmin_{\tilde{b}} \cL_{\mathrm{reg}}(\tilde{b})
\end{aligned}
\label{eqn:drig_infinite_gamma}
\end{eqnarray}
Notice that for any $b$, there exists $\tilde{e}(b)\in\mathcal{E}_\text{small}$ such that:
\begin{eqnarray*}
\begin{aligned}
\cL_{\mathrm{reg}}(b) = w^\top\bbE\left[\sum_{e\in\mathcal{E}}\omega^e(S^e-S^{\tilde{e}(b)})\right]w \geq \min_{\tilde{e}\in\mathcal{E}_\mathrm{small}}w^\top\left[\sum_{e\in\mathcal{E}}\omega^e(S^e-S^{\tilde{e}})\right]w\geq 0.
\end{aligned}
\end{eqnarray*}
Here, the last inequality follows from $\left[\sum_{e\in\mathcal{E}}\omega^e(S^e-S^{\tilde{e}})\right] \succeq 0$ from the data generating process. Since $\cL_{\mathrm{reg}}(b)$ is bounded above by zero, it must have a global minimizer. Thus, for $\gamma \to \infty$, DRIG minimizes \eqref{eqn:drig_infinite_gamma}. Since
\begin{eqnarray*}
\begin{aligned}
\cL_{\mathrm{reg}}(b) &= w^\top\left(\sum_{e\in\mathcal{E}}\omega^e{S}^e\right)w - \min_{\tilde{e}\in\mathcal{E}_\text{small}} w^\top{S}^{\tilde{e}}w \\
&=\max_{e\in\mathcal{E}_\text{small}} w^\top\left(\sum_{e\in\mathcal{E}}\omega^e({S}^e-S^{\tilde{e}})\right)w. 
\end{aligned}
\end{eqnarray*}
Since $\sum_{e\in\mathcal{E}}\omega^e({S}^e-S^{\tilde{e}}) \succ 0$, we have that $\cL_{\mathrm{reg}}(b)$ is point-wise maximum of convex functions which is a convex function. For a convex function, any local minimizer is a global minimizer, so we establish the first part of the theorem.

We now prove the second part of the theorem. Our goal is to show that $\argmin_b \cL_\mathrm{reg}(b) = b^\star$. Since  $\sum_{e\in\mathcal{E}}\omega^e({S}^e-S^{\tilde{e}}) \succ 0$, we have that $\cL_\mathrm{reg}(b) \geq 0$. Using the notation of the theorem, we have:
\begin{eqnarray*}
\begin{aligned}
\cL_{\mathrm{reg}}(b) &=\max_{e\in\mathcal{E}_\text{small}} w^\top{L}^{\star,\bar{e}}w,
\end{aligned}
\end{eqnarray*}
where $w$ is of the form \eqref{eq:def_w}. Let $\alpha=1-B^{\star\top}_{xy}(I_p-B^\star_x)^{-1}B^\star_{yx}$. We have 
\begin{equation*}
	\begin{split}
		C^\star_x&=(I_p-B^\star_x-B^\star_{yx}b^{\star\top})^{-1}~~~;~~~~
		C^\star_{xy}=(I_p-B^\star_x)^{-1}B^\star_{yx}/\alpha\\
		C^\star_{yx}&=C^{\star\top}_xb^\star~~~~;~~~~
		C^\star_{y}=1/\alpha
	\end{split}
\end{equation*}
Then we have the following equivalent definition of $w$.
\begin{equation*}
	w = 
	\begin{pmatrix}
		C^{\star\top}_x (b^\star-b)\\
		1/\alpha - C^{\star\top}_{xy} b
	\end{pmatrix} 
	=:
	\begin{pmatrix}
		w_x \\ w_y
	\end{pmatrix},
\end{equation*}
where $w_y\in\bbR$ is the last component of $w$. Thus,
\begin{eqnarray*}
\begin{aligned}
\cL_{\mathrm{reg}}(b) &=\max_{e\in\mathcal{E}_\text{small}} (b-b^\star)^T({C^\star}{L}^{\star,\bar{e}}{C^\star}^T)_{1:p,1:p}(b-b^\star),
\end{aligned}
\end{eqnarray*}
Since $\mathrm{rank}(({C^\star}{L}^{\star,\bar{e}}{C^\star}^T)_{1:p,1:p})=p$ for every $\bar{e}\in\mathcal{E}_\text{small}$, we have that:
\begin{eqnarray*}
\begin{aligned}
\argmin_{b}\cL_{\mathrm{reg}}(b) = b^\star.
\end{aligned}
\end{eqnarray*}

\begin{proof}[Proof of Theorem~\ref{thm:solution_infty}] Since we have an `observational' environment according to Assumption~\ref{ass:strength}, 
\begin{eqnarray*}
\begin{aligned}
\cL_{\mathrm{reg}}(b) &=\max_{e\in\mathcal{E}_\text{small}} w^\top{L}^{\star,\bar{e}}w = w^\top{L}^{\star,0}w 
\end{aligned}
\end{eqnarray*}
Here, the notation of $L^{\star,e}$ is defined in Theorem~\ref{thm:solution_infty_2}. For simplicity, let $L^\star:= L^{\star,0}$. Notice that:
\begin{equation*}
	\cL_{\mathrm{reg}}(b) = w_x^\top L^\star_xw_x + 2w_x^\top L^\star_{xy} w_y + w_y^2L^\star_{y}.
\end{equation*}
Taking the gradient of $\cL_{\mathrm{reg}}(b)$ with respect to $b$ and setting it to zero, we have 
\begin{align*}
	C^\star_xL^\star_{xy}+L^\star_{y}C^\star_{xy} &=  \big(C^\star_xL^\star_x C^{\star\top}_x + C^\star_xL^\star_{xy}C^{\star\top}_{xy} + C^\star_{xy}L^{\star\top}_{xy} C^{\star\top}_x + L^\star_{y}C^\star_{xy}C^{\star\top}_{xy}\big)(b^{\mathrm{opt}}_\infty-b^\star)\\
	&= [C^\star{L^\star}{C^\star}^\top]_{1:p,1:p}(b^{\mathrm{opt}}_\infty-b^\star)
\end{align*}
which leads to the desired result.
\end{proof}

\subsection{Proof of Theorem~\ref{thm:robust_adap}}\label{proof_thm_robust_adap}
\begin{proof}[Proof of Theorem~\ref{thm:robust_adap}]
We have
\begin{align*}
	\sup_{v\in \mathcal{C}_\mathrm{DRIG\text{-}A+}^\Gamma}\bbE_v[Y-b^\top X]^2 &= \sup_{v\in \mathcal{C}_\mathrm{DRIG\text{-}A+}^\Gamma}\begin{pmatrix}-b^\top &1\end{pmatrix}(I-B^\star)^{-1}\bbE[vv^\top](I-B^\star)^{-\top}\begin{pmatrix}-b\\1\end{pmatrix}\\
	&=\begin{pmatrix}-b^\top &1\end{pmatrix}\Gamma(I-B^\star)^{-1}\sum_{e\in\cE}\omega^e\Big(\bbE[\delta^e{\delta^e}^\top]-\bbE[\delta^0{\delta^0}^\top]\Big)(I-B^\star)^{-\top}\Gamma\begin{pmatrix}-b\\1\end{pmatrix}\\
&=\tilde{w}^\top\sum_{e\in\cE}\omega^e\Big(\bbE[\delta^e{\delta^e}^\top]-\bbE[\delta^0{\delta^0}^\top]\Big)\tilde{w},
\end{align*}
where $\tilde{w}=(I-B)^{-\top}\Gamma\begin{pmatrix}-b\\1\end{pmatrix}$.

Note that
\begin{equation*}
	\gamma_y Y^e - b^\top\Gamma_x X^e = \begin{pmatrix}-b^\top &1\end{pmatrix}\Gamma(I-B)^{-1}(\varepsilon+\delta^e) = \tilde{w}^\top (\varepsilon+\delta^e).
\end{equation*}
Then for all $e\in\cE$,
\begin{equation*}
	\bbE(\gamma_y Y^e - b^\top\Gamma_x X^e)^2 = \tilde{w}^\top(\bbE\varepsilon\varepsilon^\top + \bbE\delta^e{\delta^e}^\top)\tilde{w}
\end{equation*}
and thus 
\begin{equation*}
	\bbE(\gamma_y Y^e - b^\top\Gamma_x X^e)^2 - \bbE(\gamma_y Y^0 - b^\top\Gamma_x X^0)^2 = \tilde{w}^\top \Big(\bbE[\delta^e{\delta^e}^\top]-\bbE[\delta^0{\delta^0}^\top]\Big)\tilde{w}
\end{equation*}
Also we have $\bbE(Y^0-b^\top X^0)^2=w^\top(\bbE[\varepsilon\varepsilon^\top]+\bbE[\delta^0{\delta^0}^\top]) w$ as above. Thereby, the desired result follows.
\end{proof}

\subsection{Proof of Theorem~\ref{thm:adap_mse}}\label{proof_thm_adap_mse}

\begin{lemma}\label{lem:trace_ineq}
	Given a unit vector $\nu\in\bbR^p$ ($p>1$) and a $p\times p$ positive definite real matrix $K\succ 0$, we have $\tr(K)>\tr(\nu\nu^\top K \nu\nu^\top)$. 
\end{lemma}
\begin{proof}
	Let $K=Q\Lambda Q^\top$ be the eigendecomposition of $K$ where $\Lambda$ is a diagonal matrix of eigenvalues $\lambda_i>0,i=1,\dots,p$ and $Q$ is orthogonal. Let $\tilde{\nu}=Q^\top v$, so $\|\tilde{\nu}\|^2=1$. We have
	\begin{align*}
		\tr(\nu\nu^\top K \nu\nu^\top) = \nu^\top K \nu \nu^\top \nu = \nu^\top Q\Lambda Q^\top \nu = \tilde{\nu}^\top \Lambda\tilde{\nu} = \sum_{i=1}^p \lambda_i\tilde{\nu}_i^2
	\end{align*}
	Note from $\|\tilde{\nu}\|^2=1$ that $\sum_{i=1}^p \lambda_i\tilde{\nu}_i^2\leq\sum_{i=1}^p \lambda_i$.
	Now, claim $\sum_{i=1}^p \lambda_i\tilde{\nu}_i^2<\sum_{i=1}^p \lambda_i$. Otherwise, we must have for all $i$ that $\lambda_i\tilde{\nu}_i^2=\lambda_i$ and then $\tilde{\nu}_i^2=1$. This means $\|\tilde{\nu}\|^2=p>1$. Contradiction.
	
	Thus,
	\begin{equation*}
		\tr(vv^\top K vv^\top) < \sum_{i=1}^p \lambda_i = \tr(K),
	\end{equation*}
	which concludes the proof.
\end{proof}

\begin{proof}[Proof of Theorem~\ref{thm:adap_mse}]
Let 
\begin{align*}
	\tilde\gamma_y &= \frac{({(\Sigma^v_x)}^{-1/2}\Gamma^\star_x\Delta_{xy})^\top}{\|{(\Sigma^v_x)}^{-1/2}\Gamma^\star_x\Delta_{xy}\|^2}{(\Sigma^v_x)}^{-1/2}(\hat\Sigma^v_{xy}-\bbE X^0Y^0)
\end{align*}
Let $b^{\mathrm{opt}}_{\tilde\Gamma}$ be the DRIG-A solution with $\tilde\Gamma=
\begin{pmatrix}
	\Gamma^\star_x & 0\\0 &\tilde\gamma_y
\end{pmatrix}$
and $\tilde{b}_{\mathrm{tOLS}}={\Sigma^v_x}^{-1}\hat\Sigma^v_{xy}$ which are obtained based on the finite labeled sample and infinite unlabeled sample $P_{\mathrm{test}}^x$. Note that
\begin{equation*}
	b^{\mathrm{opt}}_{\tilde\Gamma}={(\Sigma_x^v)}^{-1}[\bbE X^0Y^0+\tilde\gamma_y\Gamma_x\Delta_{xy}] =: {(\Sigma_x^v)}^{-1}\hat\Sigma_{xy}^{(2)}.
\end{equation*}
For notational simplicity, below we omit the superscript $v$ in $\Sigma_x^v,\Sigma_{xy}^v,\hat\Sigma_x^v,\hat\Sigma_{xy}^v$ without introducing ambiguity. The remainder of the proof proceeds in two steps.

\medbreak
\noindent{\textit{Step I. }} 
We first compare the test MSEs of $b^{\mathrm{opt}}_{\tilde\Gamma}$ and $\tilde{b}_{\mathrm{tOLS}}$, given by
\begin{align*}
	\cL_{\mathrm{test}}(\tilde{b}_{\mathrm{tOLS}}) &= \hat\Sigma_{xy}^\top\Sigma_x^{-1}\hat\Sigma_{xy} - 2\Sigma_{xy}\Sigma_x^{-1}\hat\Sigma_{xy} + \bbE[Y^v]^2\\
	\cL_{\mathrm{test}}(b^{\mathrm{opt}}_{\tilde\Gamma}) &= \hat\Sigma_{xy}^{(2)\top}\Sigma_x^{-1}\hat\Sigma_{xy}^{(2)} - 2\Sigma_{xy}\Sigma_x^{-1}\hat\Sigma_{xy}^{(2)} + \bbE[Y^v]^2.
\end{align*}
The expected differences from the minimal test MSE are
\begin{align*}
	\bbE[\cL_{\mathrm{test}}(\tilde{b}_{\mathrm{tOLS}})] - \min_b \cL_{\mathrm{test}}(b) &= \tr\Big[\Sigma_x^{-1}\Big(\bbE\hat\Sigma_{xy}\hat\Sigma_{xy}^\top-\Sigma_{xy}\Sigma_{xy}^\top\Big)\Big] - 2\Sigma_{xy}^\top\Sigma_x^{-1}(\bbE\hat\Sigma_{xy}-\Sigma_{xy})\\
	&= \tr\Big[\cov\Big( \Sigma_x^{-1/2}\hat\Sigma_{xy}\Big)\Big]
\end{align*}
 \begin{align*}
	&\quad\ \bbE[\cL_{\mathrm{test}}(b^{\mathrm{opt}}_{\tilde\Gamma})] - \min_b \cL_{\mathrm{test}}(b) \\
	&= \tr\Big[\Sigma_x^{-1}\Big(\bbE\hat\Sigma_{xy}^{(2)}\hat\Sigma_{xy}^{(2)\top}-\bbE\hat\Sigma_{xy}^{(2)}\bbE\hat\Sigma_{xy}^{(2)\top}\Big)\Big] + \tr\Big[\Sigma_x^{-1}\big(\bbE\hat\Sigma_{xy}^{(2)}-\Sigma_{xy}\big)\big(\bbE\hat\Sigma_{xy}^{(2)}-\Sigma_{xy}\big)^\top\Big]\\
	&= \tr\Big[\cov\Big( \Sigma_x^{-1/2}\hat\Sigma_{xy}^{(2)}\Big)\Big] + \tr\Big[\big(\Sigma_x^{-1/2}\bbE\hat\Sigma_{xy}^{(2)}-\Sigma_x^{-1/2}\Sigma_{xy}\big)\big(\Sigma_x^{-1/2}\bbE\hat\Sigma_{xy}^{(2)}-\Sigma_x^{-1/2}\Sigma_{xy}\big)^\top\Big].
\end{align*}
Then
\begin{align*}
	&\bbE[\cL_{\mathrm{test}}(\tilde{b}_{\mathrm{tOLS}})]-\bbE[\cL_{\mathrm{test}}(b^{\mathrm{opt}}_{\tilde\Gamma})] \\
	=& \tr\Big[\cov\Big(\Sigma_x^{-1/2}\hat\Sigma_{xy}\Big)-\cov\Big( \Sigma_x^{-1/2}\hat\Sigma_{xy}^{(2)}\Big)\Big] - \tr\Big[\big(\Sigma_x^{-1/2}\bbE\hat\Sigma_{xy}^{(2)}-\Sigma_x^{-1/2}\Sigma_{xy}\big)\big(\Sigma_x^{-1/2}\bbE\hat\Sigma_{xy}^{(2)}-\Sigma_x^{-1/2}\Sigma_{xy}\big)^\top\Big].
\end{align*}

Let $\tilde\Sigma=\Sigma_x^{-1/2}\cov(X^vY^v)\Sigma_x^{-1/2}$. Then $\cov(\Sigma_x^{-1/2}\hat\Sigma_{xy})=\tilde\Sigma/n_l$.
By definition,
\begin{equation*}
	 \Sigma_x^{-1/2}\hat\Sigma_{xy}^{(2)} = \Sigma_x^{-1/2}\Gamma_x\Delta_{xy}\frac{(\Sigma_x^{-1/2}\Gamma_x\Delta_{xy})^\top}{\|\Sigma_x^{-1/2}\Gamma_x\Delta_{xy}\|_2^2}\Sigma_x^{-1/2}(\hat\Sigma_{xy}-\bbE X^0Y^0) + \Sigma_x^{-1/2}\bbE X^0Y^0.
\end{equation*}
Let
\begin{equation*}
	\xi = \Sigma_x^{-1/2}\Gamma_x\Delta_{xy}/\|\Sigma_x^{-1/2}\Gamma_x\Delta_{xy}\|_2.
\end{equation*}
We know $\|\xi\|_2=1$ and  
\begin{equation*}
	\Sigma_x^{-1/2}\hat\Sigma_{xy}^{(2)} = \xi\xi^\top\Sigma_x^{-1/2}(\hat\Sigma_{xy}-\bbE X^0Y^0) + \Sigma_x^{-1/2}\bbE X^0Y^0.
\end{equation*}
Then
\begin{equation*}
	\cov(\Sigma_x^{-1/2}\hat\Sigma_{xy}^{(2)}) = \cov\big(\xi\xi^\top\Sigma_x^{-1/2}\hat\Sigma_{xy}\big) = \xi\xi^\top\cov\big(\Sigma_x^{-1/2}\hat\Sigma_{xy}\big) \xi\xi^\top
\end{equation*}
and
\begin{align*}
	\Sigma_x^{-1/2}\bbE\hat\Sigma_{xy}^{(2)}-\Sigma_x^{-1/2}\Sigma_{xy} &= \xi\xi^\top\Sigma_x^{-1/2}(\Sigma_{xy}-\bbE X^0Y^0) + \Sigma_x^{-1/2}\bbE X^0Y^0 - \Sigma_x^{-1/2}\Sigma_{xy} \\
	&= (\xi\xi^\top-I)\Sigma_x^{-1/2}(\Sigma_{xy}-\bbE X^0Y^0).
\end{align*}
Let $\eta=\Sigma_x^{-1/2}(\Sigma_{xy}-\bbE[X^0Y^0])$.

Then 
\begin{equation*}
	\bbE[\cL_{\mathrm{test}}(\tilde{b}_{\mathrm{tOLS}})]-\bbE[\cL_{\mathrm{test}}(b^{\mathrm{opt}}_{\tilde\Gamma})] = \frac{1}{m}\tr\Big(\tilde\Sigma-\xi\xi^\top\tilde\Sigma\xi\xi^\top\Big) - \tr\Big[(\xi\xi^\top-I)\eta\eta^\top(\xi\xi^\top-I)^\top\Big].
\end{equation*}
Since $\|\xi\|_2=1$, we have
\begin{align*}
	\tr\big[(\xi\xi^\top-I)\eta\eta^\top(\xi\xi^\top-I)^\top\big] &= \eta^\top(\xi\xi^\top-I)(\xi\xi^\top-I)^\top \eta = \eta^\top(I-\xi\xi^\top) \eta \\
	&= \eta^\top\eta - \xi^\top\eta\eta^\top\xi = \tr(\eta\eta^\top - \xi\xi^\top\eta\eta^\top\xi\xi^\top).
\end{align*}
Then
\begin{equation*}
	\bbE[\cL_{\mathrm{test}}(\tilde{b}_{\mathrm{tOLS}})]-\bbE[\cL_{\mathrm{test}}(b^{\mathrm{opt}}_{\tilde\Gamma})] = \frac{1}{n_l}\tr\big(\tilde\Sigma-\xi\xi^\top\tilde\Sigma\xi\xi^\top\big) - \tr\big(\eta\eta^\top - \xi\xi^\top\eta\eta^\top\xi\xi^\top\big)
\end{equation*}

Thus, it suffices to show 
\begin{eqnarray*}
	&&\tr(\tilde\Sigma-\xi\xi^\top\tilde\Sigma\xi\xi^\top) > \tr(\eta\eta^\top - \xi\xi^\top\eta\eta^\top\xi\xi^\top)\\
	&\Leftrightarrow & \tr(\tilde\Sigma-\eta\eta^\top)>\tr[\xi\xi^\top(\Sigma-\eta\eta^\top)\xi\xi^\top]\\
	&\Leftrightarrow & \tilde\Sigma-\eta\eta^\top\succ0\\
	&\Leftrightarrow & \Sigma_x^{-1/2}\cov(X^vY^v)\Sigma_x^{-1/2}\succ\Sigma_x^{-1/2}(\Sigma_{xy}-\bbE[X^0Y^0])(\Sigma_{xy}-\bbE[X^0Y^0])^\top\Sigma_x^{-1/2}\\
	&\Leftrightarrow & \cov(X^vY^v)\succ(\Sigma_{xy}-\bbE[X^0Y^0])(\Sigma_{xy}-\bbE[X^0Y^0])^\top,
\end{eqnarray*}
where the second equivalence follows by applying Lemma~\ref{lem:trace_ineq} with $A=\tilde\Sigma-\eta\eta^\top$ and $\nu=\xi$, and the third equivalence comes from the notations. Then by taking $N_l$ as the largest integer that is smaller than $\tr(\tilde\Sigma-\xi\xi^\top\tilde\Sigma\xi\xi^\top) / \tr(\eta\eta^\top - \xi\xi^\top\eta\eta^\top\xi\xi^\top)>1$ as already shown, we have 
\begin{equation*}
	\bbE[\cL_{\mathrm{test}}(\tilde{b}_{\mathrm{tOLS}})]>\bbE[\cL_{\mathrm{test}}(b^{\mathrm{opt}}_{\tilde\Gamma})].
\end{equation*}

\medbreak
\noindent{\textit{Step II. }} By the weak law of large numbers, we have $\|\hat\Sigma_x-\Sigma_x\|\to0$ as $n_u\to\infty$. Then by Slutsky's theorem, we have $\cL_{\mathrm{test}}(\hat{b}_{\mathrm{tOLS}})\pto\cL_{\mathrm{test}}(\tilde{b}_{\mathrm{tOLS}})$ and $\cL_{\mathrm{test}}(b^{\mathrm{opt}}_{\tilde\Gamma})\pto \cL_{\mathrm{test}}(b^{\mathrm{opt}}_{\Gamma^\star})$ as $n_u\to\infty$.

{Due to the boundedness assumption,} this implies $\bbE[\cL_{\mathrm{test}}(\hat{b}_{\mathrm{tOLS}})]\to\bbE[\cL_{\mathrm{test}}(\tilde{b}_{\mathrm{tOLS}})]$ and $\bbE[\cL_{\mathrm{test}}(b^{\mathrm{opt}}_{\tilde\Gamma})]\to\bbE[\cL_{\mathrm{test}}(b^{\mathrm{opt}}_{\Gamma^\star})]$ as $n_u\to\infty$. Thus, there exists $N_u$ such that for all $n_u>N_u$, it holds that 
\begin{align*}
	|\bbE[\cL_{\mathrm{test}}(\hat{b}_{\mathrm{tOLS}})]-\bbE[\cL_{\mathrm{test}}(\tilde{b}_{\mathrm{tOLS}})]| &< (\bbE[\cL_{\mathrm{test}}(\tilde{b}_{\mathrm{tOLS}})] - \bbE[\cL_{\mathrm{test}}(b^{\mathrm{opt}}_{\tilde\Gamma})])/2\\
	|\bbE[\cL_{\mathrm{test}}(b^{\mathrm{opt}}_{\tilde\Gamma})] - \bbE[\cL_{\mathrm{test}}(b^{\mathrm{opt}}_{\Gamma^\star})]| &< (\bbE[\cL_{\mathrm{test}}(\tilde{b}_{\mathrm{tOLS}})] - \bbE[\cL_{\mathrm{test}}(b^{\mathrm{opt}}_{\tilde\Gamma})])/2
\end{align*}
Thus, we have $\bbE[\cL_{\mathrm{test}}(\hat{b}_{\mathrm{tOLS}})]>\bbE \cL_{\mathrm{test}}(b^{\mathrm{opt}}_{\hat\Gamma})$, which concludes the proof.
\end{proof}

\subsection{Proof of the robustness results for anchor regression}\label{proof_robust_anchor}

\begin{proof}[Proof for anchor regression]
For any regression coefficient $b$, define the vector $w$ as in \eqref{eq:def_w}. 
We note from the SCM \eqref{eq:scm_train} that
\begin{equation*}
	Y^e-b^\top X^e = w^\top\varepsilon^e,
\end{equation*}
\begin{equation*}
	\bbE(Y^e-b^\top X^e) = w^\top \mu^e.
\end{equation*}

Denote by $(X^e,Y^e)$ the random variables follow the conditional distribution of $(X,Y)$ given $A=a^e$. 
Then we have
\begin{equation*}
	\bbE[(P_A(Y-b^\top X))^2] = \sum_{e\in\cE}\omega^e[\bbE(Y^e-b^\top X^e)]^2 = w^\top\left[\sum_{e\in\cE}\omega^e\mu^e{\mu^e}^\top\right]w
\end{equation*} and
\begin{align*}
	\bbE[((I-P_A)(Y-b^\top X))^2] &= \bbE[(Y-b^\top X-\bbE(Y-b^\top X|A))^2]\\
	&= \sum_{e\in\cE}\omega^e\bbE[(Y^e-b^\top X^e-\bbE(Y^e-b^\top X^e))^2]\\
	&= \sum_{e\in\mathcal{E}}\omega^e\bbE[(Y^e-\bbE Y^e-b^\top (X^e-\bbE X^e))^2]\\
	&= \sum_{e\in\cE}\omega^e\bbE[(w^\top(\varepsilon^e-\mu^e))^2]\\
	&= w^\top\left[\sum_{e\in\cE}\omega^e\bbE(\varepsilon^e-\mu^e)(\varepsilon^e-\mu^e)^\top\right]w,
\end{align*}
where the second term on the RHS is equal to 0 when $S^e=S^0$ for all $e$. Thus,
\begin{align*}
	\cL_{\mathrm{anchor},\gamma}(b) = w^\top\left[\sum_{e\in\cE}\omega^e\bbE(\varepsilon^e-\mu^e)(\delta^e-\mu^e)^\top\right]w + \gamma w^\top\left[\sum_{e\in\cE}\omega^e\mu^e{\mu^e}^\top\right]w.
\end{align*}
Then by analyzing the worst-case risk similarly to the proof of Theorem~\ref{thm:main}, we have 
	\begin{equation*}
		\cL_{\mathrm{anchor},\gamma}(b) = \sup_{v\in \C_\mathrm{anchor}^\gamma}\bbE[(Y^v-b^\top X^v)^2].
	\end{equation*}
\end{proof}

\subsection{Proof of the robustness results for group DRO}\label{proof_robust_gdro}

The objective function of group DRO is
\begin{align*}
	\max_{e\in\cE}\bbE[(Y^e-b^\top X^e)^2] &= \max_{e\in\cE} w^\top\bbE[\varepsilon^e{\varepsilon^e}^\top]w\\
	&= w^\top\bbE[\varepsilon^m{\varepsilon^m}^\top]w\\
	&= \sup_{v\in\mathcal{C}_{\mathrm{gDRO}}}\bbE[(Y^v-b^\top X^v)^2],
\end{align*}
which concludes the proof.

\subsection{Proof of the robustness results for the causal parameter}\label{proof_robust_causal}

\begin{proof}[Proof for the causal parameter]
Let $v_x$ and $w_x$ denote the first $p$ components of $v$ and $w$, respectively. Let $M = \mathbb{E}[vv^\top]-\mathbb{E}\left[\begin{pmatrix}v_x\\0 \end{pmatrix}\begin{pmatrix}v_x\\0 \end{pmatrix}^\top\right]$.
From the proof of Theorem~\ref{thm:main}, we have for a fixed $b$:
\begin{equation*}
    \sup_{v \in \mathcal{C}_\mathrm{causal}^\gamma} \bbE[(Y^v-b^\top X^v)^2] = \sup_{v \in \mathcal{C}_\mathrm{causal}^\gamma} w^\top \bbE[vv^\top ]w =\sup_{v \in \mathcal{C}_\mathrm{causal}^\gamma} w_x^\top\mathbb{E}[v_xv_x^\top]w_x + w^\top{M}w.
\end{equation*}
Notice that for any $v \in \mathcal{C}_\mathrm{DRIG}^\gamma$, the entries of $M$ are bounded. On the other hand, 
\begin{equation*}
\sup_{v\in\mathcal{C}_{\mathrm{causal}}^\gamma}w_x^\top\bbE[v_xv_x^\top]w_x = \sup_{v_x\in\mathbb{R}^{p}}w_x^\top\bbE[v_xv_x^\top]w_x=
	\begin{cases}
		0& \text{if }w_x=0,\\
		\infty&\text{otherwise}.
	\end{cases}
\end{equation*}
Note that $w_x=0$ if and only if $b=b^\star$. Thus
\begin{equation*}
	\argmin_b\sup_{v\in\mathcal{C}_{\mathrm{causal}}}\bbE[(Y^v-b^\top X^v)^2]= b^\star.
\end{equation*}
\end{proof}

\subsection{Deriving $\Gamma^\star$}\label{app:Gamma}
\begin{lemma}\label{lem:matrix_eq}
	Let $A$ and $B$ be $p\times p$ positive definite matrices. The solution to the equation $XBX=A$ is uniquely given by $X=B^{-1/2}(B^{1/2}AB^{1/2})^{1/2}B^{-1/2}$.
\end{lemma}
\begin{proof}
	The equation $XBX=A$ is equivalent to $(XB)^2=AB=B^{-1/2}ZB^{1/2}$ with $Z=B^{1/2}AB^{1/2}$. Then we have $XB=B^{-1/2}Z^{1/2}B^{1/2}$, leading to the desired result.
\end{proof}

\begin{proof}[Proof of deriving $\Gamma^\star$]
By Lemma~\ref{lem:matrix_eq}, we know the solution to $\bbE X^0{X^0}^\top +\Gamma_x\Delta_x\Gamma_x = {\Sigma^v_x}$ is uniquely given by
$${\Gamma}^{\star}_x := \Delta_x^{-1/2}\left[\Delta_x^{1/2}\left(\Sigma^v_x - \bbE X^0{X^0}^\top\right)\Delta_x^{1/2}\right]^{1/2}\Delta_x^{-1/2}.$$
On the other hand, $\gamma^\star_y$ is defined as the solution to minimizing the test MSE of the DRIG-A solution $b^\mathrm{opt}_\Gamma$. That is,
\begin{align*}
	\gamma^\star_y &= \argmin_{\gamma_y} \bbE\big[\big(Y^v - [{\Sigma^v_x}^{-1}(\bbE X^0Y^0+\gamma_y\Gamma^\star_x\Delta_{xy})]^\top X^v\big)^2\big]\\
	&= \argmin_{\gamma_y} \big[\gamma_y^2\Delta_{xy}^\top\Gamma^\star_x{\Sigma^v_x}^{-1}\Gamma^\star_x\Delta_{xy} + 2\gamma_y\Delta_{xy}^\top\Gamma^\star_x{\Sigma^v_x}^{-1}(\bbE X^0Y^0-\bbE X^v Y^v)\big]\\
	&= \frac{({(\Sigma^v_x)}^{-1/2}\Gamma^\star_x\Delta_{xy})^\top}{\|{(\Sigma^v_x)}^{-1/2}\Gamma^\star_x\Delta_{xy}\|^2}{(\Sigma^v_x)}^{-1/2}(\Sigma^v_{xy}-\bbE X^0Y^0),
\end{align*}
which concludes the proof.
\end{proof}

\subsection{Infinite robustness of DRIG-A}\label{app:causal_driga}
\begin{proposition}\label{prop:driga_infty}
	If $\mathrm{rank}([C^\star{L^\star}{C^\star}^\top]_{1:p,1:p}) = p$, the DRIG-A solution as $\|\Gamma\|_2\to\infty$ is uniquely given by
	\begin{equation*}
		\gamma_y\Gamma_x^{-1}\left[b^\star + \left([C^\star{L^\star}{C^\star}^\top]_{1:p,1:p}\right)^{-1}(C^\star_xL^\star_{xy}+L^\star_{y}C_{xy}^\star)\right],
	\end{equation*}
	which is not equal to the causal parameter $b^\star$ when $\Gamma_x/\gamma_y\neq I_p$, even in the identifiable case of Corollary~\ref{cor:identifiable}.
\end{proposition}
\begin{proof}[Proof of Proposition~\ref{prop:driga_infty}]
	Under the assumption in Proposition~\ref{prop:driga_infty}, DRIG-A with $\|\Gamma\|_2\to\infty$ is equivalent to 
\begin{equation*}
	\min_b \sum_{e\in\cE}\omega^e\big[\bbE(\gamma_y Y^e - b^\top\Gamma_x X^e)^2 - \bbE(\gamma_y Y^0 - b^\top\Gamma_x X^0)^2\big].
\end{equation*}
Then similar to the proof of Theorem~\ref{thm:solution_infty}, we obtain the minimum solution. 

Then under the conditions of Corollary~\ref{cor:identifiable}, it is straightforward to see the solution becomes $\gamma_y\Gamma_x^{-1}b^\star\neq b^\star$ unless $\Gamma_x/\gamma_y=I_p$. 
\end{proof}

\section{Proofs for results in supplementary materials}

\subsection{Proof of Theorem~\ref{thm:theorem_finite}}\label{proof_finite_sample}
We first introduce some notations. Let $\hat{G}^e$ be the sample gram matrix of the data $(X^e,Y^e)$ and  $G^{e}$ be the population gram matrix. We will let $\hat{G}^e_{X}$ and $\hat{G}^e_{XY}$ be the sub-blocks of $\hat{G}^e$; we will use the same notation for the population analog. Finally, we will let $\hat{F} = \sum_{e\in\mathcal{E}}\omega^e\left(\gamma\hat{G}^e_X-(\gamma-1)\hat{G}^0_X\right)$ and $\hat{g} = \sum_{e\in\mathcal{E}}\omega^e\left(\gamma\hat{G}^e_{XY}-(\gamma-1)\hat{G}^0_{XY}\right)$; we will let $F^\star$ and $g^\star$ be the population analogue.

Our analysis will rely on the following well-known concentration results for the sample covariance matrix of Gaussian random variables.
\begin{lemma}[Lemma 3.9 in \cite{Chand2012}]Let $\Sigma^\star \in \mathbb{R}^{d \times d}$ be the population covariance of a Gaussian random vector and $\hat{\Sigma}$ be the sample covariance from $n$ iid observations. Let $\psi = \|\Sigma^\star\|_2$. Given any $\delta > 0$ and $\delta \leq 8\psi$, let the number of samples $n$ be such that $n \geq \frac{64p\psi^2}{\delta}$. Then, we have that:
$$\text{Pr }[\|\hat{\Sigma}-\Sigma^\star\|_2 \geq \delta] \leq 2\exp\left\{\frac{-n\delta^2}{128\psi^2}\right\}.$$
\label{lemma:concentration_res}
\end{lemma}
A straightforward corollary is that under the setting of the lemma, letting $G^\star$ be the population Gram matrix of the Gaussian random vector and $\hat{G}$ be the estimate, 
$$\text{Pr }[\|\hat{G}-G^\star\|_2 \geq \delta] \leq 2\exp\left\{\frac{-n\delta^2}{128\psi^2}\right\}.$$
Combining the result above, and given the assumptions of Theorem~\ref{thm:theorem_finite}, we have that with probability exceeding $1-|\mathcal{E}|\mathrm{exp}(-p/2)$, $\hat{G}^0 \preceq \hat{G}^e$ for every $e \in \mathcal{E}$. Thus, with a high probability, 
$$\hat{\mathcal{L}}_\gamma(b) = \hat{\mathbb{E}}[\ell(X^0,Y^0;b)] + \gamma\sum_{e\in\mathcal{E}}\omega^e(\hat{\mathbb{E}}[\ell(X^e,Y^e;b)]-\hat{\mathbb{E}}[\ell(X^0,Y^0;b)]).$$

\textbf{$\hat{b}_\gamma$ convergence} We will begin with proving the convergence result for an estimate $\hat{b}_\gamma$. From optimality conditions, we have that with a high probability, $\hat{b}_\gamma$ satisfies $\hat{F}\hat{b}_\gamma = \hat{g}$. Note that:
\begin{equation*}
\begin{aligned}
\text{minimum eigenvalue}(\hat{F}) &\geq \tau_{\mathrm{min}} - \|\hat{F} - F^\star\|_2 \\
&\geq \tau_{\mathrm{min}} - \left[\sum_{e\in\mathcal{E}} \omega^e\left(\gamma\|\hat{\Sigma}^e -{\Sigma}^{e,\star}\|_2 + |\gamma-1|\|\hat{\Sigma}^0 - \Sigma^{0,\star}\|_2\right)\right] 
\end{aligned}
\end{equation*}
For any environment $e \in \mathcal{E}$, we let $\delta = 8\psi_e\sqrt{\frac{p}{n_e}}$. Appealing to Lemma~\ref{lemma:concentration_res} and the lower-bound on $n_e$ for every $e$, we have that with probability $1-2|\mathcal{E}|\exp(-p/2)$, $\text{minimum eigenvalue}(\hat{F}) \geq \tau_{\mathrm{min}}/2 > 0$. Thus, $\hat{b}_\gamma$ is a unique solution to finite-sample DRIG estimator. Note that the optimality condition $\hat{F}\hat{b}_\gamma = \hat{g}$ can be equivalently written as:
\begin{equation*}
    \begin{aligned}
        \hat{F}(\hat{b}_\gamma-b^\text{opt}_\gamma) + (\hat{F}-F^\star)b^\text{opt}_\gamma + (\hat{g}-g^\star) + F^\star b_\gamma^\text{opt}+g^\star = 0.
    \end{aligned}
\end{equation*}
From the optimality condition of the population DRIG estimator \eqref{eqn:DRIG}, we have that $F^\star b_\gamma^\text{opt}+g^\star = 0$. Thus, 
\begin{equation*}
\hat{b}_\gamma-b^\text{opt}_\gamma = \hat{F}^{-1}\left[(\hat{F}-F^\star)b^\text{opt}_\gamma + (\hat{g}-g^\star)\right]
\end{equation*}
Thus, we can arrive at the following euclidean norm bound for the difference $\hat{b}_\gamma-b^\text{opt}_\gamma$:
\begin{equation*}
\begin{aligned}
    \|\hat{b}_\gamma-b^\text{opt}_\gamma\|_2 &\leq \frac{1}{\text{mininum eigenvalue}(\hat{F})}\left[\|\hat{F}-F^\star\|_2\|b^\text{opt}_\gamma\|_2 + \|\hat{g}-g^\star\|_2\right] \\
    &\leq \frac{2}{\tau_{\mathrm{min}}}\left[\|\hat{F}-F^\star\|_2\|b^\text{opt}_\gamma\|_2 + \|\hat{g}-g^\star\|_2\right].
    \end{aligned}
\end{equation*}
Note that:
$$\max\{\|\hat{F}-F^\star\|_2,\|\hat{g}-g^\star\|_2\} \leq \left[\sum_{e\in\mathcal{E}}\omega^e\left(\gamma\|\hat{\Sigma}^e -{\Sigma}^{e,\star}\|_2 + |\gamma-1|\|\hat{\Sigma}^0 - \Sigma^{e,\star}\|_2\right)\right].$$
Letting $\delta = 8\psi_e\sqrt{\frac{p}{n_e}}$ for every $e$ and appealing to Lemma~\ref{lemma:concentration_res} and the lower-bound on $n_e$ for every $e$, we arrive at the bound for $\|\hat{b}_\gamma-b^\text{opt}_\gamma\|_2$ in the theorem statement.\\

\textbf{$\hat{L}_\gamma(\hat{b}_\gamma)$ convergence} Note for every $e$, some simple calculations yield:
\begin{equation*}
\begin{aligned}
    &\quad\ \mathbb{E}[(Y^e-({b}^\text{opt}_\gamma)^TX^e)^2] - \hat{\mathbb{E}}[(Y^e-(\hat{b}_\gamma)^TX^e)^2] \\
    &= \hat{\Sigma}_Y^e-\Sigma^{e,\star}_Y - 2[(\hat{b}_\gamma-b^\text{opt}_\gamma)^T\hat{\Sigma}_{XY}^e + (b^\text{opt}_\gamma)^T(\hat{\Sigma}_{XY}^e-{\Sigma}_{XY}^{\star,e})]\\&\quad +(\hat{b}_\gamma-b^\text{opt}_\gamma)^T\hat{\Sigma}^{e}_X\hat{b}_\gamma + (b^\text{opt}_\gamma)^T\hat{\Sigma}_X^e(\hat{b}_\gamma-b^\text{opt}_\gamma) + (b^\text{opt}_\gamma)^T(\hat{\Sigma}_X^e - \Sigma^{e,\star}_X)b^\text{opt}_\gamma 
\end{aligned}
\end{equation*}
For notational ease, let $\xi_e = 8\psi_e\sqrt{\frac{p}{n_e}}$ and $\theta$ be the bound for $\hat{b}_\gamma-b^\text{opt}_\gamma$. Then, appealing to Lemma~\ref{lemma:concentration_res} and the lower-bound on the sample size $n_e$, we have with probability $1-2|\mathcal{E}|\exp(-p/2)$, $\|\hat{\Sigma}^e - \Sigma^{e,\star}\|_2\leq \xi_e$. Thus, some manipulations yield:
\begin{equation*}
\begin{aligned}
    &|\mathbb{E}[(Y^e-({b}^\text{opt}_\gamma)^TX^e)^2] - \hat{\mathbb{E}}[(Y^e-(\hat{b}_\gamma)^TX^e)^2]| \\
    \leq\ &\xi_e + 2(\xi_e+\psi_e)\theta + \xi\|b^\text{opt}_\gamma\|_2 + 2\theta(\xi_e+\psi_e)(\theta+\|b^\star\|_2)+\|b^\star\|_2^2\xi_e. 
\end{aligned}
\end{equation*}
By the lower-bound on the sample size $n_e$, we have that $\xi_e \leq \theta$, $\xi_e \leq \psi_e$, and ${\theta} \leq 4(1+\|b^\text{opt}_\gamma\|_2)$. Putting everything together, we can conclude that:
\begin{equation*}
\begin{aligned}
    |\mathbb{E}[(Y^e-(b^\text{opt}_\gamma)^TX^e)^2] - \hat{\mathbb{E}}[(Y^e-(b^\text{opt}_\gamma)^TX^e)^2]| \leq 15\theta(1+\psi_e)(1+\|b^\text{opt}_\gamma\|_2)^2
\end{aligned}
\end{equation*}
We can then conclude that:
\begin{equation*}
\begin{aligned}
|\hat{L}_\gamma(\hat{b}_\gamma)-{L}(b_\gamma^\text{opt})| &\leq \max\{\gamma,|1-\gamma|\}\max_{e}|\mathbb{E}[(Y^e-(b^\text{opt}_\gamma)^TX^e)^2] - \hat{\mathbb{E}}[(Y^e-(b^\text{opt}_\gamma)^TX^e)^2]| \\&\leq  15\max\{\gamma,|1-\gamma|\}\theta(1+\max_e\psi_e)(1+\|b^\text{opt}_\gamma\|_2)^2
\end{aligned}
\end{equation*}
Plugging in the value of $\theta$, we have desired result.

\subsection{Proof of Theorem~\ref{thm:continous_discrete}}
\label{proof_thm_discrete_continuous}
\begin{proof}
We have
\begin{equation*}
    \bbE\left[\begin{pmatrix}
    X^e \\ Y^e \\H^e
    \end{pmatrix}|A^e\right]
    =(I-\tilde{B}^\star)^{-1}(MA^e+\bbE[\varepsilon^e|A^e]).
\end{equation*}
For any regression coefficient $b$, define the vector $w$ as in \eqref{eq:def_w}. 
Then $\tilde{Y}^e-b^\top\tilde{X}^e=w^\top(\varepsilon^e-\bbE[\varepsilon^e|A^e])$ and 
\begin{equation*}
	\bbE[(\tilde{Y}^e-b^\top\tilde{X}^e)^2] =w^\top\tilde{S}^ew.
\end{equation*}
Thus
\begin{equation*}
	\tilde\cL_\gamma(b) =  w^\top \tilde{S}^0 w + \gamma w^\top\sum_{e\in\cE}\omega^e(\tilde{S}^e-\tilde{S}^0)w.
\end{equation*}
Also note that
\begin{align*}
	\bbE(Y^e-b^\top X^e|A^e) &= w^\top(\bbE[\varepsilon^e|A^e]+MA^e), \\
	 \bbE[\bbE(Y^e-b^\top X^e|A^e)]^2&=w^\top(\bbE[\varepsilon^e|A^e]\bbE[\varepsilon^e|A^e]^\top+M\bbE[A^e{A^e}^\top]M^\top)w,
\end{align*}
which leads to
\begin{equation*}
	\lambda\sum_{e\in\cE}\omega^e\bbE[\bbE(Y^e-b^\top X^e|A^e)]^2 = \lambda w^\top\left[\sum_{e\in\cE}\omega^e\bbE[\varepsilon^e|A^e]\bbE[\varepsilon^e|A^e]^\top+ M\bbE[A^e{A^e}^\top]M^\top\right]w.
\end{equation*}
Thereby, we conclude the proof.
\end{proof}

\subsection{Proof of Theorem~\ref{thm:main2}}
\begin{proof}
Consider the DRIG objective $\mathcal{L}_\gamma(b)$. Using similar reasoning as above, we can conclude that:
\begin{eqnarray*}
\begin{aligned}
\mathcal{L}_\gamma(b) &= \min_{e\in\mathcal{E}} w^\top{S}^ew+ \gamma\sum_{e\in\mathcal{E}}\omega^e(w^\top{S}^ew - \min_{e\in\mathcal{E}}w^\top{S}^e{w})
\\&= w^\top\left[ \gamma\sum_{e\in\mathcal{E}}\omega^eS^e\right]w +(1-\gamma)\min_{e\in\mathcal{E}}w^\top{S}^e{w}
\end{aligned}
\end{eqnarray*}

Since $S^e \succeq 0$, we have that $\mathcal{L}_\gamma(b) \geq w^\top\bbE[\varepsilon\varepsilon^\top]w$. Since $K^\star_2 \preceq S^e$ for every $e\in\mathcal{E}$, then, for $\gamma \geq 1$, $(1-\gamma)\min_{e\in\mathcal{E}}w^\top{S}^e{w} \leq (1-\gamma)\min_{e\in\mathcal{E}}w^\top{K}_2^\star{w}$. Thus, 
\begin{eqnarray*}
\mathcal{L}_\gamma(b) \leq w^\top\left[K_2^\star + \gamma\sum_{e\in\mathcal{E}}\omega^e(S^e-K^\star_2)\right]w = \mathcal{L}_{\mathcal{C}_{2,\gamma}}(b).
\end{eqnarray*}
By definition, $S^e \preceq K^\star_1$ for some $e\in\mathcal{E}$. Then, for $\gamma \geq 1$, $(1-\gamma)\min_{e\in\mathcal{E}}w^\top{K}_1^\star{w} \leq (1-\gamma)\min_{e\in\mathcal{E}}w^\top{S}^e{w}$. Thus, 
\begin{eqnarray*}
\mathcal{L}_\gamma(b) \geq  w^\top\left[K_1^\star + \gamma\sum_{e\in\mathcal{E}}\omega^e(S^e-K^\star_1)\right]w. 
\end{eqnarray*}
Since $\mathcal{L}_\gamma(b)$ is also greater than $w^\top\bbE[\varepsilon\varepsilon^\top]w$, we conclude that $\mathcal{L}_\gamma(b) \geq \mathcal{L}_{\mathcal{C}_{1,\gamma}}(b)$. 

To prove the second component, recall our block notations $B^\star = \begin{pmatrix}B^\star_{x} & b^\star \\ B^{\star\top}_{yx} & 0 \end{pmatrix}$ where $C^\star = \begin{pmatrix}C^\star_{x} & C^\star_{xy} \\ C^{\star\top}_{yx} & C^\star_{y} \end{pmatrix}$. Consider $\mathcal{L}_{\mathcal{C}}^\mathrm{robust}(b)$ defined in \eqref{eqn:worst-case} where $\mathcal{C} = \{v\in\mathbb{R}^{p+1}~|~\mathbb{E}[vv^\top] \preceq \tilde{M}\}$ for some positive definite matrix $\tilde{M}$. According to model \eqref{eq:scm_train}, we have
\begin{equation*}
    \mathcal{L}_{\mathcal{C}}^\mathrm{robust}(b)=w^\top M w,
\end{equation*}
where $w$ depends on $b$, as defined in \eqref{eq:def_w}. Let $M = \tilde{M}$. Let $\alpha=1-B^{\star\top}_{xy}(I_p-B^\star_x)^{-1}B^\star_{yx}$. We have 
\begin{equation*}
	\begin{split}
		C^\star_x&=(I_p-B^\star_x-B^\star_{yx}b^{\star\top})^{-1}~~~;~~~~
		C^\star_{xy}=(I_p-B^\star_x)^{-1}B^\star_{yx}/\alpha\\
		C^\star_{yx}&=C^{\star\top}_xb^\star~~~~;~~~~
		C^\star_{y}=1/\alpha
	\end{split}
\end{equation*}

Then we have
\begin{equation*}
	w = 
	\begin{pmatrix}
		C^{\star\top}_x (b^\star-b)\\
		1/\alpha - C^{\star\top}_{xy} b
	\end{pmatrix} 
	=:
	\begin{pmatrix}
		w_x \\ w_y
	\end{pmatrix},
\end{equation*}
where $w_y\in\bbR$ is the last component of $w$. 
Then
\begin{equation*}
	\cL_{\mathcal{C}}(b) = w_x^\top M_xw_x + 2w_x^\top M_{xy} w_y + w_y^2M_{y}.
\end{equation*}
Since $M$ is positive definite, $b^\mathrm{opt} := \argmin_{b\in\mathbb{R}^p} \cL_{\mathcal{C}}(b)$ as unique minimizer. To find this minimizer, we take a gradient of $\cL_{\mathcal{C}}(b)$ with respect to $b$ and set it to zero. Some algebra gives:
\begin{eqnarray*}
b^\mathrm{opt} = b^\star + \left([C^\star{M}{C^\star}^\top]_{1:p,1:p}\right)^{-1}[C^\star_x{M}_{xy}+M_{y}C^\star_{xy}].
\end{eqnarray*}
Let $M_1 = K_1^\star+\gamma\sum_{e\in\mathcal{E}}\omega^e(S^e-K^\star_1)$ and $M_2 = K_2^\star+\gamma\sum_{e\in\mathcal{E}}\omega^e(S^e-K^\star_2)$. Note that,  $\mathcal{L}_{\mathcal{C}_1}^\mathrm{robust}(b) = w(b)^\top{M}_1w(b)$ and $\mathcal{L}_{\mathcal{C}_2}^\mathrm{robust}(b) = w(b)^\top{M}_2w(b)$, where the dependence of $w$ on $b$ is made explicit. Following the analysis above, we have that: 
\begin{eqnarray}
\begin{aligned}
b^\mathrm{opt}_{\gamma,1} &:= \argmin_{b\in\mathbb{R}^p}\mathcal{L}_{\mathcal{C}_1}^\mathrm{robust}(b) =  b^\star + \left([C^\star{M}_1{C^\star}^\top]_{1:p,1:p}\right)^{-1}[C^\star_x[{M}_1]_{xy}+[M_1]_{y}C^\star_{xy}],\\
b^\mathrm{opt}_{\gamma,2} &:= \argmin_{b\in\mathbb{R}^p}\mathcal{L}_{\mathcal{C}_2}^\mathrm{robust}(b) =b^\star + \left([C^\star{M}_2{C^\star}^\top]_{1:p,1:p}\right)^{-1}[C^\star_x[{M}_2]_{xy}+[M_2]_{y}C^\star_{xy}].
\end{aligned}
\label{eqn:b_eqs}
\end{eqnarray}
Then:
\begin{eqnarray*}
\begin{aligned}
\min_b\mathcal{L}_{\mathcal{C}_1}(b) - \min_b\mathcal{L}_{\mathcal{C}_2}(b) &= w(b^\mathrm{opt}_{\gamma,1})^\top{M}_1w(b^\mathrm{opt}_{\gamma,1}) - w(b^\mathrm{opt}_{\gamma,2})^\top{M}_2w(b^\mathrm{opt}_{\gamma,2}),\\
&=w(b^\mathrm{opt}_{\gamma,1})^\top(M_1-M_2)w(b^\mathrm{opt}_{\gamma,1}) + (w(b^\mathrm{opt}_{\gamma,1})-w(b^\mathrm{opt}_{\gamma,2}))^\top{M}_2w(b^\mathrm{opt}_{\gamma,1}) \\&+ w(b^\mathrm{opt}_{\gamma,1})^\top M_2(w(b^\mathrm{opt}_{\gamma,1})-w(b^\mathrm{opt}_{\gamma,2})),\\&+(w(b^\mathrm{opt}_{\gamma,1})-w(b^\mathrm{opt}_{\gamma,2}))^\top{M}_{2}(w(b^\mathrm{opt}_{\gamma,1})-w(b^\mathrm{opt}_{\gamma,2})),
\end{aligned}
\end{eqnarray*}
which allows us to obtain the bound:
\begin{eqnarray}
\begin{aligned}
\min_b\mathcal{L}_{\mathcal{C}_2}(b)-\min_b\mathcal{L}_{\mathcal{C}_1}(b)  &\leq \|M_1-M_2\|_2\|w(b^\mathrm{opt}_{\gamma,1})\|_2^2 + 2\|w(b^\mathrm{opt}_{\gamma,1})\|_2\|w(b^\mathrm{opt}_{\gamma,1})-w(b^\mathrm{opt}_{\gamma,2})\|_2\|M_2\|_2\\&+\|w(b^\mathrm{opt}_{\gamma,1})-w(b^\mathrm{opt}_{\gamma,2})\|_2^2\|M_2\|_2
\end{aligned}
\label{eqn:bound_itermediate}
\end{eqnarray}
It is straightforward to show that: 
\begin{eqnarray*}
\begin{aligned}
\|M_1-M_2\| &\leq (\gamma-1)\|K_1^\star-K_2^\star\|_2\\
\|w(b^\mathrm{opt}_{\gamma,1})\|_2 &\leq \frac{4\|C^\star\|_2^2\|M_1\|_2}{\sigma_\text{min}({C^\star}^\top{M}_1{C}^\star)} + \frac{1}{\alpha}+\|C^\star\|_2\|b^\star\|_2:= c_1\\
\|w(b^\mathrm{opt}_{\gamma,1})-w(b^\mathrm{opt}_{\gamma,2})\|_2 &\leq 2\|C^\star\|_2\|b^\mathrm{opt}_{\gamma,1}-b^\mathrm{opt}_{\gamma,2}\|_2
\end{aligned}
\end{eqnarray*}
From \eqref{eqn:b_eqs}, and some algebra, we have:
\begin{eqnarray*}
\begin{aligned}
\|b^\mathrm{opt}_{\gamma,1}-b^\mathrm{opt}_{\gamma,2}\|_2 &\leq 2\|([C^\star{M}_2{C^\star}^\top]_{1:p,1:p})^{-1}\|_2\|C^\star\|_2\|M_1-M_2\|_2\\&+\|([C^\star{M}_2{C^\star}^\top]_{1:p,1:p})^{-1}\|_2^2\|C^\star\|_2^2\|M_1-M_2\|_2\\&+ \frac{\|([C^\star{M}_2{C^\star}^\top]_{1:p,1:p})^{-1}\|_2^3\|C^\star\|_2^2\|M_1-M_2\|_2}{1-\|C^\star\|_2^2\|M_1-M_2\|_2}\\
&\leq \frac{4\max\{1,\|C^\star\|_2\}^2\|M_1-M_2\|_2}{\min\{\sigma_\mathrm{min}(C^\star{M_2}{C^\star}^\top),1\}^3(1-\|C^\star\|_2^2\|M_1-M_2\|_2)}\\
&\leq  \underbrace{\frac{4\max\{1,\|C^\star\|_2\}^2(\gamma-1)\|K^\star_1-K^\star_2\|_2}{\min\{\sigma_\mathrm{min}(C^\star{M_2}{C^\star}^\top),1\}^3(1-\|C^\star\|_2^2}}_{c_2}(\gamma-1)\|K^\star_1-K^\star_2\|_2) = c_2(\gamma-1)\|K^\star_1-K^\star_2\|_2
\end{aligned}
\end{eqnarray*}
Combining these bounds with \eqref{eqn:bound_itermediate}, we have that:
\begin{eqnarray}
\min_b\mathcal{L}_{\mathcal{C}_2}(b) - \min_b\mathcal{L}_{\mathcal{C}_1}(b) \leq c_3(\gamma-1)\|K^\star_1-K^\star_2\|_2,
\label{eqn:bound_difference_loss}
\end{eqnarray}
where $c_3 = (c_1^2+4c_1c_2\|C^\star\|_2\|M_2\|_2+2c_2\|C^\star\|_2\|M\|_2)$. For $b_\gamma^\mathrm{opt}$ denoting an optimal solution of \eqref{eqn:DRIG}, and since $\mathcal{L}_{\mathcal{C}_1}(b) \leq \mathcal{L}_{\gamma}(b) \leq \mathcal{L}_{\mathcal{C}_2}(b)$, 
\begin{eqnarray}
\begin{aligned}
\mathcal{L}_{\mathcal{C}_1}(b_\gamma^\mathrm{opt})-\min_b\mathcal{L}_{\mathcal{C}_1}(b) &\leq \mathcal{L}_{\gamma}(b_\gamma^\mathrm{opt})-\min_b\mathcal{L}_{\mathcal{C}_1}(b)
\leq \min_b \mathcal{L}_{\mathcal{C}_2}(b)-\min_b\mathcal{L}_{\mathcal{C}_1}(b),
\end{aligned}
\label{eqn:bound_difference_loss_2}
\end{eqnarray}
obtaining the desired result. Furthermore, from Taylor series expansion, we have that: 
\begin{eqnarray*}
\mathcal{L}_{\mathcal{C}_1}(b_\gamma^\mathrm{opt}) - \mathcal{L}_{\mathcal{C}_1}(b^\mathrm{opt}_{\gamma,1}) = \underbrace{\nabla_b\mathcal{L}_{\mathcal{C}_1}(b^\mathrm{opt}_{\gamma,1})^\top}_{=0}(b_\gamma^\mathrm{opt}-b^\mathrm{opt}_{\gamma,1}) + \frac{1}{2}(b_\gamma^\mathrm{opt}-b^\mathrm{opt}_{\gamma,1})^\top (C^\star_{1:p,:}M_1{C^\star}^\top_{1:p,:})(b_\gamma^\mathrm{opt}-b^\mathrm{opt}_{\gamma,1})
\end{eqnarray*}
Combining the above with \eqref{eqn:bound_difference_loss} and \eqref{eqn:bound_difference_loss_2}, we obtain $\|b_\gamma^\mathrm{opt}-b^\mathrm{opt}_{\gamma,1}\|_2 \leq \frac{2c_3(\gamma-1)\|K^\star_1-K^\star_2\|_2}{\sigma_\text{min}(C^\star_{1:p,:}M_1{C^\star}^\top_{1:p,:})}$. Similarly, 
\begin{eqnarray*}
\begin{aligned}
\|b_\gamma^\mathrm{opt}-b^\mathrm{opt}_{\gamma,2}\|_2 &\leq \|b_\gamma^\mathrm{opt}-b^\mathrm{opt}_{\gamma,1}\|_2+\|b_\gamma^\mathrm{opt}-b^\mathrm{opt}_{\gamma,2}\|_2 \leq 2c_2(\gamma-1)\|K_1^\star-K_2^\star\|_2 \\
\mathcal{L}_{\mathcal{C}_2}(b_\gamma^\mathrm{opt}) - \mathcal{L}_{\mathcal{C}_2}(b^\mathrm{opt}_{\gamma,2}) &\leq \frac{2c_3(\gamma-1)\|K^\star_1-K^\star_2\|_2}{\sigma_\text{min}(C^\star_{1:p,:}M_1{C^\star}^\top_{1:p,:})}
\end{aligned}
\end{eqnarray*}
Letting $c = \max\{\frac{2c_3}{\sigma_\text{min}(C^\star_{1:p,:}M_1{C^\star}^\top_{1:p,:})},\frac{2c_3}{\sigma_\text{min}(C^\star_{1:p,:}M_2{C^\star}^\top_{1:p,:})}\}$, and $c' = 2c_2$, we have the desired result. 
\end{proof}

\subsection{Proof of Corollary~\ref{cor:identifiable}}\label{proof_cor_identifiable}
\begin{proof}
As $\delta^e_{p+1}=0$, we have $L^\star_{y}=0$ and $L^\star_{xy}=0$. Thus, by Theorem~\ref{thm:solution_infty}, we immediately know that $b^{\mathrm{opt}}_\infty=b^\star$. To see the second part of the corollary, note that in this case we have $w_x(b^{\mathrm{opt}}_\infty)=0$ and thus $\cL_{\mathrm{reg}}(b^{\mathrm{opt}}_\infty) = w_x(b^{\mathrm{opt}}_\infty)^\top L^\star_xw_x(b^{\mathrm{opt}}_\infty) = 0.$ Also we have $\bbE[(Y^0-{b^\star}^\top X^0)^2] = \bbE[\varepsilon_y^2]$, which concludes the proof.
\end{proof}


\subsection{Proof of Corollary~\ref{corr:independent_Y}}\label{proof_corr_independent_Y}
\begin{proof}
By Theorem~\ref{thm:solution_infty}, in this case we have 
\begin{equation*}
	b^{\mathrm{opt}}_\infty-b^\star = \big(C^\star_xL^\star_x C^{\star\top}_x/L^\star_y + C^\star_{xy}C^{\star\top}_{xy}\big)^{-1}C^\star_{xy}.
\end{equation*}
Thus,
\begin{align*}
	\|b^{\mathrm{opt}}_\infty-b^\star\|_\infty &= \|\big(C^\star_xL^\star_x C^{\star\top}_x/L^\star_y + C^\star_{xy}C^{\star\top}_{xy}\big)^{-1}C^\star_{xy}\|_\infty\\
	&\leq \frac{\|C^\star_{xy}\|_\infty}{\min_{\|u\|_\infty=1}\|(C^\star_xL^\star_x {C^\star_x}^\top/L_{y} + C^\star_{xy}{C^\star_{xy}}^\top)u\|_\infty}.
\end{align*}

When $b^\star=0$, $C^\star_{xy}=0$ and thus the above upper bound vanishes, leading to $b^{\mathrm{opt}}_\infty=b^\star$. Also we have $w_x(b^{\mathrm{opt}}_\infty)=0$ and $w_y(b^{\mathrm{opt}}_\infty)=1$. Thus, we have
\begin{equation*}
	\cL_\gamma(b^{\mathrm{opt}}_\infty) = \bbE[\varepsilon_y^2]+\bbE[{\delta^0_y}^2] + \gamma L_y^\star.
\end{equation*}
which tends to infinity as $\gamma\to\infty$.
\end{proof}

\subsection{Proof of Proposition~\ref{prop:int_y_robust}}\label{proof_prop_int_y_robust}
\begin{proof}
Notice that
\begin{align*}
	\sup_{v\in \bbR^{p+1}:\bbE[v_y^2]\le c}\bbE[(Y^v-b^\top X^v)^2] &= w^\top\bbE[\varepsilon\varepsilon^\top]w+\sup_{v:\bbE[v_y^2]\le c}w^\top\bbE[vv^\top]w \\
	&=w^\top\bbE[\varepsilon\varepsilon^\top]w+\sup_{v:\bbE[v_y^2]\le c}[w_x^\top\bbE(v_xv_x^\top)w_x+w_y^2\bbE(v_y^2)+2w_x^\top\bbE(v_xv_y)w_y]\\
	&=w^\top\bbE[\varepsilon\varepsilon^\top]w+w_x^\top\sup_{v_x}\bbE(v_xv_x^\top)w_x + cw_y^2+2w_x^\top\sup_{\bbE[v_y^2]\le c}\bbE(v_xv_y)w_y \\
	&=
	\begin{cases}
		\bbE\varepsilon_y^2+c, & b=b^\star\\
		\infty, & b\neq b^\star
	\end{cases}
\end{align*}
We thereby conclude the proof.
\end{proof}

\subsection{Proof of Proposition~\ref{corollary:dense}}
\label{proof_corollary_dense}
Recall that:
\begin{equation}
		b^{\mathrm{opt}}_\infty = b^\star + \left([C^\star{L^\star}{C^\star}^\top]_{1:p,1:p}\right)^{-1}(C^\star_xL^\star_{xy}+L^\star_{y}C_{xy}^\star).
\end{equation}
Let $$M = \left([C^\star{L^\star}{C^\star}^\top]_{1:p,1:p}\right)^{-1} = \big(C_x^\star{L}^\star_x {C_x^\star}^\top + C_x^\star{L}^\star_{xy}{C_{xy}^\star}^\top + C_{xy}^\star{{L}_{xy}^\star}^\top {C_{x}^\star}^\top + L_{y}^\star{C_{xy}^\star}{C_{xy}^\star}^\top\big)^{-1}.$$
Since the graph underlying the observed variables is a DAG according to Assumption~\ref{assumptiondag}, we have that $[I-B^\star]_{1:p,1:p}$ is an invertible matrix. Since the matrix $I-B^\star$ is also invertible, by Schur's complement, we have that the matrix $C^\star_x$ is an invertible matrix. Furthermore, we have the inequalities:
\begin{eqnarray}
\begin{aligned}
    \sigma_\text{min}(C^\star_x) &\geq \sigma_\text{min}(C^\star) \geq 1/\sigma_\text{max}(I-B^\star)\geq 1/(1+d\|B^\star\|_\infty) \geq 2/3,\\
    \sigma_\text{max}(C^\star_x) &\leq  \sigma_\text{max}(C^\star) \leq 1/\sigma_\text{min}(I-B^\star) \leq 1/(1-d\|B^\star\|_\infty) \leq 2,\\
    \sigma_\text{max}(C^\star_{xy}) &\leq  \sigma_\text{max}(C^\star) \leq 1/\sigma_\text{min}(I-B^\star)\leq 1/(1-d\|B^\star\|_\infty) \leq 2,
\end{aligned}
\label{eqn:intermed_dense}
\end{eqnarray}
where the last inequalities in each equation follow from the Assumption~\ref{assumption:bound_causal_coeffs} and the bound that for any matrix $N$, $\|N\|_2 \leq \|N\|_\infty{s}$, where $s$ is the maximum number of zeros in any column or row of $N$. We thus conclude that:
$$\sigma_\text{min}([C^\star{L^\star}{C^\star}^\top]_{1:p,1:p}) \geq \frac{\sigma_\text{min}(L^\star_x)}{\sigma_\text{max}(I-B^\star)^2} - \frac{2\sigma_\text{max}(L^\star_{xy})+L^\star_y}{\sigma_\text{min}(I-B^\star)^2} > \frac{\sigma_\text{min}(L^\star_x)}{2\sigma_\text{max}(I-B^\star)^2} = \mathcal{O}(\sigma_\text{min}(L^\star_x)),$$
where the second inequality follows from Assumption~\ref{assumption: sufficient_strong}. The equality follows from \eqref{eqn:intermed_dense}. By the definition of the matrix $M$, we have that $\|M\|_2 = \frac{1}{\mathcal{O}(\sigma_\text{min}(L^\star_x))}$. Furthermore, by Assumption~\ref{assumption:latent_exogenous}, notice that $$ L^\star_y = \sum_{e\in\mathcal{E}}\omega^e\left(E[(\delta_y^e)^2] + [\Gamma^\star\Sigma_{\eta^e}\Gamma^\star]_{p+1,p+1}\right),$$
where $\Sigma_{\eta^e}$ is the covariance of the perturbations on the latent variables. Therefore,
\begin{align*}
	|L^\star_y| &= \max_{e\in\mathcal{E}}E[(\delta_y^e)^2] + \|\Gamma^\star\|_2^2 \max_{i}\|\mathcal{P}_{\text{col-space}(\Gamma^\star)}e_i\|_2^2{h}^{3/2}\max_e\|\Sigma_{\eta^e}\|_\infty \\
	&= \max_{e\in\mathcal{E}}E[(\delta_y^e)^2] + \mathcal{O}\left(\frac{h^{5/2}\max_e\|\Sigma_{\eta^e}\|_\infty}{p}\right),
\end{align*}
where the last inequality follows from Assumptions~\ref{assumption:dense_latent} and \ref{assumption: bounded_latent}. Note that: 
$$L_{xy}^\star = \sum_{e\in\mathcal{E}}\omega^e\begin{pmatrix}\mathcal{I}_{p} & 0 \end{pmatrix}\Gamma^\star{\Sigma_{\eta^e}}{\Gamma^\star}^Te_{p+1}.$$
Similar as $L_{y}^\star$, we conclude:
$$\|L_{xy}^\star\|_\infty \leq \|\Gamma^\star\|_2^2 \max_{i}\|\mathcal{P}_{\text{col-space}(\Gamma^\star)}e_i\|_2^2{h}^{3/2}\max_e\|\Sigma_{\eta^e}\|_\infty = \mathcal{O}\left(\frac{h^{5/2}\max_e\|\Sigma_{\eta^e}\|_\infty}{p}\right)$$
We further have that:
\begin{eqnarray*}
\begin{aligned}
\|MC^\star_xL^\star_{xy}\|_\infty &\leq p\|MC^\star_x\|_2\|L^\star_{xy}\|_\infty =\mathcal{O}\left(\frac{p\|L^\star_{xy}\|_\infty}{\sigma_{\mathrm{min}}(L^\star_x)}\right),\\
\|M{C}^\star_{xy}L^\star_y\|_\infty &\leq \|M\|_2|L^\star_y|\|C^\star_{xy}\|_{2} =\mathcal{O}\left(\frac{|L_y^\star|}{\sigma_{\mathrm{min}}(L^\star_x)}\right).
\end{aligned}
\end{eqnarray*}
Putting everything together, we have the desired bound.


\subsection{Proof of Proposition~\ref{prop:bias_bound3}}
\label{proof:bias_bound3}
\begin{proof}
As $L^\star_{xy}=0$ and $L^\star_y=0$, we have 
\begin{equation*}
    \cL_{\mathrm{reg}}(b) = w_x^\top L^\star_xw_x = (b^\star-b)^\top \Delta_x (b^\star-b)
\end{equation*}
which is minimized whenever $\Delta_x (b^\star-b)=0$. This immediately leads to $\cI=\{b^\star+b':\Delta_x b'=0\}$.

When $L^\star_{xy}=0$ and $L^\star_y=0$, the original objective function given any $\gamma$ becomes
\begin{equation*}
	\cL(b) = \bbE[(Y^0-b^\top X^0)^2] + \gamma (b^\star-b)^\top \Delta_x (b^\star-b),
\end{equation*}
where the first term is equal to $\bbE[(\varepsilon_y-(b-b^\star)^\top X^0)^2]$. Minimizing $\cL(b)$ leads to
\begin{equation*}
	b^{\mathrm{opt}}_\gamma-b^\star = [\bbE X^0{X^0}^\top+\gamma\Delta_x]^{-1}\bbE[X^0\varepsilon_y].
\end{equation*}
Letting $\gamma\to\infty$ leads to 
\begin{equation*}
		b^{\mathrm{opt}}_\infty=b^\star+D\bbE[X^0\varepsilon_Y],
	\end{equation*}
	where $D=\lim_{\gamma\to\infty}\big[\bbE X^0{X^0}^\top+\gamma\Delta_x\big]^{-1}$.

Also notice that $\bbE[X^0\varepsilon_y]=\bbE[(C^\star_x\varepsilon_x+C^\star_{xy}\varepsilon_y)\varepsilon_y]=C^\star_x\bbE[\varepsilon_x\varepsilon_y]+C^\star_{xy}\bbE\varepsilon^2_y$. Then we have
\begin{align*}
	\|b^{\mathrm{opt}}_\infty-b^\star\|_\infty &= \|D\bbE[X^0\varepsilon_Y]\|_\infty\\
	&\leq \|D\|_\infty\|C^\star_x\bbE[\varepsilon_x\varepsilon_y]+C^\star_{xy}\bbE\varepsilon^2_y\|_\infty,
\end{align*}
which concludes the proof. 
\end{proof}

\subsection{Proof of the results in the specialized setting}\label{proof_prop_insuff_interv_ours_ols}
Let $\bbE [X^0{X^0}^\top] = \mathrm{diag}(\sigma^2_1,\dots,\sigma^2_p)$ and $\Delta_x = \mathrm{diag}(\Delta_1,\dots,\Delta_p)$. 
From the proof of Proposition~\ref{prop:bias_bound3} we know
\begin{equation*}
	b^\text{opt}_\gamma-b^\star = [\bbE X^0{X^0}^\top+\gamma\Delta_x]^{-1}\bbE[X^0\varepsilon_y] 
\end{equation*}
whose $i$th component is $\bbE[X^0_i\varepsilon_y]/(\sigma^2_i+\gamma\Delta_i)$, where $X^0_i$ is the $i$th component of $X^0$. The OLS estimator on the observational environment satisfies
\begin{equation*}
	b^0_{\mathrm{OLS}}-b^\star=	[\bbE X^0{X^0}^\top]^{-1}\bbE[X^0\varepsilon_y] 
\end{equation*}
whose $i$th component is $\bbE[X^0_i\varepsilon_y]/\sigma^2_i$.
Since $\gamma\Delta_i\geq0$, we immediately know that $\|b^\text{opt}_\gamma-b^\star\|_2\leq \|b^0_{\mathrm{OLS}}-b^\star\|$. When $\bbE[X^0_i\varepsilon_y]>0$ and $\Delta_i>0$, we have $|\bbE[X^0_i\varepsilon_y]|/(\sigma^2_i+\gamma\Delta_i)<|\bbE[X^0_i\varepsilon_y]|/\sigma^2_i$ and thus the inequality is strict.

\section{Additional empirical results}\label{app:emp_results}

\subsection{Illustrative examples}

In Section~\ref{sec:exam}, we present two illustrative examples to demonstrate the advantages of DRIG in robust prediction. Here, we provide an additional example, where the training data contains a limited amount of heterogeneity in the mean. Specifically, in Example~\ref{ex:robust1} (a covariate-intervened setting), we now set $\bbE[\delta_x^1]=0.1$, that is, there are limited mean shifts in $X$; in Example~\ref{ex:robust2} (an all-intervened setting), we now set 
$$
	\begin{pmatrix}
		\delta_x^1 \\ \delta_y^1
	\end{pmatrix}
	\sim\cN\left(
	\begin{pmatrix}
		0.1\\0.1
	\end{pmatrix},
	\begin{pmatrix}
		1 & 0.1\\
		0.1 & 0.05\\
	\end{pmatrix}
	\right)
$$
where we only change $\bbE[\delta_x^1]$ from 0.5 to 0.1 so the amount of mean shifts is again limited. As shown in Figure~\ref{fig:illus_small_meanshifts}, anchor regression that can only exploit mean shifts performs very close to the pooled OLS. In contrast, DRIG maintains competitive robustness performance. 

\begin{figure}
\centering
\begin{tabular}{cc}
    \includegraphics[width=0.4\textwidth]{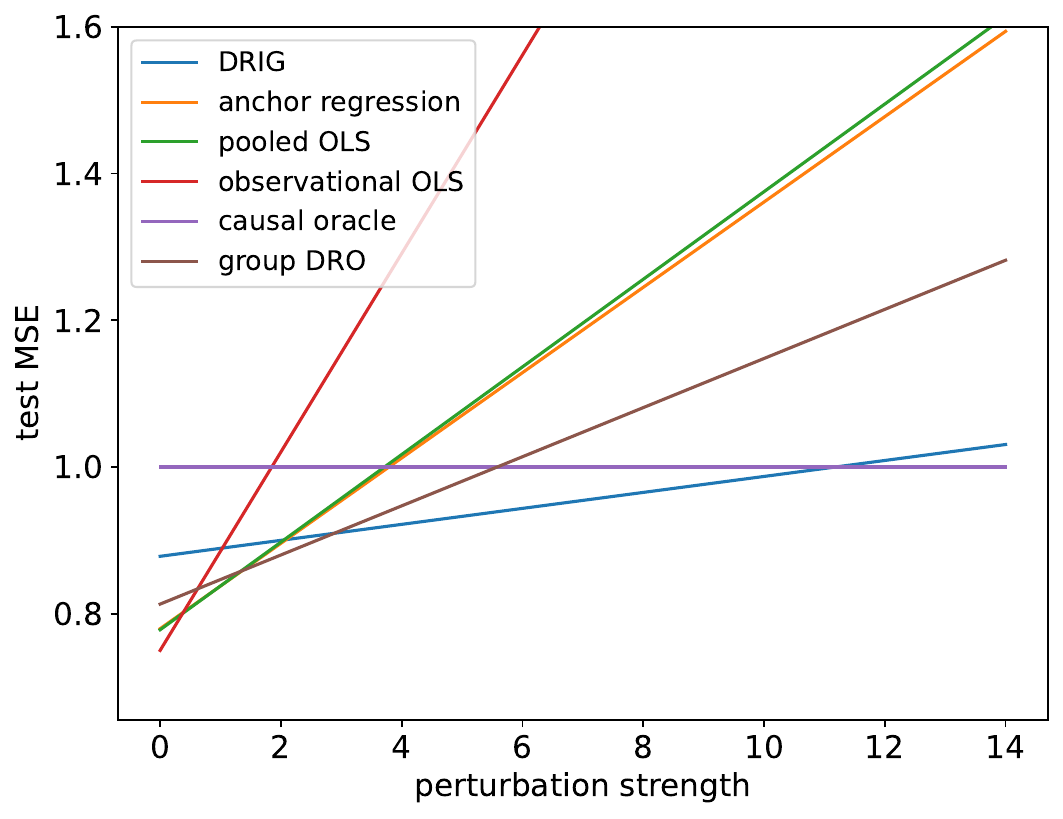} &
    \includegraphics[width=0.4\textwidth]{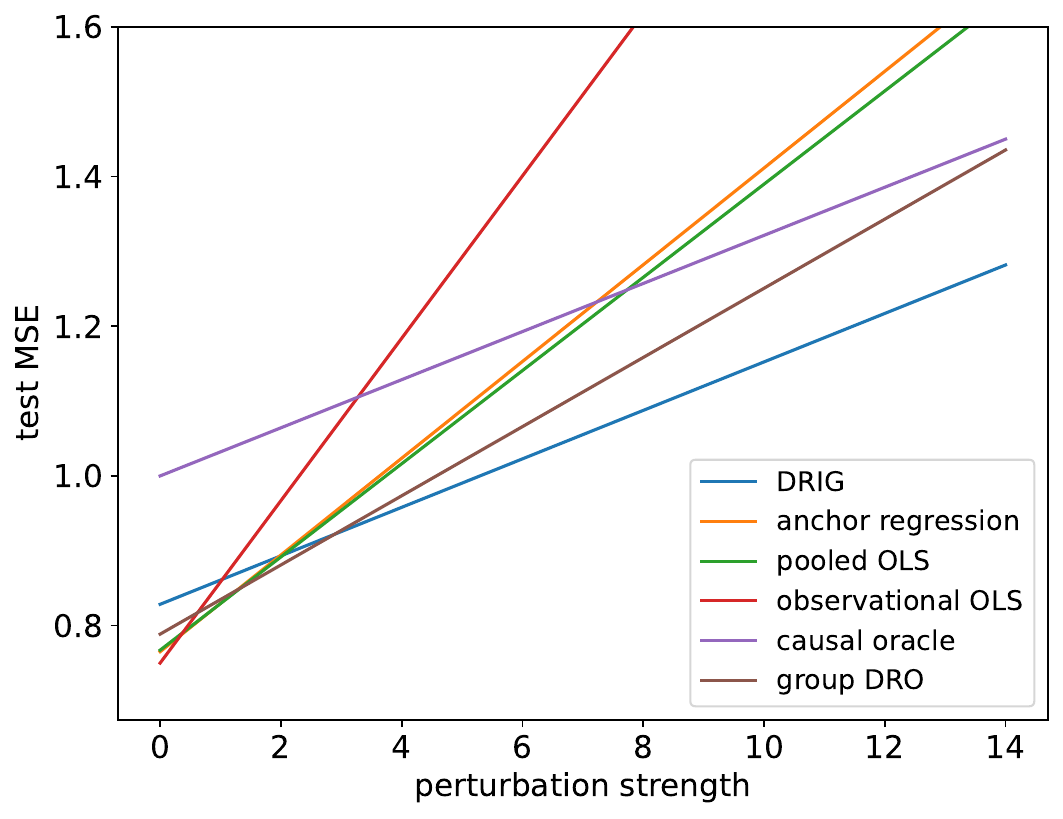}\\
    (a) covariate-intervened & (b) all-intervened
\end{tabular}
\caption{Illustrative examples with small mean shifts.}\label{fig:illus_small_meanshifts}
\end{figure}

\subsection{Synthetic simulations}\label{sec:simulations}
We next provide additional synthetic simulations to compare the robustness performance of DRIG and DRIG-A+ with competing methods. We consider a setting with $p=10$ covariates and a response variable and simulate $10^4$ observational data according to the linear SCM in \eqref{eq:scm_train}, where $B^\star$ is a randomly generated Erdos-Renyi directed acyclic graph and $\varepsilon\sim\cN(0,S^0)$ with $S^0$ being a randomly sampled positive definite matrix. Details of the sampling scheme are given in Supplementary~\ref{app:exp_detail}. We also simulate $10^4$ interventional data each from three environments, governed by SCMs \eqref{eq:scm_train}, where $\delta^e\sim\cN(\mu^e,S^e)$, $e=1,2,3$. Finally, we generate 20 test environments according to SCM \eqref{eq:scm_test}, where $B^\star$ is the same as the training SCM, while the intervention variables in the test environment are generated according to $v_j\sim\cN(\sqrt{\alpha}\mu_j,\alpha{S}_j),j=1\dots,20$ where the scalar $\alpha >0$ controls the perturbation strengths in the test environment. We consider the following two scenarios within the setting described above:
\begin{enumerate}
\item covariate-intervened case with interventions on $X$ but no intervention on $Y$ or $H$: here, we set the last entry of $(\mu^e,\mu_j), e = 1,2,3, j = 1,2, \dots, 20$ and the last row and column of $(S^e,S_j), e = 1,2,3, j = 1,2,\dots,20$ to zero, and choose the remaining components at random.
\item all-intervened case with interventions on all of $X$, $Y$, and $H$: the vectors $(\mu^e,\mu_j)$ and the matrices $(S^e,S_j)$ are chosen at random for every $e = 1,2,3$ and $j = 1,2,\dots,20$.
\end{enumerate}
Given a training data distribution, we repeat the process of drawing training samples, as described above, for 50 times and report the average performance.

We apply our proposed methods as well as existing approaches on the training data to obtain linear prediction models. We then compute the population MSE of each estimated model on each of the $20$ test environments, and report the worst-case error among all of the environments. For DRIG and anchor regression, we consider three schemes for choosing the regularization strength $\gamma$: a fixed $\gamma=10$, an oracle choice of $\gamma$ for each test environment that gives the smallest MSE on that environment, and our proposed DRIG-A+ that chooses a matrix $\Gamma$ for each test environment by exploiting a small test sample of size $50$ from that environment. For DRIG and DRIG-A+, we assign uniform weights to each environment, i.e., $\omega^e=1/4$.

Figures~\ref{fig:simu_no_interv_y}-\ref{fig:simu_interv_y} present the worst-case test MSEs for varying perturbation strengths $\alpha$ in the test distributions, where we plot the mean of the worst-case errors over the 50 random repetitions with the 95\% bootstrapped confidence intervals. Overall, DRIG estimators tend to be the most competitive method. With either a fixed or the oracle choice of $\gamma$, DRIG exhibits better performance than anchor regression with the same scheme of choosing $\gamma$. Anchor regression, while better than the OLS estimators, offers limited advantages compared to DRIG. This suggests that DRIG achieves better distributional robustness, potentially due to its ability to exploit heterogeneity in the variances. 

Regarding the selection of hyperparameter $\gamma$, DRIG with a fixed $\gamma>1$ can already yield satisfying robust performance compared to baseline approaches, especially in the causal-identifiable case, while the oracle choice further enhance the advantage. As shown in panel (b) in both figures, the oracle $\gamma$ monotonically increases with respect to the perturbation strength, which aligns with the earlier message that a larger $\gamma$ enhances robustness against stronger perturbations. More interestingly, our DRIG-A+ that leverages additional test information consistently stands out as the best-performing method due to its more flexible and adaptive regularization scheme. These observations suggest that in practice a fixed $\gamma>1$ could already lead to reasonably well robustness compared to OLS; when a small number of samples from the test distribution is available, we further improve the robustness performance by DRIG-A+.

The causal parameter exhibits invariant performance regardless of the perturbation strength in the covariate-intervened case, but performs significantly worse than the other methods when all variables are intervened on. 



\begin{figure}
\centering
\begin{tabular}{@{}c@{}c@{}}
	\includegraphics[width=0.5\textwidth]{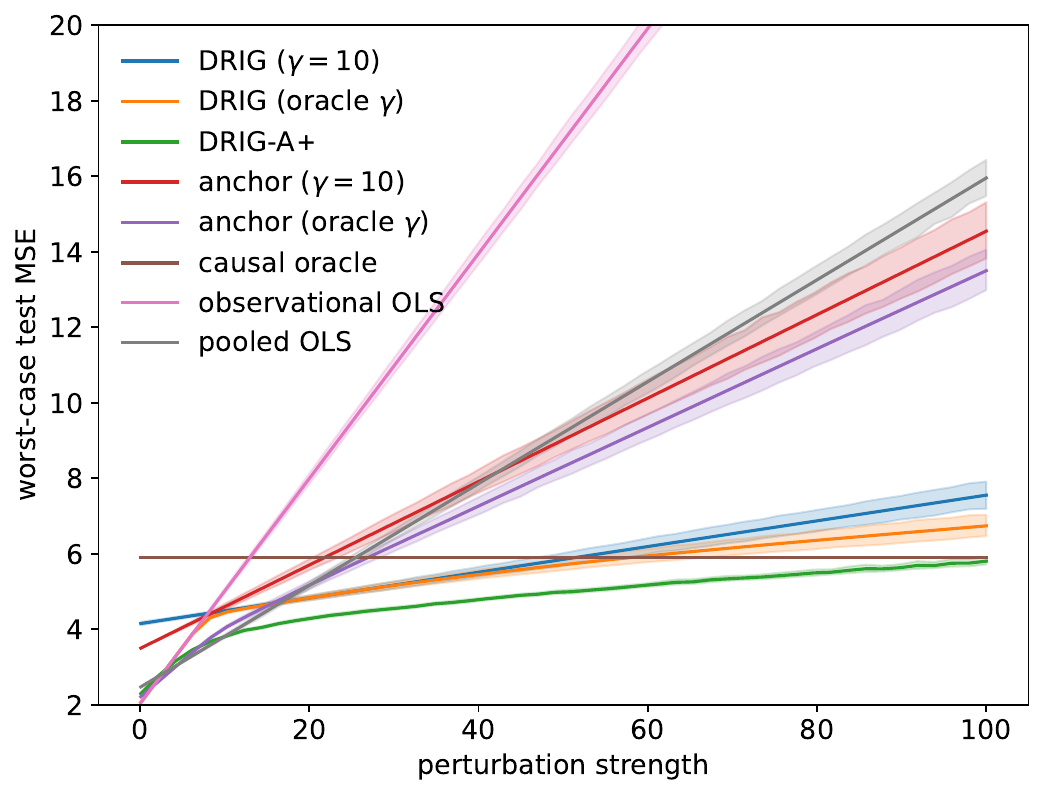} & \includegraphics[width=0.38\textwidth]{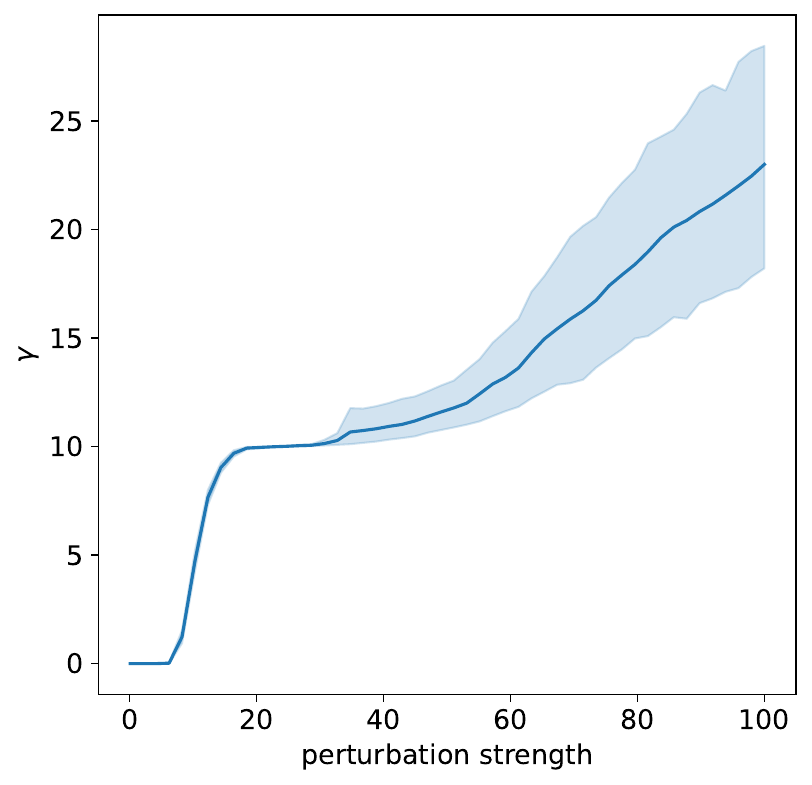}\\
	\small{(a) worst-case test MSE} & \small{(b) oracle $\gamma$ for DRIG}
\end{tabular}
\caption{(a) the worst-case test MSEs for varying perturbation strengths in the covariate-intervened case; (b) the oracle $\gamma$ for DRIG for different perturbation strengths. {Lines represent the means and 2.5\% and 97.5\% qunatiles.}}
\label{fig:simu_no_interv_y}
\end{figure}

\begin{figure}
\centering
\begin{tabular}{@{}c@{}c@{}}
	\includegraphics[width=0.5\textwidth]{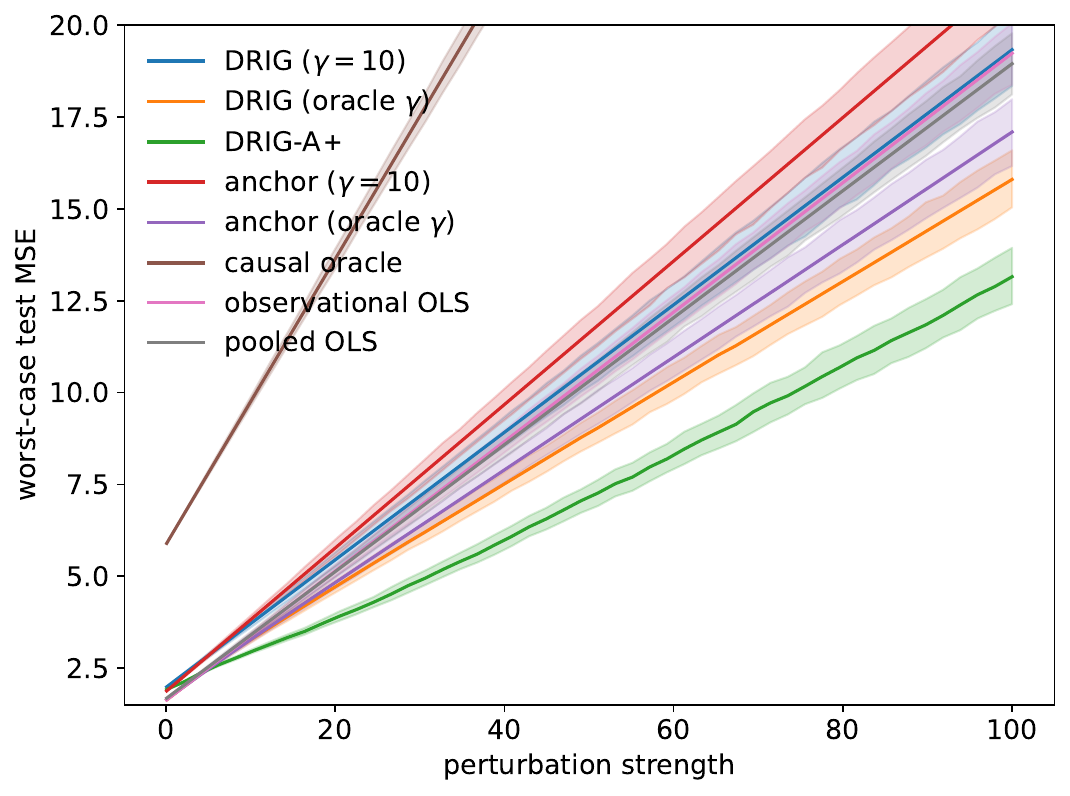} & \includegraphics[width=0.38\textwidth]{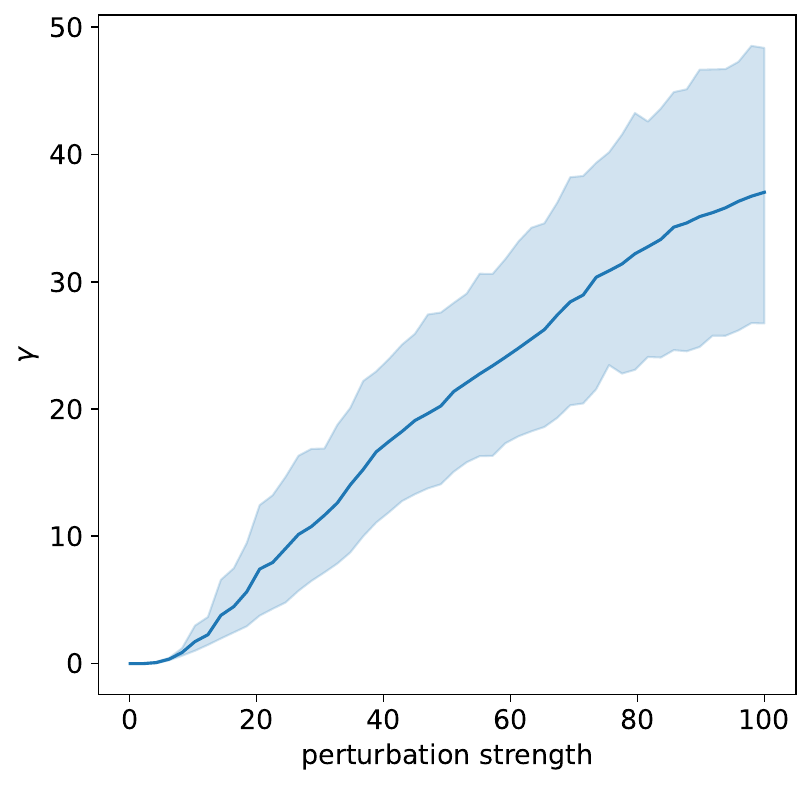}\\
	\small{(a) worst-case test MSE} & \small{(b) oracle $\gamma$ for DRIG}
\end{tabular}
\caption{Same plots {as in Figure~\ref{fig:simu_no_interv_y}} for the all-intervened case. }
\label{fig:simu_interv_y}
\end{figure}


\subsection{Illustrations for DRIG-A+}
\label{sec:illustration_DRIG_A+}

\begin{example}\label{ex:adap}
We set $p=20$ and two training environments $e=0,1$ with a randomly sampled mean vector $\mu^1$ and covariance matrices $S^0$ and $S^1$, where the last rows, or columns are zeros, indicating no interventions on $Y$. Details of the sampling scheme are given in Appendix~\ref{app:exp_detail}. Consider a test distribution following  SCM \eqref{eq:scm_test} with $\bbE[vv^\top] = \alpha G^v$, where $G^v$ is a randomly sampled positive definite matrix whose last row and column are zeros. We assume a small labeled test sample of size $n_l = 50$ and population of $X^v$ (i.e., $n_u\to\infty$). 
\end{example}

In Figure~\ref{fig:adap}(b), we plot the test MSEs of various methods including the methods that make use of the test samples (DRIG-A+, test OLS and the population versions of them), the baseline approaches that only use the training data, and the oracle causal parameter. Compared to the test OLS estimator, DRIG-A+ consistently yields much smaller test MSEs, which is aligned with Theorem~\ref{thm:adap_mse}. Furthermore, compared to the other methods that do not leverage the test data, DRIG-A+ has better predictive performance; we show in Appendix~\ref{sec:simulations} that DRIG-A+ remains superior even if an oracle choice of $\gamma$ that minimizes test MSE is used in anchor regression and DRIG. Finally, the causal parameter, while invariant across all test perturbations, is overly conservative under moderate and weak perturbations. 

\begin{figure}[h!]
\centering
\includegraphics[width=0.4\textwidth]{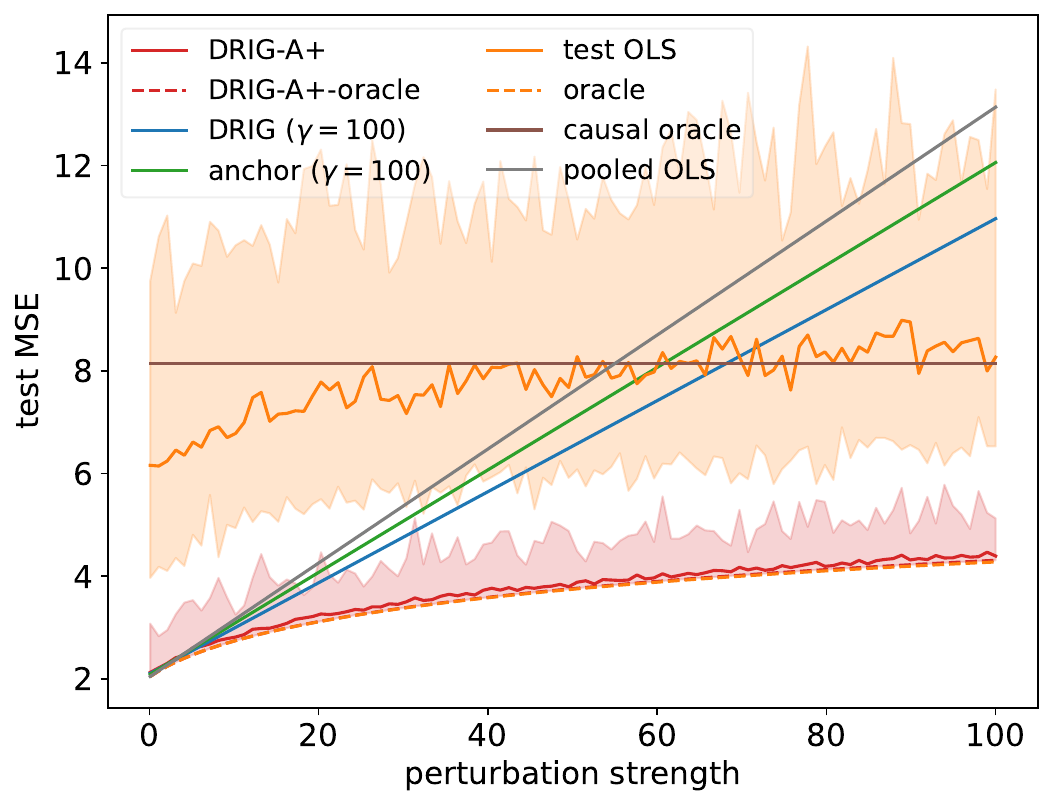} \vspace{-0.2in}
\caption{\small Test MSE for varying perturbation strengths $\alpha$ of various methods.``DRIG-A-oracle" and ``oracle" refer to ${b}_{\hat{\Gamma}}^\mathrm{opt}$ and $\hat{b}_{\mathrm{tOLS}}$ as $n_u,n_l\to\infty$, respectively. 
For the DRIG-A+ and test-OLS estimators, we randomly draw a labeled sample size of $n_l = 50$ from the test distribution. The DRIG-A+ estimator is obtained based on the labeled test sample and training population, and test-OLS is obtained from the labeled test sample. We repeat this procedure for 50 times and show the median test MSE along with the 2.5\% and 97.5\% quantiles.}\label{fig:adap}
\end{figure}

\subsection{Single-cell data}\label{app:singlecell}
Figure~\ref{fig:variances} shows the variances of all observed variables in each environment, shedding light on the heterogeneity of gene expression across different interventions. We observe that the last variable is the only one that consistently exhibits a higher variance in interventional environments than in the observational environment. Also, when intervening on the last variable, we barely see increases in the variances of the other variables. This observation roughly suggests that interventions on the last gene have limited impact on the variability of the other genes, supporting the conjecture that the last gene may act as a leaf node in the causal graph among the 10 observed genes. Based on this reasoning, we select the last gene as our response variable and consider the remaining {9} genes as covariates. 

Next, we investigate how the methods perform differently on test environments generated by some specific interventions. In Figure~\ref{fig:sc_specific_test_env}, we show the MSEs on several test environments with different patterns of interventions. In the first row of the figure, we observe that the MSE decreases as $\gamma$ increases, which, according to our theory, suggests that these interventions are relatively strong. In the middle row, the MSE initially decreases and then increases with increasing $\gamma$, indicating a moderate perturbation strength. In the bottom row,  we observe that the MSE consistently grows with $\gamma>0$, which suggests that these environments are likely to be close to the observational environment.

\begin{figure}
\centering
\includegraphics[width=0.4\textwidth]{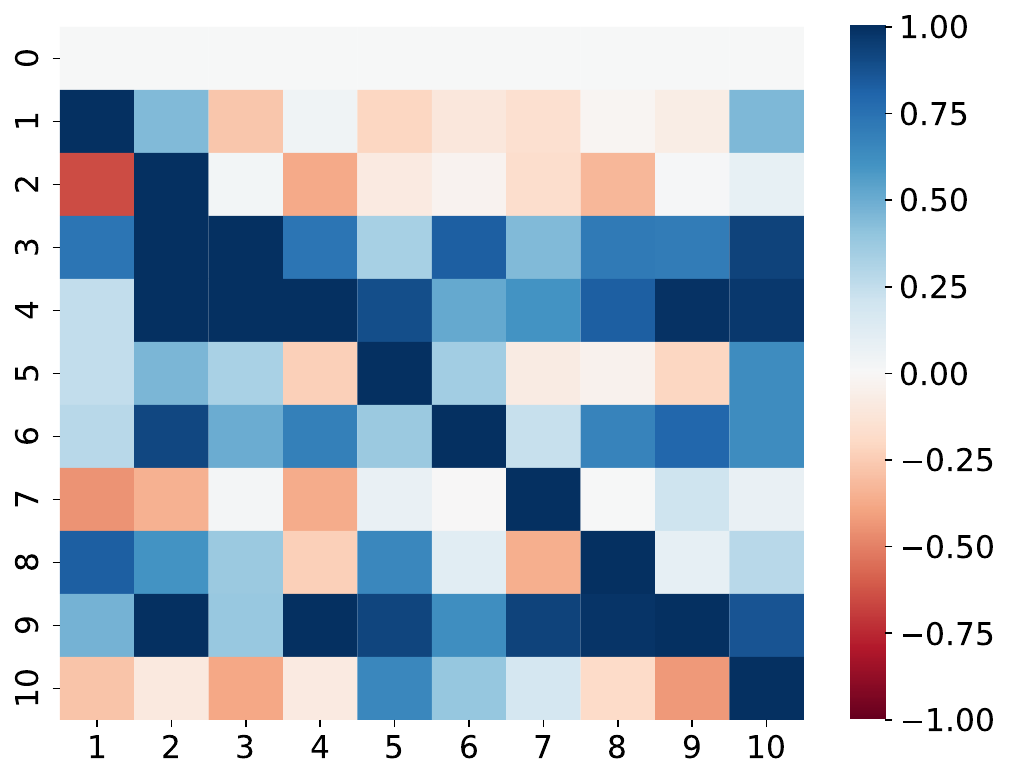}
\caption{Variances of each observed variable in each environment. For the purpose of illustration, we take a transformation $h(v_i)=\mathrm{tanh}(2(v_i/v_0-1))$, where $v$ is the variance of a variable on the $i$-th environment and $v_0$ is its variance on the observational environment.}\label{fig:variances}
\end{figure}

\begin{figure}
\centering
\includegraphics[width=0.8\textwidth]{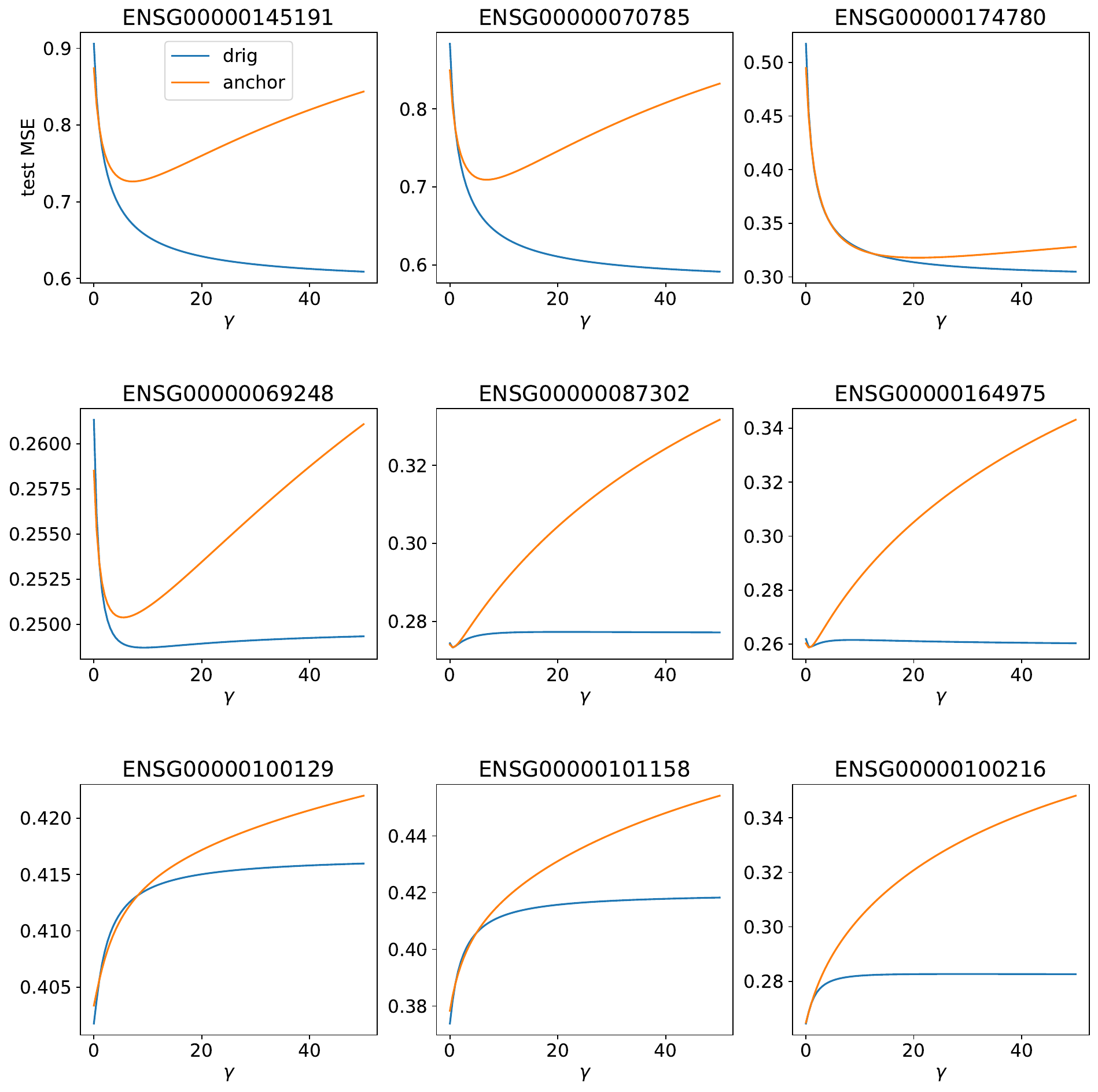}
\caption{Performance {of DRIG} on several specific test environments with different patterns of interventions.}\label{fig:sc_specific_test_env}
\end{figure}

\subsection{Optimization}
\label{sec:optimization}
In all our numerical experiments, we use the Adam optimizer with a learning rate of $10^{-3}$ and train each model for 10k iterations. 
We show some numerical examples for optimization. In the settings with an observational environment, DRIG has a closed form solution. We hence check the convergence of the gradient descent algorithm to the analytical optimal solution. In Figure~\ref{fig:opt}, we plot the convergence curve of the loss $\cL(b)$ in \eqref{eqn:DRIG} and the bias $\|b-b^{\rm opt}\|$ between $b$ at each iteration and the global optimizer $b^{\rm opt}$ using the closed form solution. 
\begin{figure}
	\centering
	\begin{tabular}{@{}c@{}c@{}}
		\includegraphics[width=0.5\textwidth]{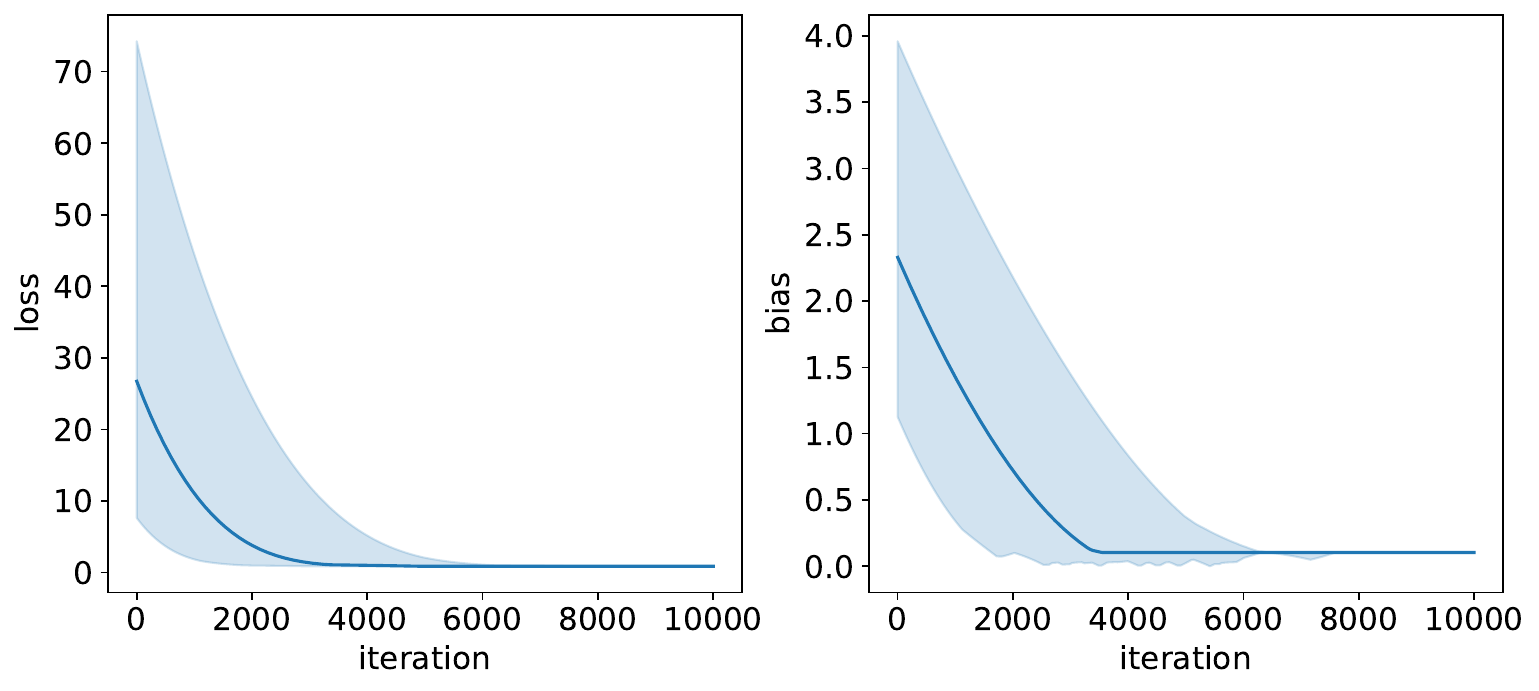} & 
		\includegraphics[width=0.5\textwidth]{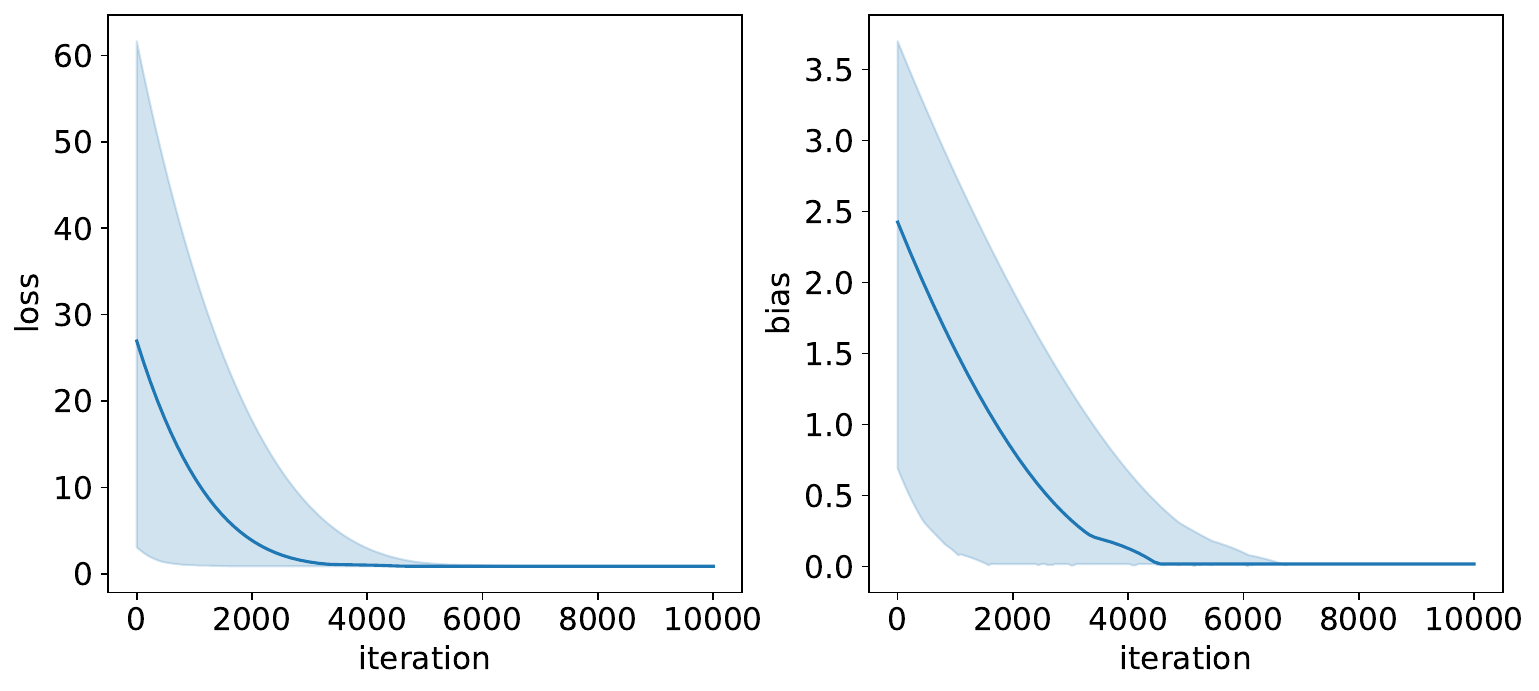}\vspace{-0.2in}\\
		\small{Example~\ref{ex:robust1}} & \small{Example~\ref{ex:robust2}}
	\end{tabular}
	\caption{Convergence curve of the loss and absolute bias with respect to training iterations from 20 random initializations for each example.}\label{fig:opt}
\end{figure}

\section{Experimental details}\label{app:exp_detail}
\subsection{Simulations}
We describe how we sample the mean vectors and covariance matrices for the noise $\varepsilon$ and the intervention variables $\delta^e$ in Examples \ref{ex:adap} and simulations in Section~\ref{sec:simulations}. We sample the components of the mean vectors independently form $\mathrm{Unif}[0,1]$. For the covariance matrices, we first sample a random matrix $\tilde{S}$ whose components are independently drawn form $\mathrm{Unif}[0,1]$ and then get the covariance matrix by $\tilde{S}\tilde{S}^\top$. To explicitly control the perturbation strength, we normalize the means and covariances of the interventions variables to always have vector or matrix 2-norm 1.  If $Y$ and $H$ are assumed not to be intervened on, we set the last component of the mean vectors and the last row and column of the covariance matrices to zero. For simulations in Section~\ref{sec:simulations}, we sample the mean vectors and covariance matrices of all interventions variables $\delta^e,e=1,2,3$ in training environments as well as $v$ in test environments. To ensure there is sufficient amount of heterogeneity among training environments, we multiply the mean vectors of $\delta^e$ by a factor of $\sqrt{10}$ and multiply the covariance matrices by a factor of $10$. Note that during test, we vary the perturbation strength from 1 to 100, as shown in Figures~\ref{fig:simu_no_interv_y}-\ref{fig:simu_interv_y}. Thus, the perturbation strength during test exceeds substantially that during training, resulting in a challenging task for robustness. 

\subsection{ICU data}
We select covariates with less than 10\% observations missing, which leads to 17 variables: blood urea nitrogen (bun), calcium (ca), chloride (cl), creatinine (crea), glucose (glu), hemoglobin (hgb), heart rate (hr), potassium (k), mean arterial pressure (map), sodium (na),
       oxygen saturation (o2sat), respiratory rate  (resp), white blood cell count (wbc), age, sex, height, and weight. For the 14 variables among them with missing data, we impute the missing entries them with a constant (zero) and add a binary indicator for the missingness. Then we use all 31 variables as covariates to predict the outcome. 
eICU dataset consists of four regions: midwest, south, west, and northeast, which are used as four training environments. 